\documentclass[%
paper=A4,%
pagesize,%
oneside,%
bibliography=totoc,%
]{scrartcl}

\usepackage[english]{babel}

\usepackage{amsmath,amssymb,amsthm,amsfonts,aliascnt}
\usepackage{thmtools} % Fixes \autoref{} for shared counters
\usepackage{mathtools}
\usepackage{exscale}

\usepackage{bbm}
\usepackage{mathrsfs} %mathscr{}
\usepackage[T1]{fontenc}
\usepackage{lmodern}

\usepackage{graphicx}

\DeclareFontFamily{U}{bigshuffle}{}
\DeclareFontShape{U}{bigshuffle}{m}{n}{
  <5-8> s*[1.7] shuffle7
  <8->  s*[1.7] shuffle10
}{}
\DeclareSymbolFont{BigShuffle}{U}{bigshuffle}{m}{n}
\DeclareMathSymbol\bigshuffle{\mathop}{BigShuffle}{"001}

\usepackage{shuffle}

\usepackage{ifthen}

\usepackage{xcolor}

\usepackage{tabularx,booktabs}

\definecolor{links}{rgb}{0,0.3,0}
\definecolor{mylinks}{rgb}{0.8,0.2,0}
\usepackage[colorlinks=true,citecolor=blue,urlcolor=links,breaklinks=true,linkcolor=mylinks]{hyperref}

\newcommand{\orcidurl}[1]{https://orcid.org/#1}
\newcommand{\orcidicon}[1]{\href{\orcidurl{#1}}{\protect\includegraphics[height=1.5ex]{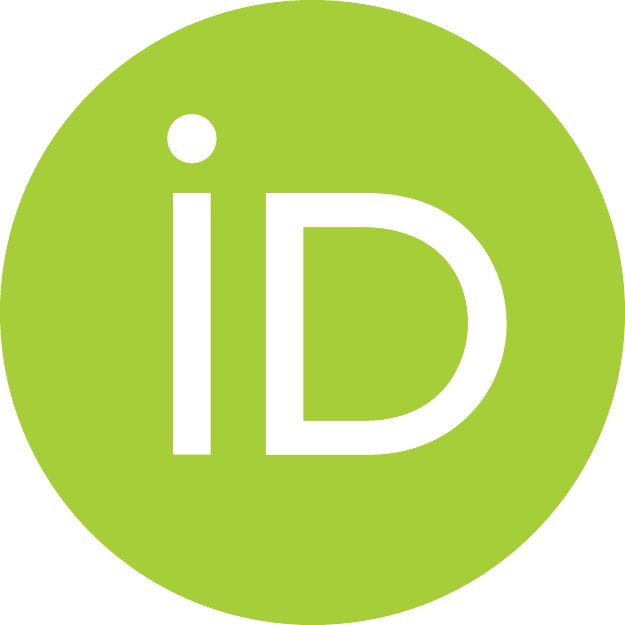}}}

\newcommand{\email}[1]{\href{mailto:#1}{#1}}

\numberwithin{equation}{section}

\usepackage[numbers,sort&compress]{natbib}
\newcommand{\Q}{\mathbbm{Q}}

\newcommand{\Z}{\mathbbm{Z}}
\newcommand{\R}{\mathbbm{R}}
\newcommand{\C}{\mathbbm{C}}
\newcommand{\FF}[1]{\mathbbm{F}_{#1}}
\newcommand{\RP}{\mathbbm{RP}}
\newcommand{\Projective}{\mathbbm{P}}

\newcommand{\Tens}[1]{\left\langle #1 \right\rangle}

\newlength{\wurelwidth}

\newcommand{\wurel}[2][=]{\mathrel{\mathop{#1}\limits_{\!\scalebox{0.5}{\makebox[\the\wurelwidth]{#2}}\!}}}

\newcommand{\td}[1][]{\mathrm{d}^{#1}}

\newcommand{\Char}{\chi}
\newcommand{\dChar}{\check{\chi}}

\newcommand{\TwoSum}[2]{\mathop{_{#1}\oplus_{#2}}}
\newcommand{\restrict}[2]{%
{\left. #1 \right|}_{#2}%
}

\DeclareMathOperator*{\Res}{Res}

\DeclareMathOperator{\Aut}{Aut}

\DeclareMathOperator{\RePart}{Re}

\newcommand{\defas}{\mathrel{\mathop:}=}

\newcommand{\set}[1]{
\left\{ #1 \right\}
}
\newcommand{\setexp}[2]{
\left\{ #1\!:\ #2 \right\}
}

\newcommand{\abs}[1]{
\left\lvert #1 \right\rvert
}
\newcommand{\norm}[1]{
\left\lVert #1 \right\rVert
}

\newcommand{\Dim}{d}

\newcommand{\dual}{\star}

\newcommand{\cOne}[1]{c_{#1}}

\newcommand{\Derksen}[1]{\mathcal{G}\left( #1 \right)}

\newcommand{\Crapo}[1]{\beta\left(#1\right)}
\newcommand{\Period}[1]{\mathcal{P}\left( #1 \right)}
\newcommand{\tPeriod}[1]{\mathcal{P}( #1 )}
\newcommand{\Hepp}[2][]{\mathcal{H}\ifthenelse{\equal{#1}{\empty}}{}{_{#1}}\ifthenelse{\equal{#2}{\empty}}{}{\!\left( #2 \right)}}
\newcommand{\HeppFlat}[2][]{\mathcal{H}^{\Flat}\ifthenelse{\equal{#1}{\empty}}{}{_{#1}}\ifthenelse{\equal{#2}{\empty}}{}{\!\left( #2 \right)}}
\newcommand{\HeppComp}[1]{\widehat{\mathcal{H}}\ifthenelse{\equal{#1}{\empty}}{}{\!\left( #1 \right)}}
\newcommand{\PeriodComp}[1]{\widehat{\mathcal{P}}\ifthenelse{\equal{#1}{\empty}}{}{\!\left( #1 \right)}}
\newcommand{\HeppOne}[1]{\mathcal{H}_{\text{\upshape{den}}}\left( #1 \right)}
\newcommand{\HeppOneFlat}[1]{\mathcal{H}^{\Flat}_{\text{\upshape{den}}}\left( #1 \right)}
\newcommand{\HeppOneLoop}[1]{\mathcal{H}_{\log}\left( #1 \right)}

\newcommand{\injects}{\hookrightarrow}

\newcommand{\STrees}[1]{\mathcal{T}_{#1}}
\newcommand{\Flags}[2][]{\mathcal{F}\ifthenelse{\equal{#1}{\empty}}{}{^{#1}}_{#2}}
\newcommand{\Lat}[2][]{\mathcal{L}\ifthenelse{\equal{#1}{\empty}}{}{^{#1}}_{#2}}
\newcommand{\Sing}[1]{\mathcal{S}_{#1}} % singular subgraphs (i.e. those corresponding to facets)
\newcommand{\CycFlats}[1]{\mathcal{Z}_{#1}}
\newcommand{\Flats}[1]{\Lat[\Flat]{#1}}
\newcommand{\Cycs}[1]{\Lat[\OnePI]{#1}}

\newcommand{\UM}[2]{U_{#1}^{#2}}

\newcommand{\OnePI}{\ensuremath{\mathrm{1pi}}}
\newcommand{\Flat}{\mathrm{flat}}
\DeclareMathOperator{\Span}{span}
\newcommand{\closeM}[1]{\Span(#1)}
\DeclareMathOperator{\Cyc}{cyc}
\newcommand{\intM}[1]{\Cyc(#1)}
\DeclareMathOperator{\br}{br}
\newcommand{\Bridges}[1]{\br(#1)}

\newcommand{\CycFlatFactor}[1]{\rho\left( #1 \right)}
\newcommand{\CycFlatFactorDual}[1]{\tilde{\rho}\left( #1 \right)}

\newcommand{\PlogDiv}{\ensuremath{\text{p-log}}}

\newcommand{\perms}[1]{\mathfrak{S}_{#1}}
\newcommand{\HeppSec}[2][]{\mathcal{D}_{#2}\ifthenelse{\equal{#1}{\empty}}{}{^{#1}}}

\newcommand{\TreePath}[2]{P_{#1}^{#2}}
\newcommand{\field}{\phi}

\newcommand{\PsiPol}{\Psi}
\newcommand{\trop}{\mathrm{tr}}
\newcommand{\PsiTrop}{\PsiPol^{\trop}}
\newcommand{\PhiTrop}{\Phi^{\trop}}

\newcommand{\Newton}[1]{\mathcal{N}_{#1}}
\newcommand{\Polar}[1]{\mathcal{N}_{#1}^{\circ}}

\newcommand{\GMat}[1]{\mathcal{M}( #1 )} % cycle matroid of a graph
\newcommand{\GGMat}[1]{\mathcal{M}\left( #1 \right)} % cycle matroid of a graph
\newcommand{\DMat}[1]{{#1}^{\dual}} % dual matroid
\newcommand{\Bases}[1]{\mathcal{B}_{#1}} % bases of a matroid
\newcommand{\Circuits}[1]{\mathcal{C}_{#1}} % circuits of a matroid
\newcommand{\Indeps}[1]{\mathcal{I}_{#1}}

\newcommand{\exc}[1]{\delta_{#1}}

\newcommand{\sdc}[2][]{\omega%
\ifthenelse{\equal{#1}{}}{}{_{#1}}%
\ifthenelse{\equal{#2}{}}{}{( #2 )}}
\newcommand{\Dsdc}[2][]{\Delta^{#2}\ifthenelse{\equal{#1}{\empty}}{}{_{#1}}\sdc{}}

\newcommand{\ConvDom}{\Theta}

\newcommand{\asyO}[1]{\mathcal{O}( #1 )}

\DeclareMathOperator{\rk}{rk}
\newcommand{\rank}[1]{\rk(#1)}
\newcommand{\loops}[1]{\ell(#1)}
\newcommand{\loopsIn}[2]{\ell_{#1}(#2)}
\newcommand{\VG}[1]{V_{#1}}

\newcommand{\EG}[1]{E_{#1}}
\newcommand{\nCG}[1]{\kappa\ifthenelse{\equal{#1}{\empty}}{}{(#1)}}
\newcommand{\TriG}{\Graph[0.16]{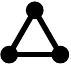}}
\newcommand{\StarG}{\Graph[0.16]{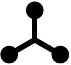}}

\newcommand{\Graph}[2][1.0]{%
\vcenter{\hbox{\includegraphics[scale=#1]{graphs/#2}}}%
}
\newcommand{\gKite}{\Graph[0.3]{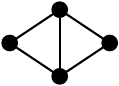}}
\newcommand{\gBox}{\Graph[0.25]{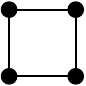}}
\newcommand{\gTri}{\Graph[0.27]{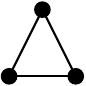}}
\newcommand{\gEdge}{\Graph[0.3]{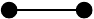}}
\newcommand{\gKfour}{\Graph[0.14]{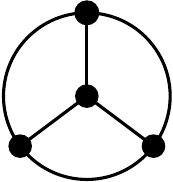}}

\newcommand{\RadInt}[1]{\rho\ifthenelse{\equal{#1}{\empty}}{}{(#1)}}

\DeclareMathOperator{\pr}{pr}

\newcommand{\SP}{x}
\newcommand{\SPlog}{y}
\newcommand{\ind}{a}

\newcommand{\WS}[1]{W_{#1}}
\newcommand{\Fan}[1]{F_{#1}}
\newcommand{\Cycle}[1]{C_{#1}}
\newcommand{\Bond}[1]{D_{#1}}
\newcommand{\Path}[1]{P_{#1}}

\newcommand{\gf}[1]{\mathscr{#1}}

\newcommand{\uv}[1]{\vec{\mathsf{e}}_{#1}}

\newcommand{\vt}[1]{\vec{\mathsf{v}}_{#1}}

\newcommand{\mzv}[2][]{\zeta^{#1}( #2 )}

\newcommand{\apode}{\mathsf{S}}
\newcommand{\buch}[1]{\langle #1 \rangle}
\DeclareMathOperator{\conv}{conv}
\DeclareMathOperator{\Vol}{Vol}

\newcommand{\prO}[1]{\pr^{\bot}_{#1}}

\theoremstyle{plain}
\newtheorem{theorem}{Theorem}[section]
\newtheorem{lemma}[theorem]{Lemma}
\newtheorem{corollary}[theorem]{Corollary}
\newtheorem{proposition}[theorem]{Proposition}
\newtheorem{conjecture}[theorem]{Conjecture}

\theoremstyle{definition}
\newtheorem{example}[theorem]{Example}
\newtheorem{definition}[theorem]{Definition}

\theoremstyle{remark}
\newtheorem{remark}[theorem]{Remark}

\newcommand{\lrs}{\href{http://cgm.cs.mcgill.ca/~avis/C/lrs.html}{\texttt{lrs}}}
\newcommand{\cddrp}{\href{http://www.inf.ethz.ch/personal/fukudak/cdd_home/}{\texttt{cddr+}}}

\newcommand{\MyTitle}{Hepp's bound for Feynman graphs and matroids}

\title{\MyTitle}
\author{\thanks{%
	University of Oxford, UK,
	\email{erik.panzer@maths.ox.ac.uk}}
	Erik Panzer
	\orcidicon{0000-0002-9897-5812}
}
\date{\today}

\hypersetup{%
	pdftitle = {\MyTitle},
	pdfauthor = {Erik Panzer}
}

\begin{document}
\maketitle
\begin{abstract}
	We study a rational matroid invariant, obtained as the tropicalization of the Feynman period integral. It equals the volume of the polar of the matroid polytope and we give efficient formulas for its computation. This invariant is proven to respect all known identities of Feynman integrals for graphs. We observe a strong correlation between the tropical and transcendental integrals, which yields a method to approximate unknown Feynman periods.
%The ideas in this article connect Mellin transforms, tropical geometry, matroids and convex geometry.
\end{abstract}

\tableofcontents

\section{Introduction}
To a connected graph $G$ with $N$ edges, Kirchhoff \cite{Kirchhoff:GalvanischeStroeme} attached the graph polynomial
\begin{equation}
	\PsiPol_G = \sum_{T \in \STrees{G}} \prod_{e\notin T} \SP_e
	\in \Z[\SP_1,\ldots,\SP_N]
	\label{eq:psipol}%
\end{equation}
given by a sum over the set $\STrees{G}$ of spanning trees.
In the context \cite{Nakanishi:GraphTheoryFeynmanIntegrals,BognerWeinzierl:GraphPolynomials} of perturbative quantum field theory, the variables $x_e$ associated to each edge $e$ are called Schwinger parameters. The scalar Feynman integral encoded by $G$ contributes the \emph{period} \cite{BlochEsnaultKreimer:MotivesGraphPolynomials,Periods}
\begin{equation}
	\Period{G}
	=\left( \prod_{e=1}^{N-1} \int_0^{\infty} \td \SP_e \right) \restrict{\frac{1}{\PsiPol_G^{\Dim/2}}}{\SP_N=1}
%	\in \R,
	\label{eq:period}%
\end{equation}
to the beta function of the field theory in $\Dim$ dimensions of space-time \cite{Schnetz:Census,Todorov:CausalityPositionSpace,KompanietsPanzer:phi4eps6}.
This integral is well-defined when $G$ is \emph{primitive logarithmically divergent} (\PlogDiv), which means that $\sdc{G}=0$ and $\sdc{\gamma}>0$ for every non-empty, proper subgraph $\gamma \subset G$, where
\begin{equation*}
	\sdc{G} = 
	\abs{G} - \tfrac{\Dim}{2} \cdot \loops{G}
	=
	\#\set{\text{edges in $G$}} - \tfrac{\Dim}{2} \cdot \# \set{\text{loops in $G$}}
\end{equation*}
is called the \emph{superficial degree of convergence} of $G$.
For example, the complete graph $K_4$ with $\abs{K_4}=6$ edges and $\loops{K_4}=3$ loops is {\PlogDiv} in $\Dim=4$ dimensions. Its period is
\begin{equation}
	\Period{K_4}
	= \Period{\Graph[0.25]{w3small}}
	= 6 \mzv{3}
	\approx 7.21%234
	\label{eq:ws3-period}%
\end{equation}
in terms of the Riemann zeta function. The transcendental numbers \cite{BroadhurstKreimer:KnotsNumbers,PanzerSchnetz:Phi4Coaction,BrownSchnetz:K3phi4} emerging as periods of graphs are extremely difficult to compute exactly, and even approximations are very challenging. For most graphs, the periods thus remain unknown.

\begin{figure}
	\centering%
	$
		\Graph[0.7]{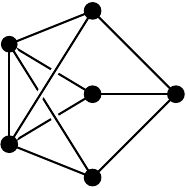}
		\quad\xleftarrow{\text{delete $v$}}\quad
		\Graph[0.7]{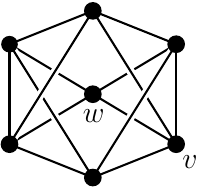}
		\quad\xrightarrow{\text{delete $w$}}\quad
		\Graph[0.7]{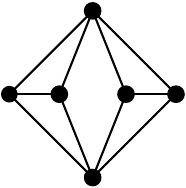}
	$%
	\caption{Two non-isomorphic graphs with the same completion.}%
	\label{fig:completion}%
\end{figure}
This complexity stimulates the search for simpler graph invariants, that are easier to compute, but still capture information about the period \cite{Crump:PhD}. To be meaningful, such invariants should obey period identities. For example, conformal invariance \cite{Broadhurst:5loopsbeyond} equates the periods of the complements $G{\setminus}v$ and $G{\setminus}w$ in \autoref{fig:completion}. This \emph{completion} relation and the \emph{product} identity show that
\begin{equation*}
	\Period{\Graph[0.3]{5Rv}}
	= \Period{\Graph[0.3]{5Rw}}
	= \Period{\Graph[0.25]{w3small}}^2
	= 36 \zeta(3)^2.
\end{equation*}
Further relations include planar \emph{duality}, the \emph{twist} \cite{Schnetz:Census} and the recently discovered \emph{Fourier split} \cite{HuSchnetzShawYeats:Further}, which generalizes the \emph{uniqueness relations} \cite{Kazakov:MethodOfUniqueness}. It is a challenge to construct non-trivial graph invariants with these symmetries.
In fact, apart from the period itself, only two such invariants had been found so far: the $c_2$ invariant \cite{Schnetz:Fq} and the (extended) graph permanent \cite{Crump:ExtendedPermanent,CrumpDeVosYeats:Permanent}. 

The $c_2$ invariant is constructed from the point counts of the hypersurface $\set{\PsiPol_G=0} \subset \FF{q}^N$ over finite fields $\FF{q}$. It is related to the number theory of the period \cite{BrownYeats:SpanningForestPolynomials,BrownSchnetz:ModularForms,PanzerSchnetz:Phi4Coaction,BrownSchnetz:K3phi4}, but several of the symmetries remain conjectural despite recent progress \cite{Yeats:SpecialCaseCompletion,Doryn:4face,Doryn:InvariantInvariant}.
For the permanent (of copies of the incidence matrix), the first four symmetries above are proven. It is not yet clear, however, what the permanent implies for the period.

In this paper, we study a new invariant obtained by a drastic simplification of the period integral: In the spirit of tropical geometry, replace $\Psi_G$ by its maximal monomial,
\begin{equation}
	\PsiTrop_G \defas \max_{T \in \STrees{G}} \prod_{e\notin T} \SP_e.
	\label{eq:psitrop}%
\end{equation}
This function is locally just some monomial, but which particular monomial it is depends on the actual values of the Schwinger parameters.
We refer to the corresponding integral
\begin{equation}
	\Hepp{G} \defas
%	\int \frac{\Omega}{(\PsiTrop_G)^{\Dim/2}} =
	\left( \prod_{e=1}^{N-1} \int_0^{\infty} \td \SP_e \right) \restrict{\frac{1}{(\PsiTrop_G)^{\Dim/2}}}{\SP_N=1}
	\in \Q
	\label{eq:hepp}%
\end{equation}
as \emph{the Hepp bound}, which defines a rational number for each {\PlogDiv} graph. It is indeed a bound on the period, since we have $\PsiTrop_G \leq \PsiPol_G \leq \PsiTrop_G \cdot \abs{\STrees{G}}$ and therefore
\begin{equation}
%	\frac{\Hepp{G}}{\abs{\STrees{G}}^{\Dim/2}}
%	\Hepp{G}\big/\abs{\STrees{G}}^{\Dim/2}
	\Hepp{G}\cdot\abs{\STrees{G}}^{-\Dim/2}
	\leq \Period{G} \leq \Hepp{G}.
	\label{eq:hepp-period-bound}%
\end{equation}
Hepp \cite{Hepp:BP} used this idea to deduce the convergence of the integral $\Period{G}$ from a power-counting argument, by dissecting the integration domain into regions
\begin{equation}
	\HeppSec{\sigma}
	=\big\{x_{\sigma(1)}<\cdots<x_{\sigma(N)}\big\} 
	\subset \R_+^N
	\label{eq:Hepp-sector}%
\end{equation}
according to the permutation $\sigma$ % \in \perms{N}$ 
of the edges determined by the order of the Schwinger parameters. These regions $\HeppSec{\sigma}$, called \emph{Hepp sectors}, have wide applications to renormalization, regularization and asymptotic expansion of Feynman integrals \cite{CalanRivasseau:Phi44,BergereCalanMalbouisson:Asymptotic,Schultka:ToricFeynman,SmirnovSmirnov:HeppSpeer}.

\paragraph{Symmetries}
The surprising observation is that the crude bound \eqref{eq:hepp} is in fact very well behaved and closely related to the actual period \eqref{eq:period}.
Firstly, we will prove
\begin{theorem}\label{thm:symmetries}
	The Hepp bound respects the five period symmetries from \cite{Schnetz:Census,HuSchnetzShawYeats:Further}.
\end{theorem}
This suggests that graphs with equal periods might also have the same Hepp bound. Analogous conjectures were made for the $c_2$ invariant and the permanent mentioned above. We conjecture that, for the Hepp bound, also the converse is true---at least in the case of $\field^4$ theory \cite{KleinertSchulteFrohlinde:CriticalPhi4}.
Concretely, we say that a graph $G$ is in $\field^4$ if it is {\PlogDiv} in $\Dim=4$ dimensions and every vertex has degree at most $4$ (for example $K_4$ from before).
More than a thousand $\field^4$ periods are known \cite{PanzerSchnetz:Phi4Coaction}, and they are all in agreement with:
\begin{conjecture}\label{con:faithful}
	Two $\field^4$ graphs have equal periods if and only if they have equal Hepp bounds.
\end{conjecture}
This significant strengthening of \autoref{thm:symmetries} is wrong for the $c_2$ invariant and the permanent, because there exist pairs of $\field^4$ graphs with the same $c_2$ invariant or permanent, but whose periods are known and different. There still seems to be a possibility, however, that $c_2$ and the permanent combined might distinguish periods \cite[Appendix~A]{Crump:ExtendedPermanent}.

The ``faithfulness'' of the Hepp bound according to \autoref{con:faithful} would imply new relations between yet unknown periods, which are not explained by the five operations discussed in \cite{Schnetz:Census,HuSchnetzShawYeats:Further}. The first examples of still unproven, conjectural identities of $\field^4$ periods appear at $8$ loops, where in the notation of \cite{Schnetz:Census} we find two pairs
\begin{equation}
	\Hepp{P_{8,30}{\setminus}v}
	= \Hepp{P_{8,36}{\setminus}v}
	= \tfrac{1724488}{3}
	\quad\ \text{and}\ \quad
	\Hepp{P_{8,31}{\setminus}v}
	= \Hepp{P_{8,35}{\setminus}v}
	= 536760
	\label{eq:hepp-8loop-pairs}%
\end{equation}
of graphs with equal Hepp bounds and thus conjecturally equal periods (see \autoref{fig:8-loop-hepp-coincidences}). Of these four, only $\Period{P_{8,31}} \approx 460.09$ could be computed exactly in \cite{Schnetz:Census}. The combinatorial origin of the equalities \eqref{eq:hepp-8loop-pairs} is currently not understood.

\paragraph{Hepp--Period correlation}
For explicit computations of the Hepp bound, the integral representation \eqref{eq:hepp} is not very practical.
In \autoref{prop:hepp-1pi-flags} we rewrite it as a sum over flags of bridgeless subgraphs, a generalization of ear decompositions. This formula reads
\begin{equation*}
	\Hepp{G} = 
	\sum_{\gamma_1\subsetneq \cdots \subsetneq \gamma_{\ell}=G}
	\frac{
		\abs{\gamma_1} \cdot \abs{\gamma_2\setminus\gamma_1}\cdots \abs{G\setminus \gamma_{\ell-1}}
	}{
		\sdc{\gamma_1} \cdots \sdc{\gamma_{\ell-1}}
	}
\end{equation*}
and allows the calculation of $\Hepp{G}$ for most graphs of interest. We used it to obtain the Hepp bounds of all $\field^4$ graphs with $\loops{G}\leq 11$ loops. For example, we find
\begin{equation}
	\Hepp{K_4} 
	= \Hepp{\Graph[0.25]{w3small}}
	= 84, \label{eq:ws3-hepp}%
\end{equation}
which should be compared with the much smaller period \eqref{eq:ws3-period}. So the Hepp bound is very crude indeed and it can exceed the period by several orders of magnitude. Surprisingly, this bound nonetheless allows us to predict the numeric values of periods within a range of a few percent.
Namely, we observe that the period is very strongly correlated with the Hepp bound, as illustrated in \autoref{fig:hepp-period}. At higher loop orders, a smooth curve interpolating the known periods then gives estimates for unknown periods.

\begin{figure}%
	\centering%
	\includegraphics[width=0.8\textwidth]{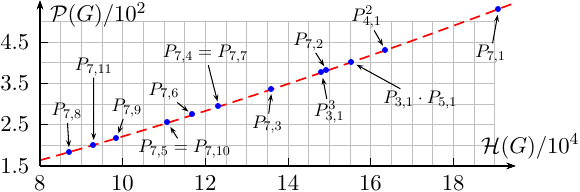}%
\caption{The period as a function of the Hepp bound for all $\field^4$ periods at $7$ loops, the dashed line is a power law fit $\Period{G} \approx 3.96/10^5 \cdot [\Hepp{G}]^{1.3495}$. The graphs are labelled according to \cite{Schnetz:Census}, see \autoref{tab:hepp-7loop} for details.}%
	\label{fig:hepp-period}%
\end{figure}
It is remarkable that the rational number $\Hepp{G}$, which is easy to compute for any graph, gives such a sensitive measure of the intricate period integral $\Period{G}$. This connection was exploited in \cite{KompanietsPanzer:phi4eps6} to estimate the contributions from higher orders in perturbation theory to a calculation of the beta function in $\field^4$ theory, and we are optimistic that generalizations and refinements of this method can provide a new approach to the numeric evaluation of Feynman integrals, efficient even at large loop orders.

\paragraph{Geometry}
The correlation of $\Hepp{G}$ with $\Period{G}$ is so far an empirical observation, but it appears to be related to a geometric interpretation of the Hepp bound. The approximation of Feynman integrals by the method of sector decomposition uses a resolution of singularities \cite{BognerWeinzierl:ResolutionOfSingularities} of the graph hypersurface $\set{\PsiPol_G=0}$. This is achieved most efficiently by a triangulation of the normal fan of the Newton polytope of $\PsiPol_G$ \cite{KanekoUeda:GeometricSector,SchlenkZirke:SecDec3.0,BHJKSZ:SecDec3.0,Schultka:ToricFeynman}. It is also known as the \emph{spanning tree polytope} of $G$, and we define it as the convex hull
\begin{equation}
	\Newton{G} \defas \conv \setexp{\vt{T}}{T \in \STrees{G}}
	\subset \R^N
	\label{eq:Newton}%
\end{equation}
of the characteristic vectors $\vt{T}$ of the spanning trees $T$, with coordinates $\vt{T,e}=1$ for all edges $e\in T$ in the tree, and $\vt{T,e}=-1$ whenever $e\notin T$. The \emph{polar} of this polytope is
\begin{equation}
	\Polar{G}
	= \bigcap_{T \in \STrees{G}} \setexp{\vec{\SPlog}}{\vec{\SPlog} \cdot \vt{T} \leq 1}
	\subset \R^N.
	\label{eq:polar}%
\end{equation}

\begin{theorem}\label{thm:Hepp=Vol}
	For a graph $G$ with $N$ edges that is {\PlogDiv} in $\Dim=4$ dimensions, we have
	\begin{equation}
		\Hepp{G} = (N-1)! \cdot \Vol \left( \Polar{G} \cap \set{\SPlog_N = 0} \right).
		\label{eq:Hepp=Vol}%
	\end{equation}
\end{theorem}
The asymptotic growth of period integrals implies that the volume of the polytope $\Polar{G}$ is concentrated on directions near the facet normals $\vt{T}$, like a cross-polytope. Dually, $\Newton{G}$ behaves like a cube in the sense that its volume is concentrated in the corners. We show in \autoref{lem:Period-SPlog} that the period integral \eqref{eq:period} can be written as an integral of a log-concave function over the polytope $\Polar{G}$, and argue that this explains at least qualitatively the correlation between the period and $\Hepp{G}$.

An important tool for our proofs of the symmetries is a functional generalization of the Hepp bound: Instead of the mere number \eqref{eq:hepp}, we consider the rational function
\begin{equation}
	\Hepp{G,\vec{\ind}}
	\defas
	\left( \prod_{e=1}^{N-1} \int_0^{\infty} \SP_e^{\ind_e-1} \td \SP_e \right)
	\left. \frac{1}{(\PsiTrop_{G})^{\Dim/2}} \right|_{x_N=1}
	\in \Q(\ind_1,\ldots,\ind_N)
	\label{eq:hepp-mellin}%
\end{equation}
given by the Mellin transform of $(\PsiTrop_G)^{-\Dim/2}$. This is a well-defined rational function for \emph{arbitrary} biconnected graphs, and not just $\PlogDiv$ graphs. The period symmetries from \cite{Schnetz:Census,HuSchnetzShawYeats:Further} extend to this functional setting, and we prove them in this generality.
We exploit that $\Hepp{G,\vec{\ind}}$ is a function with simple poles, located on hyperplanes $\setexp{\vec{\ind}}{\sdc{\gamma}=0}$ for suitable graphs $\gamma \subset G$, where the superficial degree of convergence is the linear function
\begin{equation*}
	\sdc{\gamma} 
	= \sum_{e \in \gamma} \ind_e - \frac{\Dim}{2} \cdot \loops{\gamma}.
\end{equation*}
The Hepp bound has a pole at $\sdc{\gamma}=0$ precisely when $\gamma$ and its quotient $G/\gamma$ are biconnected. Such subgraphs correspond to a facet of the Newton polytope, and it factorizes as $\Newton{\gamma} \times \Newton{G/\gamma} \subset \partial \Newton{G}$.
Generalizing \eqref{eq:Hepp=Vol}, the function $\Hepp{G,\vec{\ind}}$ is the volume of the polar of the translated Newton polytope $\vec{a} + \Newton{G}$, such that the facets of $\Newton{G}$ correspond directly to the poles of the Hepp bound. Hence the residues are products
\begin{equation}
	\Res_{\sdc{\gamma} = 0}
	\Hepp{G,\vec{\ind}}
	= \Hepp{\gamma,\vec{\ind}'} 
	\cdot \Hepp{G/\gamma, \vec{\ind}''}
	\label{eq:intro-residues}%
\end{equation}
and separate the dependence on variables $\vec{\ind}'=(\ind_e)_{e \in \gamma}$ associated to the subgraph and the quotient, $\vec{\ind}''=(\ind_e)_{e \in G\setminus \gamma}$. 
This gives a tool for inductive proofs of identities of rational functions, similar to the BCFW recursion \cite{BrittoCachazoFengWitten:DirectProof}.

The same mechanism of associating rational functions with simple poles and factorizing residues to polytopes is also used for tree level scattering amplitudes, under the name \emph{canonical forms} \cite{ArkaniBaiLam:PositiveCanonical}. In that context, the Mellin variables $\ind_e$ play the role of Mandelstam invariants obtained from momenta of particles; the factorization above is interpreted as \emph{unitarity}, and the fact that only simple poles occur is attributed to \emph{locality}.

\paragraph{Matroid invariants}
The Hepp bound \eqref{eq:hepp-mellin} is not restricted to graphs and extends to all matroids. In this paper we work in this more general universe, with the sole exception of three symmetries: completion, twist and Fourier-split are only defined for graphs.

Our definition of the Hepp bound gives zero whenever a matroid is not connected. Otherwise, let $\Sing{M}$ denote the set of submatroids $\gamma\subsetneq M$ such that both $\gamma$ and $M/\gamma$ are connected. These $\gamma$ label the facets of $\Newton{M}$, and we will show:
\begin{proposition}
	The Hepp bound $\Hepp{M,\vec{\ind}}$ of a connected matroid $M$ is a non-zero rational function with simple poles, precisely on the hypersurfaces $\sdc{\gamma}=0$ for $\gamma \in \Sing{M}$, and factorizing residues as in \eqref{eq:intro-residues}.
\end{proposition}
Matroid polytopes have been studied extensively, but little seems to be known about their polars. Our findings suggest that the volume of the (sliding) polar $\Polar{M}(\vec{\ind})$, as a rational function, is a very interesting matroid invariant. It has a rich structure and it determines the matroid completely (\autoref{rem:inverse-Mellin}). Furthermore, the symmetries of the Hepp bound show that polar volumes are subject to more identities than the volumes of the matroid polytopes themselves.

\subsection*{Outline of the paper}
This article aims to be broadly accessible and includes relevant definitions and results from the combinatorial literature. The focus here is on the mathematical properties of the Hepp bound, but the particle physicist will recognize the motivation and applications.

We give a combinatorial definition of the Hepp bound for arbitrary matroids in \autoref{sec:basics}, which is consistent with the Mellin integral, and obtain its poles and the factorization of residues discussed above. We compute the Hepp bound of uniform matroids, and illustrate relations to Crapo's and Derksen's matroid invariants in \autoref{sec:other-invariants}.

The remaining sections are essentially independent of each other:
Formulas in terms of flags of bridgeless submatroids or flats are derived in \autoref{sec:flags} and applied to compute the Hepp bound of all wheel graphs.
The five period symmetries are proven for the Hepp bound in \autoref{sec:symmetries}.
We report the results for $\field^4$ graphs in \autoref{sec:phi4}, addressing the correlation with the period and unexplained identities, and we discuss improvements of the Hepp bound.
The convex geometric point of view is worked out in \autoref{sec:geometry}.

\subsection*{Acknowledgments}
I am indebted to Karen Yeats for invitations to SFU Vancouver in 2016, where the groundwork for this research was laid, and to the University of Waterloo in 2018. At these workshops I benefited especially from extended discussions with Karen, Michael Borinsky and Iain Crump.
Furthermore I thank the organizers and the participants of the \emph{summer school on structures in local quantum field theory} in \href{https://www.houches-school-physics.com/}{Les Houches} and the \href{https://indico.cern.ch/event/646820/}{\emph{Amplitudes} workshop} at SLAC for chances to present aspects of this work in 2018.

This project and my understanding was further shaped by inputs from Marko Berghoff, Jacob Bourjaily, David Broadhurst, Francis Brown, Yoav Len, Carlos Mafra, Alejandro Morales, Max M\"{u}hlbauer, Matteo Parisi, Giulio Salvatori, Johannes Schlenk, Oliver Schnetz and Konrad Schultka.

Finally I thank the \href{https://eth-its.ethz.ch/}{Institue for Theoretical Studies} at ETH Z\"{u}rich for hospitality and perfect working conditions during final stages of the preparation of this manuscript in 2019, at the thematic programme \href{https://eth-its.ethz.ch/activities/talks-in-theoretical-sciences1/Modular_forms_periods_and_scattering_amplitudes.html}{\emph{Modular forms, periods and scattering amplitudes}}.

\section{Definition and basic properties}
%The combinatorial Hepp bound}
\label{sec:basics}

Our construction is motivated by the Feynman integrals of perturbative quantum field theory \cite{Nakanishi:GraphTheoryFeynmanIntegrals,Brown:PeriodsFeynmanIntegrals,Speer:GeneralizedAmplitudes}. The period $\Period{G}$ defined in \eqref{eq:period} is only one particular integral that can be associated to a graph $G$. More generally, one considers the Mellin transform
\begin{equation}
	\Period{G,\vec{\ind}}
	\defas
	\left( \prod_{e=1}^{N-1} \int_0^{\infty} \SP_e^{\ind_e-1} \td \SP_e \right)
	\left. \frac{1}{\Psi_{G}^{\Dim/2}} \right|_{x_N=1}
	\label{eq:period-mellin}%
\end{equation}
of $\PsiPol_G^{-\Dim/2}$ as a multivariate function of the variables $\vec{\ind}=(\ind_1,\ldots,\ind_N)$, which we call \emph{indices}.
In particle physics, they are the exponents of the momentum space propagators attached to each edge of the graph. The dependence of \eqref{eq:period-mellin} on $\vec{\ind}$ is very useful to understand Feynman integrals. For example, difference equations with respect to the indices are heavily used in practical calculations \cite{ChetyrkinTkachov:IBP,Grozin:IBP}. When amplitudes in string theories are seen as Mellin transforms, then the indices $\vec{\ind}$ play the role of particle momenta \cite{ArkaniHamedHeLam:Stringy}.

The function \eqref{eq:period-mellin} is called the \emph{analytically regularized} Feynman integral. If the graph $G$ is biconnected, this integral converges for suitable indices, and it extends to a unique meromorphic function of the indices, with poles on families of hyperplanes \cite{Speer:GeneralizedAmplitudes}.
\begin{example}\label{ex:period-bubble}
	The cycle $\Cycle{2}=\Graph[0.3]{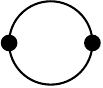}$ with two edges (also called ``bubble'') has $\loops{C_2}=1$ loop and it is {\PlogDiv} when $\ind_1,\ind_2>0$ and $\frac{\Dim}{2}=\ind_1+\ind_2$. Since $C_2$ has precisely two spanning trees $\set{e_1}$ and $\set{e_2}$, consisting of the individual edges, the graph polynomial is $\PsiPol_{C_2} = \SP_1 + \SP_2$. The regularized Feynman integral therefore becomes
	\begin{equation*}
		\Period{\Graph[0.3]{bubble},\vec{\ind}}
		=
		\int_0^{\infty} \frac{\SP_1^{\ind_1-1} \td \SP_1}{(\SP_1+1)^{a_1+a_2}}
		= \frac{\Gamma(\ind_1)\Gamma(\ind_2)}{\Gamma(\ind_1+\ind_2)},
	\end{equation*}
	which is meromorphic in $\vec{\ind} \in \C^2$ with poles on the hyperplanes $\ind_1,\ind_2 \in \Z_{\leq 0}$.
\end{example}

The Hepp bound \eqref{eq:hepp-mellin} is the variant of \eqref{eq:period-mellin} obtained by replacing the graph polynomial $\PsiPol_G$ with its tropical analogue $\PsiTrop_G$. This yields a \emph{rational} function of the indices, which captures precisely the first pole in each family of singularities of $\Period{G,\vec{\ind}}$.
\begin{example}\label{ex:hepp-bubble}
	The tropical graph polynomial of the bubble is $\PsiTrop_{C_2} = \max\set{\SP_1,\SP_2}$.
%	Whenever both indices have positive real parts $\RePart \ind_1,\RePart \ind_2 > 0$, the Hepp bound integral
	Whenever both indices are positive, $\ind_1>0$ and $\ind_2 > 0$, the Hepp bound integral
	\begin{equation*}
		\Hepp{\Graph[0.3]{bubble},\vec{\ind}}
		= \int_0^{\infty} \frac{\SP_1^{\ind_1-1} \td \SP_1}{(\max \set{\SP_1,1})^{a_1+a_2}}
		= \int_0^1 \SP_1^{\ind_1-1} \td \SP_1 + \int_1^{\infty} \SP_1^{-\ind_2-1} \td \SP_1
%		= \frac{1}{\ind_1} + \frac{1}{\ind_2}
		= \frac{\ind_1+\ind_2}{\ind_1\ind_2}
	\end{equation*}
	is absolutely convergent. It has poles at $\ind_1=0$ and $\ind_2=0$, and it can be obtained formally by replacing each gamma function in $\Period{C_2,\vec{\ind}}$ by its first pole, $\Gamma(s) \mapsto 1/s$.
\end{example}
The Hepp bound emerges as the \emph{tropical limit} of the Feynman integral: A substitution $\SP_e^{\varepsilon} \mapsto \SP_e$ transforms
$
	\SP_e^{\varepsilon \ind -1}\td \SP_e
	\mapsto\SP_e^{\ind -1}\td \SP_e/\varepsilon
$
and
$
	\PsiPol_{G}^{\varepsilon}
	\mapsto
	(\sum_{T\in \STrees{G}} \prod_{e \notin T} \SP_e^{1/\varepsilon}\big)^{\varepsilon}
$, the latter converges to $\PsiTrop_G$ for positive $\varepsilon \rightarrow 0$, and so the integral \eqref{eq:period-mellin} turns into the Hepp bound:
%$
%	\varepsilon\SP^{\varepsilon \ind -1}\td \SP 
%	= (\SP^{\varepsilon})^{\ind-1} \td \SP^{\varepsilon}
%$,
\begin{equation}
	\Hepp{G,\vec{\ind}}
	= \lim_{\varepsilon\rightarrow 0}
	\left[
		\varepsilon^{N-1}
		\Period{G,\varepsilon\vec{\ind}}
	\right]
	.
	\label{eq:tropical-limit}%
\end{equation}

The inequality $\PsiTrop_G\leq \PsiPol_G \leq \PsiTrop_G \cdot \abs{\STrees{G}}$ underlying \eqref{eq:hepp-period-bound} shows that both Mellin integrals \eqref{eq:hepp-mellin} and \eqref{eq:period-mellin} have the same domain of convergence. Outside this domain, the integrals define the Hepp bound and Feynman integral only indirectly via analytic continuation.
Our first goal is a definition of the Hepp bound as an explicit rational function, valid for all indices, and without reference to integrals.

\begin{remark}\label{rem:inverse-Mellin}
	If $G$ is biconnected with $N\geq 2$ edges, we will see that the Mellin integrals converge for suitable indices. The inverse Mellin transform (an integral over $\vec{\ind}$) then allows us to recover the function $\vec{\SP}\mapsto\PsiTrop_G(\vec{\SP})$ from the Hepp bound $\Hepp{G,\vec{\ind}}$. Every spanning tree dictates $\PsiTrop_G$ in some domain of the Schwinger parameters, so we can obtain the set $\STrees{G}$ from the tropical graph polynomial. Similarly, we can reverse engineer the spanning trees from the Feynman integral $\Period{G,\vec{\ind}}$. This may be illustrated as
	\begin{equation*}
		\Hepp{G,\vec{\ind}}
		\longleftrightarrow \PsiTrop_G(\vec{\SP})
		\longleftrightarrow \STrees{G}
		\longleftrightarrow \PsiPol_G(\vec{\SP})
		\longleftrightarrow
		\Period{G,\vec{\ind}},
	\end{equation*}
	where each arrow $A\longleftrightarrow B$ indicates that $A$ determines $B$ and vice versa. We see that the rational function $\Hepp{G,\vec{\ind}}$ completely determines the Feynman integral $\Period{G,\vec{\ind}}$ as a function of $\vec{\ind}$. This does not, however, impinge on \autoref{con:faithful}, which is a statement about special values at $\ind_1=\cdots=\ind_N=1$, and only for certain graphs ({\PlogDiv} in $\field^4$).
\end{remark}

\subsection{Combinatorial definition}

We consider arbitrary undirected graphs, which may have multiple edges between the same pair of vertices, and edges with both endpoints at the same vertex (self-loops) are also allowed. We write $\abs{G}\defas \abs{\EG{G}}$ for the number of edges, which is often also denoted by $N$.
The \emph{loop number} $\loops{G}$ is the first Betti number of the graph, which is
\begin{equation}
	\loops{G} = \abs{G} - \abs{\VG{G}} + \nCG{G}
	\label{eq:euler}%
\end{equation}
in terms of the number $\abs{\VG{G}}$ of vertices and the number $\nCG{G}$ of connected components.
The \emph{superficial degree of convergence} of a subgraph $\gamma\subseteq G$ is the linear function
\begin{equation}
	\sdc{\gamma} = \sdc[\vec{\ind}]{\gamma}
	\defas \sum_{e\in \gamma} \ind_e - \frac{\Dim}{2} \cdot \loops{\gamma},
	\label{eq:sdc}%
\end{equation}
and we will always impose the condition $\sdc{G}=0$ called `logarithmic divergence'.
For graphs with loops, it means that the dimension is determined by the indices as
\begin{equation*}
	\Dim=2\frac{\ind_1+\cdots+\ind_N}{\loops{G}}.
\end{equation*}
For forests (graphs without loops), the dimension disappears from \eqref{eq:sdc} and plays no role. The condition $\sdc{G}=0$ then imposes the constraint $\ind_1+\cdots+\ind_N=0$.
\begin{definition}\label{def:multi-hepp-from-sectors}
	If $G$ is a graph with $N\geq 1$ edges and we are given a permutation $\sigma$ of its edges, we denote by $G^{\sigma}_k \defas \set{\sigma(1),\ldots,\sigma(k)}$ the subgraphs formed by the first $k$ edges in the order $\sigma$.
	The Hepp bound of $G$ is the homogeneous rational function
	\begin{equation}
		\Hepp{G,\vec{\ind}}
		\defas \sum_{\sigma \in \perms{N}} \frac{1}{\sdc{G^{\sigma}_1}\cdots \sdc{G^{\sigma}_{N-1}}}
%		\in \Q(\ind_1,\ldots,\ind_N)
		\label{eq:multi-hepp-from-sectors}%
	\end{equation}
	of degree $1-N$, obtained by summing over all $N!$ permutations.
%	If $\loops{G}=0$, the $\sdc{G^{\sigma}_k}$ are to be computed in the dimension fixed by $\sdc{G}=0$.
	For a single edge $N=1$, the empty product in the denominator is defined as unity such that $\Hepp{G,\ind_1} = 1$.
\end{definition}
\begin{example}\label{ex:hepp-bubble-flags}
	For the bubble from \autoref{ex:hepp-bubble}, there are only two $N!=2$ permutations to consider. The graphs $G^{(1,2)}_1 = \set{1}$ and $G^{(2,1)}_1 = \set{2}$ formed by the first edges have no loops and we recover the result $\Hepp{\Cycle{2},\vec{\ind}} = \frac{1}{\ind_1} + \frac{1}{\ind_2}$ from $\sdc{\set{e}}=\ind_e$.
\end{example}
\begin{figure}
	\centering
	\begin{tabular}{ccccccccc}
		$\Graph[0.6]{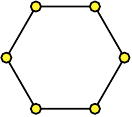}$ &
		$\Graph[0.6]{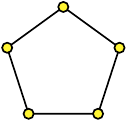}$ &
		$\Graph[0.6]{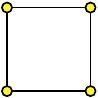}$ &
		$\Graph[0.6]{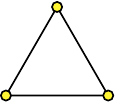}$ &
		$\Graph[0.6]{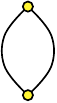}$ &
		$\Graph[0.6]{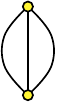}$ &
		$\Graph[0.6]{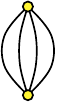}$ &
		$\Graph[0.6]{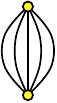}$ &
		$\Graph[0.6]{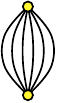}$
		\\
		$ \Cycle{6}$ &
		$ \Cycle{5}$ &
		$ \Cycle{4}$ &
		$ \Cycle{3}$ &
		$ \Cycle{2}=\Bond{2}$ &
		$ \Bond{3}$ &
		$ \Bond{4}$ &
		$ \Bond{5}$ &
		$ \Bond{6}$
	\end{tabular}%
	\caption{The first few cycle/polygon graphs $\Cycle{n}$, and the bonds/dipoles/melons $\Bond{n}$.}%
	\label{fig:cycles-bonds}%
\end{figure}
\begin{example}\label{ex:hepp-cycle-1}
	Consider any cycle $\Cycle{N}$ (\autoref{fig:cycles-bonds}) with unit indices $\ind_1=\cdots=\ind_N=1$. Because every proper subgraph is a forest, we get $\sdc{G^{\sigma}_k}=k$ in \eqref{eq:multi-hepp-from-sectors}. So each of the $N!$ summands contributes $1/(N-1)!$ and the total Hepp bound is $\Hepp{\Cycle{N}} = N = \Dim/2$.
\end{example}
\begin{example}\label{ex:2-forest}
	If $G=\Graph[0.3]{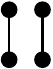}$ consists of two isolated edges, we obtain the same expression $\frac{1}{\ind_1}+\frac{1}{\ind_2}$ from the sum \eqref{eq:multi-hepp-from-sectors}. But in this case without loops, we consider it as a function on the hyperplane $0=\sdc{G}=\ind_1+\ind_2$ where it vanishes. Hence we find $\Hepp{\Graph[0.3]{twosticks},\vec{\ind}}=0$.
\end{example}
In the same way, we will see later that $\Hepp{G,\vec{\ind}}=0$ for all forests with $N \geq 2$ edges. The Hepp bound is therefore only really interesting for graphs with loops.

\subsection{The Mellin integral}
In order to relate \autoref{def:multi-hepp-from-sectors} to the integral \eqref{eq:hepp-mellin} for arbitrary graphs, we define the tropical graph polynomial $\PsiTrop_G$ for disconnected graphs similarly as in \eqref{eq:psitrop}, but with a sum over all spanning forests. In particular, $\PsiTrop_G=1$ whenever $G$ is itself a forest.

It follows from \eqref{eq:euler} that the graph polynomial is homogeneous of degree $\loops{G}$, so
\begin{equation*}
%	\PsiPol_G(\lambda \SP_1,\ldots,\lambda \SP_N) = \lambda^{\loops{G}} \cdot \PsiPol_G
%	\quad\text{and}\quad
	\PsiTrop_G(\lambda \SP_1,\ldots,\lambda \SP_N) = \lambda^{\loops{G}} \cdot \PsiTrop_G(\SP_1,\ldots,\SP_N)
%	\label{eq:psi-homo}%
\end{equation*}
for all positive $\lambda>0$ and $\vec{\SP}\in \R_+^N$.
The function $(\PsiTrop_G)^{-\Dim/2}\prod_e \SP_e^{\ind_e}$ is therefore homogeneous of degree $\sdc{G}$, and the condition $\sdc{G}=0$ ensures that the Mellin integral \eqref{eq:hepp-mellin} is in fact an integral over projective space, written in the chart $\SP_N=1$. It is therefore irrelevant which edge we choose to label $N$, and we can write \eqref{eq:hepp-mellin} symmetrically as
\begin{equation*}
	\Hepp{G,\vec{\ind}} =
%	\int \frac{\Omega(\vec{\ind})}{(\PsiTrop_G)^{\Dim/2}} \defas 
	\int_{\Projective^G} \frac{\Omega(\vec{\ind})}{(\PsiTrop_G)^{\Dim/2}}
	.
\end{equation*}
The integration domain is the positive orthant $\Projective^G \defas \setexp{[\SP_1:\cdots:\SP_N]}{\SP_1,\ldots,\SP_N>0} \subset \RP^{N-1}$ inside real projective space, and $\Omega(\vec{\ind})$ denotes the $N-1$ form
\begin{equation}
	\Omega(\vec{\ind})
	\defas \left( \prod_{e=1}^N \SP_e^{\ind_e-1} \right) \sum_{e=1}^N (-1)^{e-1} \SP_e \bigwedge_{f\neq e} \td \SP_f.
	\label{eq:volume-form}%
\end{equation}
To compute the integral, we subdivide the domain $\Projective^G$ into the Hepp sectors \eqref{eq:Hepp-sector}:
\begin{equation*}
	\Hepp{G,\vec{\ind}} = \sum_{\sigma \in \perms{N}} \int_{\HeppSec{\sigma}} \frac{\Omega(\vec{\ind})}{\left( \PsiTrop_G \right)^{\Dim/2}}.
\end{equation*}
Each summand is an integral over the projective simplex with $0<\SP_{\sigma(1)}<\cdots<\SP_{\sigma(N)}$. Hepp noted that within every sector $\HeppSec{\sigma}$, there is a unique spanning tree $T_{\sigma} \in \STrees{G}$ that dominates all others, i.e.\ the function $\PsiTrop_G$ is given by a fixed monomial inside the sector:
\begin{equation*}
	\PsiTrop_G(x) \Big|_{x \in \HeppSec{\sigma}}  = \prod_{e \notin T_{\sigma}} \SP_e.
	\tag{$\ast$} \label{eq:psitrop|sector=monomial}%
\end{equation*}
Indeed, the dominating spanning tree $T_{\sigma}$ (or forest, if $G$ is disconnected) is nothing but the minimum weight spanning tree with respect to the edge weights $\log \SP_e$, and following Kruskal \cite{Kruskal:SSTandTSP} this spanning tree is uniquely determined by the total order $\sigma$ of the weights:
\begin{lemma}[Kruskal's greedy algorithm]
	\label{lem:Kruskal}%
	If we are given a total order of the edge weights, hence a permutation $\sigma \in \perms{N}$ of the edges, %let $G^{\sigma}_k$ denote the subgraph of $G$ induced by the edges $\sigma(1),\ldots,\sigma(k)$ and let $G^{\sigma}_0$ be the empty graph.
	then the minimum weight spanning forest $T_{\sigma}$ consists of precisely those edges that do not increase the loop number:
	\begin{equation*}
		T_{\sigma} = \setexp{\sigma(k)}{\loops{G^{\sigma}_k}=\loops{G^{\sigma}_{k-1}}}
		\in \STrees{G}.
	\end{equation*}
\end{lemma}
The edges $e \notin T_{\sigma}$ contributing to the dominating monomial \eqref{eq:psitrop|sector=monomial} are therefore precisely those edges $e=\sigma(k)$ at which the loop number $\loops{G^{\sigma}_k}=1+\loops{G^{\sigma}_{k-1}}$ increases.
In the affine chart $\SP_{\sigma(N)}=1$, the integral over a Hepp sector can therefore be written as
\begin{equation*}
	\int_{\HeppSec{\sigma}} \frac{\Omega(\vec{a})}{\left( \PsiTrop_G \right)^{\Dim/2}}
	= \int_{0<\SP_{\sigma(1)}<\cdots<\SP_{\sigma(N)}=1}
	\prod_{k=1}^{N-1} \SP_{\sigma(k)}^{\ind_{\sigma(k)}-1-\frac{\Dim}{2}\big( \loops{G^{\sigma}_k}-\loops{G^{\sigma}_{k-1}} \big)} \td \SP_{\sigma(k)}.
\end{equation*}
Changing variables to $y_k=x_{\sigma(k)}/x_{\sigma(k+1)}$, this evaluates to the summand in \eqref{eq:multi-hepp-from-sectors}:
\begin{equation*}
%	\int_{\HeppSec{\sigma}} \frac{\Omega(\vec{a})}{\left( \PsiTrop_G \right)^{\Dim/2}} = 
	\prod_{e=1}^{N-1}
	\int_0^1 y_k^{\big(\sum_{i=1}^k \ind_{\sigma(i)}\big)-1-\frac{\Dim}{2}\loops{G^{\sigma}_k}} \td y_{k}
	= \frac{1}{\sdc{G^{\sigma}_1} \!\cdots \sdc{G^{\sigma}_{N-1}}}.
%	\label{eq:hepp-sector-integral}%
\end{equation*}
This integral converges precisely when all real parts $\RePart \sdc{G^{\sigma}_1},\ldots,\RePart \sdc{G^{\sigma}_{N-1}}>0$ are positive. We can summarize this calculation as follows:
\begin{proposition}\label{cor:convergence-domain}
	The Mellin integrals \eqref{eq:hepp-mellin} and \eqref{eq:period-mellin} converge precisely for those indices $\vec{\ind} \subset \C^N$ whose real part lies in the open convex polyhedral cone
	\begin{equation}
		\ConvDom \defas \bigcap_{\emptyset \neq \gamma \subsetneq G} \setexp{\vec{\ind} \in \R^N}{\sdc[\vec{\ind}]{\gamma} > 0}
		\subseteq \R^N.
		\label{eq:convergence-domain}%
	\end{equation}
	For such $\vec{\ind}$, the Hepp bound integral \eqref{eq:hepp-mellin} coincides with the function in \autoref{def:multi-hepp-from-sectors}.
\end{proposition}
The characterization of the convergence of Feynman integrals in terms of \emph{power counting} conditions $\sdc{\gamma}>0$ is fundamental for renormalization in physics \cite{Dyson:SmatrixQED,Weinberg:HighEnergy}. Many of these constraints are redundant, however. The independent constraints are given in \eqref{eq:conv-minimal}; for more general Feynman integrals with kinematics see \cite{Speer:SingularityStructureGenericFeynmanAmplitudes,Speer:GeneralizedAmplitudes}.

In the sequel we will mostly use the combinatorial formula \eqref{eq:multi-hepp-from-sectors}, which allows us to ignore questions of convergence. However, the convergence domain $\ConvDom$ is non-empty in all cases of interest (\autoref{lem:convdom-nonempty}), and we may thus use the Mellin integral freely.
\begin{example}\label{ex:hepp-cycle}
	We compute the Hepp bound of the cycle graph $\Cycle{N}$ as a function of the indices. Its $N$ spanning trees are the edge complements $T=\Cycle{N}\setminus\set{e}$ and thus $\PsiTrop = \max\set{ \SP_1,\ldots,\SP_N}$. 
	The integral over the domain where $\PsiTrop=\SP_k$ is maximal gives
	\begin{equation*}
		\int_{\PsiTrop = \SP_k}
%		\int_{\substack{\SP_k>\SP_e\\ \text{for all}\ e\neq k}}
		\frac{\Omega(\vec{\ind})}{(\PsiTrop)^{\Dim/2}}
		= \left( \prod_{e \neq k} \int_0^{\SP_k} \SP_e^{\ind_e-1} \td \SP_e \right)
		\frac{\SP_k^{\ind_k-1}}{\SP_k^{\Dim/2}} \bigg|_{\SP_k=1}
		= \prod_{e \neq k} \frac{1}{\ind_e},
	\end{equation*}
	computed in the chart $\SP_k=1$.
	Adding all these contributions, we find in generalization of \autoref{ex:hepp-bubble} and \autoref{ex:hepp-cycle-1} that the full Hepp bound function is given by
	\begin{equation}
		\Hepp{\Cycle{N},\vec{\ind}}
		= \frac{\ind_1+\cdots+\ind_N}{\ind_1 \cdots \ind_N}
		= \frac{\Dim/2}{\ind_1\cdots \ind_N}.
		\label{eq:hepp-cycle}%
	\end{equation}
\end{example}

\subsection{Matroids}

The Hepp bound depends only on the set of spanning trees, so it is not sensitive to the full combinatorial structure of a graph. This suggests a generalization, and indeed the weaker notion of a matroid is sufficient.
We use standard terminology as in \cite{Oxley:MatroidTheory}.

A matroid $M=(\EG{M},\Indeps{M})$ consists of a \emph{ground set} $\EG{M}$ and a non-empty family $\Indeps{M}$ of subsets of $\EG{M}$, called the \emph{independent sets}, such that
\begin{enumerate}
	\item Every subset $\delta\subset \gamma$ of an independent set $\gamma\in\Indeps{M}$ is independent: $\delta\in\Indeps{M}$.
	\item If $\delta,\gamma\in\Indeps{M}$ are independent and $\abs{\delta}<\abs{\gamma}$, then we can find an element $e\in \gamma\setminus\delta$ such that $\delta \cup \set{e} \in \Indeps{M}$ is independent.
\end{enumerate}
\begin{example}
	Every graph defines the \emph{cycle matroid} $\GMat{G}=(\EG{G},\Indeps{G})$ on the edges $\EG{G}$ as ground set \cite{Tutte:LecturesOnMatroids,Oxley:MatroidTheory}. Its independent sets 
$
	\Indeps{G} = \setexp{\gamma \subseteq \EG{G}}{\loops{\gamma}=0}
%	\label{eq:graphic-matroid}%
$
	are the forests (loopless subgraphs) of $G$.
	It is well understood when two graphs share isomorphic cycle matroids \cite{Whitney:2isomorphic,Truemper:OnWhitney}, and this is exploited in practical calculations \cite{ManteuffelStuderus:Reduze2}.
\end{example}
Matroids that come from graphs in this way are called \emph{graphic}, and most matroids are not graphic. Even non-graphic matroids do arise in Feynman integral calculations \cite{KreimerYeats:TensorMatroids}.
\begin{example}
	The uniform matroid $\UM{n}{r}$ with rank $0 \leq r \leq n$ is defined on the ground set $\set{1,\ldots,n}$ and its independent subsets are precisely all subsets of size at most $r$. The extremes $\UM{n}{n}$ and $\UM{n}{0}$ are the cycle matroids of forests and collections of self-loops, respectively. The only other graphic uniform matroids are the cycles $\UM{n}{n-1} \cong \GMat{\Cycle{n}}$ and the \emph{bonds} (also called \emph{dipoles}) $\UM{n}{1} \cong \GMat{\Bond{n}}$ illustrated in \autoref{fig:cycles-bonds}.
\end{example}
The maximal independent sets $\Bases{M} \subseteq \Indeps{M}$ of a matroid are called its \emph{bases}, and in the case of the cycle matroid the bases are precisely the spanning trees. The (tropical) graph polynomial therefore generalizes naturally to the (tropical) \emph{matroid polynomial}\footnote{%
	Graph polynomials can also be interpreted as \emph{configuration polynomials} \cite{Patterson:SingularStructure,DenhamSchulzeWalther:MatroidConf}. Applied to matroids, these polynomials are typically not unique and different from \eqref{eq:matroid-psi}; they agree only for regular matroids. It not clear if a sensible Hepp bound can be defined for configuration polynomials.
}
\begin{equation}
	\PsiPol_M \defas \sum_{T \in \Bases{M}} \prod_{e\notin T} \SP_e
	\quad\text{and}\quad
	\PsiTrop_M \defas \max_{T \in \Bases{M}} \prod_{e\notin T} \SP_e,
	\label{eq:matroid-psi}%
\end{equation}
and the Mellin integral \eqref{eq:hepp-mellin} may thus be considered for an arbitrary matroid. We can also extend the combinatorial \autoref{def:multi-hepp-from-sectors} to all matroids: The \emph{rank} of a submatroid $\gamma \subseteq \EG{M}$ is the maximal size of an independent set contained in it,
\begin{equation*}
	\rank \gamma = \max_{F \in \Indeps{M}, F \subseteq \gamma} \abs{F}.
\end{equation*}
The surplus of edges is the \emph{corank} $\loops{\gamma} \defas \abs{\EG{\gamma}}-\rank{\gamma}$, and in the case of a graphic matroid it is precisely the loop number \eqref{eq:euler}.
Using the corank, the superficial degree of convergence \eqref{eq:sdc} defines a linear function $\sdc{\gamma}$ for each submatroid, and so \eqref{eq:multi-hepp-from-sectors} defines the Hepp bound for all matroids.

\begin{example}\label{ex:hepp-uniform}%
	The uniform matroid $M=\UM{n}{r}$ with rank $0 < r < n$ has $\ell=n-r$ loops, such that $\frac{\Dim}{2}=\frac{n}{\ell}$ for unit indices $\ind_1=\cdots=\ind_n=1$.
	Every subset $\gamma\subset \UM{n}{r}$ with $k=\abs{\gamma}\leq r$ elements has rank $k$, and for $k>r$, the rank of $\gamma$ is $r$. Every permutation of the edges therefore produces the same sequence of superficial degrees of convergence,
	\begin{equation*}
		\sdc{M^{\sigma}_k} = \begin{cases}
			k & \text{for $1\leq k \leq r$ and} \\
%		\sdc{\gamma_{k}} = 
			k-(k-r)\frac{n}{\ell} = 
			(n-k)\frac{r}{\ell} & \text{for $r< k <n$.} \\
		\end{cases}
	\end{equation*}
	Therefore the Hepp bound of the uniform matroid with unit indices is given by
	\begin{equation}
		\Hepp{\UM{n}{r}}
		= 
		\frac{n!}{(r-1)!\ell!}
		\left( \frac{\ell}{r} \right)^{\ell}.
		\label{eq:hepp-uniform}%
	\end{equation}
	This reproduces the cycles $\Hepp{\Cycle{n}} = \Hepp{\UM{n}{n-1}} = n$ from \autoref{ex:hepp-cycle-1}. For the smallest non-graphic matroid we get $\Hepp{\UM{4}{2}} = 12$.
\end{example}
\begin{remark}\label{rem:um-minimal}
	The uniform matroid is the unique minimizer of the Hepp bound among all matroids of fixed rank and size: If $M$ has $n$ elements and rank $r$, then $\Hepp{M} \geq \Hepp{\UM{n}{r}}$, with equality if and only if $M\cong \UM{n}{r}$. This follows from $\Bases{M} \subseteq \Bases{\UM{n}{r}}$ and \eqref{eq:matroid-psi}.
\end{remark}
Almost all results in this paper apply to arbitrary matroids; in fact, the only exception are the completion and twist symmetries in \autoref{sec:symmetries}, which we only define for graphs.
In particular, the compatibility of the Mellin integral and the combinatorial definition as stated in \autoref{cor:convergence-domain} holds for arbitrary matroids.
This hinges on the fact that the greedy algorithm from \autoref{lem:Kruskal} works for arbitrary matroids \cite{Edmonds:Greedy}.

We find it convenient, however, to use graph-inspired notation. We refer to the elements $e\in\EG{M}$ of the ground set as \emph{edges} and denote bases by $T\in \Bases{M}$. We also write $e\in M$ for $e\in\EG{M}$ and more generally we denote submatroids as $\gamma \subseteq M$. The number of edges is $N=\abs{M}=\abs{\EG{M}}$, and we write the group of permutations of the edges as $\perms{M}$.

\subsection{Convergence}
\begin{definition}\label{def:direct-sum}
	The \emph{direct sum} of two matroids $A$ and $B$ is the matroid $M=A \oplus B$ on the disjoint union $\EG{M} = \EG{A} \sqcup \EG{B}$ such that a subset $\gamma\subseteq M$ is independent precisely when $\gamma \cap A$ and $\gamma\cap B$ are independent in $A$ and $B$, respectively.
	A matroid $M$ is \emph{disconnected} if it can be written as a direct sum of two non-empty, proper submatroids. If no such decomposition exists, $M$ is \emph{connected}.
\end{definition}
For example, note that the uniform matroid $\UM{n}{n} \cong (\UM{1}{1})^{\oplus n}$ of a forest is a direct sum of $n$ copies of the single edge $\UM{1}{1} = \GMat{\gEdge}$. 
Similarly, we have $\UM{n}{0} \cong (\UM{1}{0})^{\oplus n}$ for a union of self-loops $\UM{1}{0} \cong \GMat{\Graph[0.3]{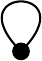}}$.
Both matroids $\UM{n}{0}$ and $\UM{n}{n}$ are therefore disconnected when $n\geq 2$. All other uniform matroids $\UM{n}{r}$ ($0<r<n$) are connected.
\begin{lemma}\label{lem:convdom-nonempty}
	The convergence domain $\ConvDom$ of the Mellin integral \eqref{eq:hepp-mellin} for a matroid, given by \eqref{eq:convergence-domain}, is non-empty precisely when the matroid is connected.
\end{lemma}
\begin{proof}
	If $M$ is disconnected, let $M\cong A\oplus B$ with non-empty $A,B\subsetneq M$. Since $\loops{M} = \loops{A} + \loops{B}$, we note $0=\sdc{M} = \sdc{A} + \sdc{B}$, so at least one of $\sdc{A}$ and $\sdc{B}$ is not positive. This implies $\ConvDom=\emptyset$.
	Now assume that $M$ is connected, and consider the vector
	\begin{equation*}
		\vec{o} \defas \sum_{T \in \Bases{M}} \uv{T^c}
		= \sum_{i \in M} \uv{i} \cdot \abs{\setexp{T\in\Bases{M}}{i \notin T}}.
	\end{equation*}
	Its entries sum to $\abs{\Bases{M}}\loops{M}$ because $\abs{T^c}=\loops{M}$, and so $\sdc[\vec{o}]{M}=0$ in $\Dim=2\abs{\Bases{M}}$ dimensions.
	For a subset $\gamma\subseteq M$, recall that $\max_T \abs{\gamma \cap T} = \rank{\gamma} = \abs{\gamma} - \loops{\gamma}$. Hence
	\begin{equation*}
		\sdc[\vec{o}]{\gamma}
%		= \sum_{i\in \gamma} o_i - \frac{\Dim}{2} \loops{\gamma}
		= \sum_{T\in\Bases{M}} \abs{\gamma \setminus T} - \abs{\Bases{M}} \loops{\gamma}
		= \sum_{T\in\Bases{M}} \left( \rk{\gamma} - \abs{\gamma \cap T} \right)
		\geq 0
	\end{equation*}
	and equality holds only if $\abs{\gamma \cap T}= \rank{\gamma}$ for every basis $T$. But in this situation we get $\abs{\gamma^c \cap T} = \abs{T} - \abs{\gamma \cap T} = \rank{M} - \rank{\gamma}$ for all bases, such that $\rank{\gamma^c} = \rank{M} - \rank{\gamma}$, which implies that $M = \gamma \oplus \gamma^c$. Since $M$ is connected, this is impossible unless $\gamma=\emptyset$ or $\gamma=M$. We conclude that we have the strict inequality $\sdc[\vec{o}]{\gamma}>0$ for all $\emptyset \neq \gamma \subsetneq M$.
\end{proof}
The connectedness of a graph $G$ is \emph{not} the same as connectedness of the cycle matroid. For instance, adding an isolated vertex disconnects a graph, but does not change $\GMat{G}$; and a tree with $\geq2$ edges is a connected graph with disconnected cycle matroid.
\begin{definition}\label{def:separation}
	A \emph{separation} of a graph $G$ is a partition $\EG{G}=A \sqcup B$ of its edges into two non-empty sets, which meet in at most one vertex (see \autoref{fig:separations}). We call $G$ \emph{separable} if a separation exists; otherwise, we say $G$ is \emph{nonseparable}. A graph is \emph{biconnected} if it is connected (as a graph) and still remains connected after deleting any vertex.
\end{definition}
\begin{figure}
	\centering
	\includegraphics[scale=0.55]{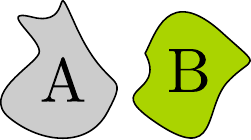}\qquad
	\includegraphics[scale=0.55]{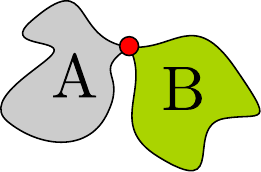}\qquad
	\includegraphics[scale=0.55]{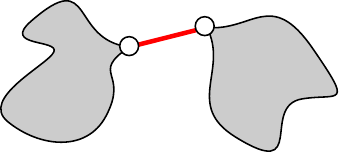}\qquad
	\includegraphics[scale=0.55]{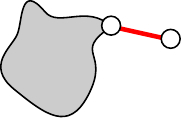}\qquad
	\includegraphics[scale=0.55]{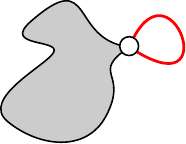}
	\caption{A separation can be disconnected in the graph sense (far left) or meet in a single vertex (called articulation point, left). Special instances are bridges (centre) including the case $B\cong \gEdge$ (right), and self-loops $B\cong \Graph[0.3]{1rose}$ (far right).}%
	\label{fig:separations}%
\end{figure}
A graph $G$ is nonseparable if and only if it is either a graph with at most one edge, or a union of isolated vertices and one biconnected component without self-loops. This characterizes graphs with connected cycle matroids \cite[Proposition~4.1.7]{Oxley:MatroidTheory}:
\begin{lemma}\label{lem:cycle-matroid-connected}
	The cycle matroid $\GMat{G}$ is disconnected if and only if $G$ is separable.
\end{lemma}
For physical applications, we are therefore only interested in biconnected graphs with $N\geq 2$ edges and no self-loops. Such graphs are necessarily $2$-edge connected (bridgeless), referred to as ``one-particle irreducible'' (1PI) in field theory (see \autoref{sec:1pi-flags}).

\subsection{Zeroes and shuffles}
If a matroid $M$ is connected, \autoref{lem:convdom-nonempty} shows that there exist points $\vec{o}\in\ConvDom\neq \emptyset$ for which the integral \eqref{eq:hepp-mellin} converges (\autoref{cor:convergence-domain}). Since the integrand is positive, the corresponding values of the Hepp bound $\Hepp{M,\vec{o}}>0$ are also positive. In particular, the rational function $\Hepp{M,\vec{\ind}}$ from \eqref{eq:multi-hepp-from-sectors} is non-zero.

For disconnected $M$, the Mellin integral does not make sense ($\ConvDom=\emptyset$). In this case, the Hepp bound is the zero function:
\begin{theorem}\label{prop:hepp-disconnected}
	The rational function $\Hepp{M,\vec{\ind}}$ defined in \eqref{eq:multi-hepp-from-sectors} is identically zero on the space $\set{\sdc[\vec{\ind}]{G}=0}$ if and only if the matroid $M$ is disconnected.
\end{theorem}
\begin{corollary}\label{cor:hepp-zero-forest}
	The Hepp bound of a loopless matroid (forest) is constant: For a single edge, $\Hepp{M,\ind_1}=1$, and $\Hepp{M,\vec{\ind}}=0$ for $\abs{M} \geq 2$ edges. This generalizes \autoref{ex:2-forest}.
\end{corollary}
The vanishing $\Hepp{M,\vec{\ind}}=0$ of \eqref{eq:hepp-cycle} in zero dimensions is thus a general fact:
\begin{corollary}\label{cor:hepp(d=0)=0}
	If $M$ has at least two edges, then its Hepp bound vanishes at $\Dim=0$.
\end{corollary}
\begin{proof}
	In zero dimensions, $\sdc{\gamma} = \sum_{e\in \gamma} \ind_e$ is blind to the structure of the graph and the same as if $M$ were a forest.
\end{proof}
To prove \autoref{prop:hepp-disconnected}, we exploit a property of the rational functions
\begin{equation}
	\dChar(s_1,\ldots,s_N) \defas \frac{1}{s_1(s_1+s_2)\cdots(s_1+\cdots+s_{N-1})}
	\in \Q(s_1,\ldots,s_N)
	\label{eq:dChar}%
\end{equation}
that furnish the summand of \eqref{eq:multi-hepp-from-sectors}: If $\Dsdc{\sigma}=\buch{\Dsdc[1]{\sigma},\ldots,\Dsdc[N]{\sigma}}$ denotes the increments
\begin{equation*}
	\Dsdc[k]{\sigma}
	\defas \sdc{M^{\sigma}_k} - \sdc{M^{\sigma}_{k-1}}
	= \begin{cases}
		\ind_{\sigma(k)}
		& \text{if $\loops{M^{\sigma}_k} = \loops{M^{\sigma}_{k-1}}$ and} \\
		\ind_{\sigma(k)} - \frac{\Dim}{2}
		& \text{if $\loops{M^{\sigma}_k} = 1+\loops{M^{\sigma}_{k-1}}$} \\
	\end{cases}
\end{equation*}
of the superficial degree of convergence, then the Hepp bound is precisely $\sum_{\sigma} \dChar(\Dsdc{\sigma})$. This sum vanishes for the $2$-forest in \autoref{ex:2-forest} due to $s_1+s_2=0$ and the factorization
\begin{equation*}
	\dChar(s_1,s_2) + \dChar(s_2,s_1)
	= \frac{1}{s_1} + \frac{1}{s_2}
	= \frac{s_1+s_2}{s_1 s_2}.
	\tag{$\ast$}%
	\label{eq:dChar-2}%
\end{equation*}
To state and generalize such identities, it is convenient to extend \eqref{eq:dChar} linearly and to view it as a function $\dChar\colon \Z\Tens{S} \longrightarrow \Q(S)$ on the space of all finite linear combinations
\begin{equation*}
	\Z\Tens{S}
	= \Z \oplus \bigoplus_{k \geq 1} \bigoplus_{s_1,\ldots,s_k \in S} \Z \buch{s_1,\ldots,s_k}
\end{equation*}
of words $\buch{s_1,\ldots,s_k}$ in the letters $S = \{\ind_e,\ind_e-\frac{\Dim}{2}\colon 1\leq e \leq N\}$. We set $\dChar(\buch{})\defas 0$ for the empty word $k=0$. The left side of \eqref{eq:dChar-2} can now be written as $\dChar(\buch{s_1,s_2} + \buch{s_2,s_1})$.
\begin{definition}
	The $(n,m)$-shuffles $\perms{n,m}$ are those permutations $\sigma \in \perms{n+m}$ that maintain the order among the first $n$ elements and also among the last $m$ elements such that
%	\begin{equation*}
		$\sigma^{-1}(1)<\cdots<\sigma^{-1}(n)$ and
%		\quad\text{and}\quad
		$\sigma^{-1}(n+1)<\cdots<\sigma^{-1}(n+m)$.
%	\end{equation*}
	The shuffle product of two words is $\buch{s_1,\ldots,s_n} \shuffle \buch{s_{n+1},\ldots,s_{n+m}} = \sum_{\sigma \in \perms{n,m}} \buch{s_{\sigma(1)},\ldots,s_{\sigma(n+m)}}$.
\end{definition}
\begin{lemma}
	\label{lem:dChar=0}%
	For a word $w=\buch{s_1,\ldots,s_k}$, let $\abs{w} \defas s_1+\cdots+s_k$ denote the sum of its letters. If $v$ and $w$ are two non-empty words with $\abs{v}+\abs{w}=0$, then
	$ \dChar(v \shuffle w) = 0$.
\end{lemma}
%For example, $\buch{s_1} \shuffle \buch{s_2,s_3} = \buch{s_1,s_2,s_3} + \buch{s_2,s_1,s_3} + \buch{s_2,s_3,s_1}$, and we can state \eqref{eq:dChar-2} as $\dChar(\buch{s_1} \shuffle \buch{s_2})=(s_1+s_2)/(s_1 s_2)$.
\begin{proof}[Proof of \autoref{prop:hepp-disconnected}]
The loop number of a direct sum $M=A\oplus B$ is additive: For any submatroid $\gamma \subseteq M$, we have
	$\loops{\gamma} = \loops{\gamma \cap A} + \loops{\gamma \cap B}$
and therefore also
\begin{equation*}
	\sdc{\gamma}  = \sdc{\gamma \cap A} + \sdc{\gamma \cap B}.
\end{equation*}
So if the edge $\sigma(k)$ of a permutation $\sigma$ of $M$ belongs to $A$, then the increment $\Dsdc[k]{\sigma}$ depends only on the set $M^{\sigma}_k \cap A$. Let  $\set{i_1<\cdots<i_n} = \sigma^{-1}(A)$ denote the places where $A$ appears in $\sigma$, and write $\alpha=(\sigma(i_1),\ldots,\sigma(i_n)) \in \perms{A}$ for the total order induced on $A$. In the same way, $\sigma$ determines a permutation $\beta=(\sigma(j_1),\ldots,\sigma(j_m)) \in \perms{B}$. Then
	\begin{equation*}
		\Dsdc[i_k]{\sigma} = \Dsdc[k]{\alpha}
		\quad\text{(for $1\leq k \leq n$)}
		\quad\text{and}\quad
		\Dsdc[j_k]{\sigma} = \Dsdc[k]{\beta}
		\quad\text{(for $1\leq k \leq m$)}
	\end{equation*}
	show that the increment word $\Dsdc{\sigma}$ is an $(n,m)$-shuffle of the increments $\Dsdc{\alpha}$ and $\Dsdc{\beta}$. Summing over all $\alpha$ and $\beta$, and applying $\dChar$, we conclude that
	\begin{equation*}
		\Hepp{M,\vec{\ind}}
		=
		\sum_{\sigma \in \perms{M}} \dChar(\Dsdc{\sigma})
		=
%		\sum_{\substack{\alpha \in \perms{A} \\ \beta \in \perms{B}}}
		\sum_{\alpha \in \perms{A}}
		\sum_{\beta \in \perms{B}} 
		\dChar\Big(
			\Dsdc{\alpha} \shuffle \Dsdc{\beta}
		\Big)
%		=
%		\sdc{M}
%		\sum_{\substack{\alpha \in \perms{A} \\ \beta \in \perms{B}}}
%		\Char(\Dsdc{\alpha}) \Char(\Dsdc{\beta})
		= 0
	\end{equation*}
	due to $\sdc{M}=0$ and \autoref{lem:dChar=0}.
	This shows that disconnectedness is sufficient to ensure $\Hepp{M,\vec{\ind}}=0$. For connected $M$, however, $\Hepp{M,\vec{\ind}}$ cannot be identically zero, because it takes positive values on $\ConvDom$, which is a non-empty set due to \autoref{lem:convdom-nonempty}.
\end{proof}
Observe that $\dChar(s_1,\ldots,s_N)$ does not depend on the last letter $s_N$.
We define the linear map $\Char\colon \Z\Tens{S} \longrightarrow \Q(S)$ by adding one more denominator to \eqref{eq:dChar},
\begin{equation}
	\Char(s_1,\ldots,s_k) \defas \frac{1}{s_1(s_1+s_2)\cdots(s_1+\cdots+s_k)}
	\in \Q(s_1,\ldots,s_k)
	\label{eq:Char}%
\end{equation}
for all $k \geq 1$ and setting $\Char(\buch{})\defas 1$ for the empty word. We think of $\dChar$ as the residue of $\Char$ when all letters sum to zero, and we will frequently use the relations
\begin{equation*}
	\dChar(s_1,\ldots,s_k)
	= \Char(s_1,\ldots,s_{k-1})
	= (s_1+\cdots+s_k) \Char(s_1,\ldots,s_k).
\end{equation*}
They translate \eqref{eq:dChar-2} into
$
	\Char(\buch{s_1} \shuffle \buch{s_2})
	= \frac{1}{s_1(s_1+s_2)} + \frac{1}{s_2(s_1+s_2)}
	= \frac{1}{s_1 s_2}
	= \Char(s_1) \Char(s_2)
$,
and the generalization of this identity to all shuffle products will be very useful.
\begin{proposition}\label{lem:Char-shuffle}%
	The map $\Char$ is multiplicative: For arbitrary words $a$ and $b$, we have
	\begin{equation}
		\Char(a \shuffle b)
		= \Char(a)  \Char(b).
		\label{eq:Char-shuffle}%
	\end{equation}%
\end{proposition}
\begin{proof}
	The claim is trivial when $a$ or $b$ are the empty word, so we proceed by induction on the lengths of the words. For letters $\alpha$ and $\beta$, the shuffle product solves the recursion
	$ a\alpha \shuffle b\beta = (a\alpha \shuffle b)\beta + (a\shuffle b\beta)\alpha$, because the final letter must be either $\alpha$ or $\beta$. Thus,
	% for the words $a\alpha$ and $b\beta$ with total length $n+m+2$, we have
	\begin{equation*}
		\Char(a\alpha \shuffle b\beta)
		= \frac{\Char(a \shuffle b\beta) + \Char(a\alpha \shuffle b)}{\abs{a}+\alpha+\abs{b}+\beta}
		= \frac{\Char(a)  \Char(b\beta) + \Char(a\alpha) \Char(b)}{\abs{a}+\alpha+\abs{b}+\beta}
	\end{equation*}
	where $\abs{a}$ denotes the sum of all letters in $a$. The second step invokes the claim for shorter words than on the left.
	Now expand $\Char(a\alpha)=(\abs{a}+\alpha)\Char(a)$ and $\Char(b\beta)=(\abs{b}+\beta)\Char(\beta)$.
\end{proof}
\begin{proof}[Proof of \autoref{lem:dChar=0}]
	The general identity $\dChar(a\shuffle b) = (\abs{a}+\abs{b}) \Char(a)\Char(b)$ of rational functions in the letters of $a$ and $b$ follows from \eqref{eq:Char-shuffle}. Take the limit $\abs{a}+\abs{b}\rightarrow 0$.
\end{proof}
\begin{remark}
	The property \eqref{eq:Char-shuffle} is called \emph{symmetral} in the language of moulds \cite[Example~3.2 and Section~A.3]{ChapotonHivertNovelliThibon:OpMould}, and the proof above was given in \cite[Lemme~II.38]{Cresson:CalculMoulien}. Parke--Taylor factors fulfil a closely related identity \cite[Equation~(3.15)]{AHBCPT:onMHVbeyond} underlying the Kleiss--Kuijf relations \cite{KleissKuijf:5jet,DelDucaDixonMaltoni:KKproof}, see \cite[Section~4.1]{MafraSchlotterer:BGandBCJ}.
	The special case $\Char(\buch{s_1}\shuffle\cdots\shuffle\buch{s_n})=1/(s_1\cdots s_n)$ is used in the factorization of infrared divergences, see \cite[Equation~(6.75)]{PeskinSchroeder:QFT}.
\end{remark}
As a further application of the multiplicativity of $\Char$, we compute the Hepp bound of all uniform matroids for arbitrary indices. This generalizes \autoref{ex:hepp-uniform} and \eqref{eq:hepp-cycle}.
Let
\begin{equation*}
	\apode \buch{s_1,\ldots,s_n} \defas (-1)^n \buch{s_n,\ldots,s_1}
%	\label{eq:apode}%
\end{equation*}
denote the \emph{antipode} of the shuffle algebra, with $\apode(\apode w) = w$ and $\apode(v\shuffle w) = (\apode v)(\apode w)$.
\begin{lemma}\label{lem:dChar-apode}%
	When the letters of a word $w$ sum to zero, $\abs{w}=0$, then $\dChar(w) = -\dChar(\apode w)$.
	More explicitly, under the constraint that $s_1+\cdots+s_n=0$, we have the identity
	\begin{equation}
		\dChar(s_1,\ldots,s_n)
		=-(-1)^n \dChar(s_n,\ldots,s_1)
		=(-1)^{n-1} \Char(s_n,\ldots,s_2)
		.
		\label{eq:dChar-apode}%
	\end{equation}
\end{lemma}
\begin{proof}
	The recursion
	$\apode w = -w -\sum_{k=1}^{n-1} \buch{ s_1,\ldots, s_k} \shuffle \apode \buch{s_{k+1},\ldots, s_n}$
	for the antipode is well known.
	The shuffle products cancel due to \autoref{lem:dChar=0}.
\end{proof}
\begin{proposition}\label{lem:Hepp-Unr}
	The Hepp bound of a uniform matroid $\UM{n}{r}$ with rank $0<r<n$ can be computed as a sum over all subsets of size $r$: Let $\ind_{\gamma}\defas \sum\limits_{e\in \gamma} \ind_e$ and $\ind^{\gamma} \defas \prod\limits_{e\in \gamma}\ind_e$, then
	\begin{equation}
		\Hepp{\UM{n}{r},\vec{\ind}} =
		\sum_{\substack{\gamma \subset \set{1,\ldots,n}\\\abs{\gamma}=r}}
		\frac{\ind_{\gamma}}{\ind^{\gamma} \prod_{e\notin {\gamma}} (\frac{\Dim}{2}-\ind_e)}
		.
		\label{eq:hepp-Unr}%
	\end{equation}
\end{proposition}
\begin{proof}
	Recall that every flag $M^{\sigma}_{\bullet}$ has the same rank sequence, and the increments are
	\begin{equation*}
		\Dsdc{\sigma}
		= \buch{ \ind_{\sigma(1)},\ldots,\ind_{\sigma(r)},
			\ind_{\sigma(r+1)}-\tfrac{\Dim}{2},
			\ldots,
			\ind_{\sigma(n)}-\tfrac{\Dim}{2}
		}.
	\end{equation*}
	Consider any submatroid $\gamma$ of $\UM{n}{r}$ with $r$ elements; note that $\gamma \cong \UM{r}{r}$ is a forest with $\sdc{\gamma}=\ind_{\gamma}$.
	The flags through $\gamma=M^{\sigma}_{r}=\set{\sigma(1),\ldots,\sigma(r)}$ are in bijection with pairs of permutations $(\sigma|_{\gamma},\sigma|_{M\setminus\gamma})$ of $\gamma$ and its complement. The sum of all these pairs adds
	\begin{equation*}
		\sum_{\sigma\colon M^{\sigma}_r=\gamma} \dChar(\Dsdc{\sigma})
		=
		\Char\Big(\bigshuffle_{e\in \gamma} \buch{\ind_{e}} \Big) 
		\cdot \ind_{\gamma} \cdot
		\dChar\Big( \buch{\ind_{\gamma}} \Big[ \bigshuffle_{e \notin \gamma} \buch{\ind_e-\tfrac{\Dim}{2} } \Big] \Big)
%		= \frac{\ind_\gamma}{\ind^{\gamma}}
%		(-1)^{n-r}
%		\Char\bigg(\bigshuffle_{e \notin \gamma} (\ind_e-\tfrac{\Dim}{2} ) \bigg)
	\end{equation*}
	to the Hepp bound. The first term on the right is $1/\ind^{\gamma}$ by \eqref{eq:Char-shuffle}. With \eqref{eq:dChar-apode} we rewrite the last term as $(-1)^{n-r} \Char(\shuffle_{e \notin \gamma} \buch{\ind_e-\frac{\Dim}{2}})$ and apply  the multiplicativity once more.
\end{proof}
In the special case of cycles $\UM{n}{n-1} \cong \GMat{\Cycle{n}}$, the sum goes over edge complements $\gamma = \set{1,\ldots,n}\setminus\set{e}$. Since $0=\sdc{\Cycle{n}}=\ind_{\gamma} + \ind_e - \Dim/2$, the summand simplifies to $1/\ind^{\gamma}$ and we recover \eqref{eq:hepp-cycle}. Analogously, the Hepp bound of a bond $\UM{n}{1} \cong \GMat{\Bond{n}}$ becomes
	\begin{equation}
		\Hepp{\Bond{n},\vec{\ind}} = \frac{\Dim/2}{(\Dim/2-\ind_1)\cdots(\Dim/2-\ind_n)}.
		\label{eq:hepp-bond}%
	\end{equation}

\subsection{Poles and factorizations}

\begin{lemma}\label{lem:simple-poles}
	The singularities of the Hepp bound of a matroid $M$ are a subset of the hyperplanes $\set{\sdc{\gamma}=0}$ where $\emptyset \neq \gamma \subsetneq M$. All poles are simple.
\end{lemma}
\begin{proof}
	If $M$ is a forest, then there are no poles due to \autoref{cor:hepp-zero-forest}, so let $\loops{M}>0$.
	By \autoref{def:multi-hepp-from-sectors} the first claim is obvious.
	For all summands $\sigma$ of \eqref{eq:multi-hepp-from-sectors}, the linear map
	\begin{equation*}
		\C^N \ni (\ind_1,\ldots,\ind_N)
		\mapsto
		(\sdc{M^{\sigma}_1},\ldots,\sdc{M^{\sigma}_{N-1}},\Dim) \in \C^N
	\end{equation*}
	has inverse $\ind_{\sigma(k)}=\sdc{M^{\sigma}_k}-\sdc{M^{\sigma}_{k-1}}+\frac{\Dim}{2}(\loops{M^{\sigma}_k}-\loops{M^{\sigma}_{k-1}})$, where $\sdc{M^{\sigma}_0}=\loops{M^{\sigma}_0}=0$ for $k=1$ and $\sdc{M^{\sigma}_N}=0$ at $k=N$. All factors in the denominator of the summand $\sigma$ are therefore independent coordinates on $\C^N$.
\end{proof}
Most of these potential poles are actually absent due to cancellations in the sum \eqref{eq:multi-hepp-from-sectors}. The residues can be expressed in terms of sub- and quotient matroids.
\begin{definition}
	The \emph{quotient} (contraction) of a matroid $M$ by a subset $\gamma \subseteq M$ is the matroid $M/\gamma$ on the complement $\EG{M/\gamma} = \EG{M}\setminus \gamma$ such that every $\delta \subseteq \EG{M/\gamma}$ has corank
	\begin{equation}
		\loopsIn{M/\gamma}{\delta} = \loops{\delta \cup \gamma} - \loops{\gamma}.
		\label{eq:loops-quotient}%
	\end{equation}
\end{definition}
\begin{figure}
	\centering
	$ G=\Graph[0.4]{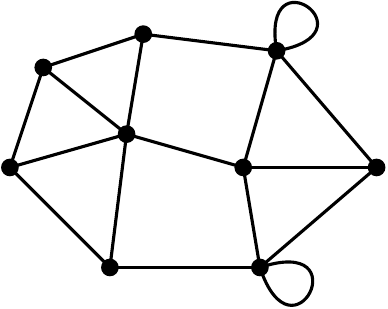}$
	\qquad $\supset$\qquad
	$\gamma=\Graph[0.4]{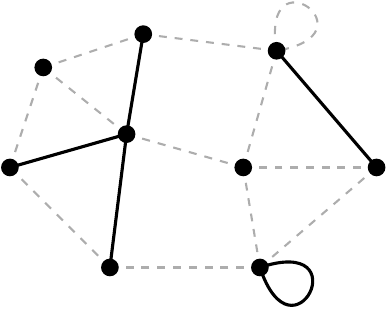}$
	\qquad$\mapsto$\qquad
	$ G/\gamma = \Graph[0.45]{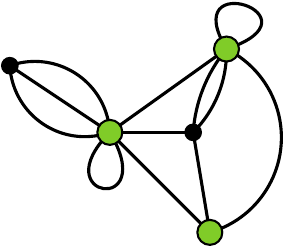}$
	\caption{A subgraph $\gamma\subset G$ with three non-trivial connected components. In the quotient, these components correspond to the highlighted vertices.}%
	\label{fig:quotient}%
\end{figure}
\begin{example}
	Given a graph $G$ and a subgraph $\gamma$, the \emph{quotient graph} $G/\gamma$ is obtained by contracting each connected component of $\gamma$ to a single vertex (see \autoref{fig:quotient}). This construction computes the matroid quotient: $\GMat{G}/\GMat{\gamma} \cong \GMat{G/\gamma}$.
\end{example}
\begin{proposition}\label{lem:hepp-residue}
	Given a connected matroid $M$ and a submatroid $\emptyset \neq \gamma \subsetneq M$, let $\vec{\ind}_{\gamma}=(\ind_e)_{e \in \gamma}$ denote only those indices that belong to $\gamma$, and write $\vec{\ind}_{\gamma^c}=(\ind_e)_{e \notin \gamma}$ for the rest such that $\vec{\ind}=(\vec{\ind}_{\gamma},\vec{\ind}_{\gamma^c})$.
	Then the residue of the Hepp bound at $\sdc{\gamma}=0$ factorizes:
	\begin{equation}
%		\Res_{\sdc{\gamma}=0} \Hepp{M,\vec{\ind}} = \Hepp{\gamma,\vec{\ind}'} \Hepp{M/\gamma,\vec{\ind}''}.
%		\Res_{\sdc{\gamma}=0} \Hepp{M,\vec{\ind},\breve{\ind}} = \Hepp{\gamma,\vec{\ind}} \Hepp{M/\gamma,\breve{\ind}}.
		\Res_{\sdc{\gamma}=0} \Hepp{M,\vec{\ind}_{\gamma},\vec{\ind}_{\gamma^c}} = \Hepp{\gamma,\vec{\ind}_{\gamma}} \Hepp{M/\gamma,\vec{\ind}_{\gamma^c}}.
		\label{eq:hepp-residue}%
	\end{equation}
\end{proposition}
\begin{proof}
	Note that $\abs{M}\geq 2$, so we must have $\loops{M} \geq 1$ for $M$ to be connected. The linear function $\loops{M}\sdc{\gamma}=\vec{\ind} \cdot \vec{c}$ of $\vec{\ind}$ has exactly two different coefficients, $c_e=\loops{M}-\loops{\gamma}$ for $e \in \gamma$ and $c_e = -\loops{\gamma}$ when $e\notin \gamma$. Only one other submatroid yields the same partition, namely the complement $\gamma^c=M\setminus \gamma$. But $\sdc{\gamma}$ and $\sdc{\gamma^c}$ are linearly independent, since
	\begin{equation*}
		\det 
		\begin{pmatrix}
			\loops{M}-\loops{\gamma} & -\loops{\gamma} \\
			-\loops{\gamma^c} & \loops{M} -\loops{\gamma^c} \\
		\end{pmatrix}
		= \loops{M} \big[ \loops{M}-\loops{\gamma}-\loops{\gamma^c} \big]
		\neq 0
	\end{equation*}
	because $\loops{M} = \loops{\gamma} + \loops{\gamma^c}$ would imply that $M\cong \gamma \oplus \gamma^c$ is disconnected. It follows that a summand $\sigma$ in \eqref{eq:multi-hepp-from-sectors} is singular on $\sdc{\gamma}=0$ only if its flag goes through $\gamma = M^{\sigma}_k$ at $k=\abs{\gamma}$. So $\alpha \defas \sigma|_{\set{1,\ldots,k}}$ is a permutation of $\gamma$, and we can view $\beta\defas\sigma|_{\set{k+1,\ldots,N}}$ as a permutation of the quotient $Q\defas M/\gamma$. We can therefore write
	\begin{equation*}
		\Res_{\sdc{\gamma}=0} \Hepp{M,\vec{\ind}}
		= \Res_{\sdc{\gamma}=0}
		\bigg(
			\sum_{\alpha \in \perms{\gamma}} \dChar(\Dsdc{\alpha}) 
		\bigg)
		\frac{1}{\sdc{\gamma}}
		\bigg( 
			\sum_{\beta \in \perms{Q}}
			\prod_{i=1}^{N-k-1}
			\frac{1}{\sdc{\gamma} + \sdc{Q^{\beta}_i}}
		\bigg)
	\end{equation*}
	because we have $\sdc{M^{\sigma}_{k+i}}=\sdc{\gamma} + \sdc{Q^{\beta}_i}$ due to \eqref{eq:loops-quotient}. The sum over $\alpha$ gives $\Hepp{\gamma,\vec{\ind}_{\gamma}}$, and similarly we get $\Hepp{Q,\vec{\ind}_Q}$ from the sum over $\beta$, since $\sdc{\gamma} = 0$ on the pole.
\end{proof}
\begin{remark}
	Formula \eqref{eq:hepp-residue} is wrong for disconnected matroids. The forest $G=\Graph[0.3]{twosticks}$ has a subgraph $\gamma=\set{1}\cong \gEdge$ with quotient $\set{2}\cong\gEdge$. The right hand side of \eqref{eq:hepp-residue} gives $1$ for the residue of $\Hepp{G,\vec{\ind}}=0$ at $\sdc{\gamma}=\ind_1=0$. This contradiction arises because also the subgraph $\gamma=\set{2}$ gives vanishing $\sdc{\set{2}}=0$ on $\sdc{\gamma}=0$, since $0=\sdc{G}=\ind_1+\ind_2$.
\end{remark}
\begin{corollary}
	If $M$ is a connected matroid with at least two edges, and $e\in M$, then
	\begin{equation}
		\Res_{\ind_e=0} \Hepp{M,\vec{\ind},\ind_e}
		= \Hepp{M/e,\vec{\ind}}
		\quad\text{and}\quad
		\Res_{\ind_e=\Dim/2} \Hepp{M,\vec{\ind},\ind_e}
		= \Hepp{M\setminus e,\vec{\ind}}.
		\label{eq:hepp-res-e}%
	\end{equation}
\end{corollary}
\begin{proof}
	Since $M$ is connected, $e$ is not a self-loop and thus $\sdc{\set{e}} = \ind_e$. This proves the first claim, because $\Hepp{\set{e},\ind_e}=1$. Similarly, $e$ cannot be a bridge, so we must have $\loops{M\setminus e} = \loops{M}-1$ and therefore $\sdc{M\setminus e} = \sdc{M} + \Dim/2 - \ind_e$. Now use $\sdc{M}=0$.
\end{proof}
This yields another proof of one half of \autoref{prop:hepp-disconnected}, namely that, if $M$ is connected, then $\Hepp{M,\vec{\ind}}$ is not the zero function. We use the following fact:
\begin{lemma}[{\cite[Claim~6.5]{Tutte:ConnectivityMatroids}}]
	If a connected matroid $M$ and an edge $e\in M$ are given, then at least one of $M\setminus e$ and $M/e$ is also connected.
\end{lemma}
\begin{corollary}\label{cor:hepp-nonzero}
	If $M$ is a connected matroid with at least one edge, then $\Hepp{M,\vec{\ind}}\neq 0$.
\end{corollary}
\begin{proof}
	The case $\abs{M}=1$ of a single edge is $\Hepp{M,\ind_1} = 1 \neq 0$. We proceed by induction over the number of edges. Suppose $\abs{M} \geq 2$ and pick any $e \in M$. If $ M/e$ is connected, then we know by induction that $\Hepp{M/e,\vec{\ind}} \neq 0$. If $M\setminus e$ is connected, we may similarly assume that $\Hepp{M\setminus e,\vec{\ind}}\neq 0$. In both cases, \eqref{eq:hepp-res-e} shows that the Hepp bound of $M$ cannot be the zero function, because it has a non-vanishing residue.
\end{proof}
\begin{corollary}
	The Hepp bound of a connected matroid $M$ has a pole on the hypersurface $\sdc{\gamma}=0$ if, and only if, both $\gamma$ and its quotient $M/\gamma$ are connected.
\end{corollary}
\begin{proof}
	Apply \autoref{prop:hepp-disconnected} to the right-hand side of \eqref{eq:hepp-residue}.
\end{proof}
Together with \autoref{lem:simple-poles}, this completely characterizes the poles of the Hepp bound:%
\begin{definition}\label{def:singularities}
	Given a connected matroid $M$, a \emph{singularity} of $M$ is a non-empty submatroid $\gamma \subsetneq M$ such that $\gamma$ and $M/\gamma$ are connected. We denote them as the set
	\begin{equation}
		\Sing{M} \defas \setexp{\emptyset \neq \gamma \subsetneq M}{\text{$\gamma$ and $M/\gamma$ are connected}}.
		\label{eq:singularities}%
	\end{equation}
\end{definition}
\begin{corollary}\label{cor:hepp-poles}
	The Hepp bound of a connected matroid $M$ is a non-zero rational function with simple poles, precisely on the hypersurfaces $\sdc{\gamma}=0$ for $\gamma \in \Sing{M}$.
\end{corollary}
\begin{example}\label{ex:uniform-poles}
	All submatroids and quotients of $\UM{n}{r}$ are themselves uniform: If $\gamma \subseteq \UM{n}{r}$ has $k=\abs{\gamma}\leq r$ elements, then $\gamma\cong\UM{k}{k} \cong (\UM{1}{1})^{\oplus k}$; if $k \geq r$, then $\gamma\cong \UM{k}{r}$.
	The respective quotients are $\UM{n}{r}/\UM{k}{k} \cong \UM{n-k}{r-k}$ and $\UM{n}{r}/\UM{k}{r} \cong \UM{n-k}{0} \cong (\UM{1}{0})^{\oplus (n-k)}$.
%	If we assume that $2 \leq r \leq n-2$, this shows that
	So only individual edges $\set{e}$ and their complements $\set{e}^c = \set{1,\ldots,n} \setminus \set{e}$ are singular, such that
	\begin{equation*}
		\Sing{\UM{n}{r}} =
%		\setexp{\set{e}}{1\leq e \leq n} \sqcup \setexp{\set{1,\ldots,n}\setminus\set{e}}{1\leq e \leq n}
		\setexp{\set{e},\set{e}^c}{1\leq e \leq n}
		\quad\text{if we assume $2\leq r\leq n-2$.}
	\end{equation*}
	The $2n$ residues on $\ind_e = 0$ and $\ind_e=\Dim/2$ in \eqref{eq:hepp-res-e} are therefore non-zero, and these poles of $\Hepp{\UM{n}{r},\vec{\ind}}$ align perfectly with the denominators in the formula \eqref{eq:hepp-Unr}. In the case $r=n-1$ of cycles $\UM{n}{n-1}\cong \GMat{\Cycle{n}}$, the edge complements $\Cycle{n}\setminus\set{e} \cong \Path{n}$ are paths and therefore disconnected as matroids, $\set{e}^c \cong \UM{n-1}{n-1} \cong (\UM{1}{1})^{\oplus(n-1)}$. Therefore,
	\begin{equation*}
		\Sing{\UM{n}{n-1}} = \setexp{\set{e}}{1\leq e \leq n}
	\end{equation*}
	shows that $\Hepp{\Cycle{n},\vec{\ind}}$ only has the poles on $\ind_e=0$, as is obvious from \eqref{eq:hepp-cycle}. For a bond $\UM{n}{1}\cong\GMat{\Bond{n}}$, edges have disconnected quotients $\UM{n}{1}/\set{e} \cong  (\UM{1}{0})^{\oplus(n-1)}$ and $\Sing{\UM{n}{1}}$ consists only of the $n$ complements $\set{e}^c$. Indeed, we only see poles at $\ind_e=\Dim/2$ in \eqref{eq:hepp-bond}.
\end{example}
The precise knowledge of the singularities of the Hepp bound also tells us the facets of the convergence cone $\ConvDom$ from \eqref{eq:convergence-domain}. If $\sdc[\vec{\ind}]{\gamma}>0$ for all singular $\gamma \in \Sing{M}$, then the Hepp bound is finite for these indices $\vec{\ind}$. Approaching the boundary $\partial \ConvDom$ where the Mellin integral \eqref{eq:hepp-mellin} diverges therefore implies that $\sdc[\vec{\ind}]{\gamma}\rightarrow 0$ for at least one singular $\gamma$.
\begin{corollary}
	The convergence domain of a connected matroid $M$ is equal to the following intersection of half-spaces, and none of these inequalities is redundant:
	\begin{equation}
		\ConvDom = \bigcap_{\gamma \in \Sing{M}} \setexp{\vec{\ind}\in \R^N}{\sdc[\vec{\ind}]{\gamma}>0}
		\subseteq \R^N.
		\label{eq:conv-minimal}%
	\end{equation}
\end{corollary}
This amounts to a well-known description of matroid polytopes, see \autoref{cor:conv-conv}.

\subsection{Other matroid invariants}
\label{sec:other-invariants}
As explained in \autoref{rem:inverse-Mellin}, we can recover a connected matroid $M$ from $\Hepp{M,\vec{\ind}}$. In principle, every invariant of $M$ can therefore be calculated from its Hepp bound; but in practice it may not be obvious how to achieve this efficiently. It seems worthwhile, then, to identify the aspects of the function $\Hepp{M,\vec{\ind}}$ that are encoded in other matroid invariants, and to exhibit their connection as explicitly as possible.
We merely sketch a glimpse here and limit our discussion to the invariants of Crapo and Derksen.

Recall that the increments of the superficial degree of convergence associate a sum
\begin{equation*}
	\Phi(M)
	\defas
	\sum_{\sigma \in \perms{M}} \Dsdc{\sigma}
%	= \sum_{\sigma \in \perms{E}} \buch{
%		\ind_{\sigma(1)}-\tfrac{\Dim}{2}\loops{M^{\sigma}_1}
%		,\ldots,
%		\ind_{\sigma(k)}-\tfrac{\Dim}{2}\loops{M^{\sigma}_k/M^{\sigma}_{k-1}}
%		,\ldots
%	}
	\quad\in\quad
	\Z \big\langle \ind_e,\ind_e-\tfrac{\Dim}{2}\colon e\in M \big\rangle
\end{equation*}
of words with letters of the form $\ind_e$ and $\ind_e-\Dim/2$ to every matroid. To obtain the Hepp bound, we apply the map $\dChar$ or $\Char$ from \eqref{eq:dChar} and \eqref{eq:Char} to this sum,
\begin{equation*}
	\Hepp{M,\vec{\ind}}
	= \Res_{\sdc{M}=0} \Char(\Phi(M))
	= \dChar(\Phi(M))|_{\sdc{M}=0}
	.
\end{equation*}
If we set all indices to $\ind_e=1$, then the words in $\Phi(M)$ contain only two letters, $\buch{1}$ and $\buch{1-\Dim/2}$. This specializes at $\Dim=2$ to the invariant studied by Derksen \cite{Derksen:SymQSymMatroids},
\begin{equation}
	\Derksen{M}
	\defas 
%	\sum_{\sigma\in\perms{M}}
%	\Delta^{\sigma} \rank{}
%	=
	\sum_{\sigma\in\perms{M}}
	\buch{
		\rank{M^{\sigma}_1},
		\rank{M^{\sigma}_2}-\rank{M^{\sigma}_1},
%		\rank{M^{\sigma}_3}-\rank{M^{\sigma}_2},
		\ldots
%		\rank{M}-\rank{M^{\sigma}_{N-1}}
	}
	\in \Z\Tens{0,1}
	\label{eq:Derksen}%
\end{equation}
which is universal for \emph{valuative} matroid invariants \cite{DerksenFink:ValuativeInvariants} with values in $\Q$. It thus determines several other matroid invariants, like the Tutte polynomial \cite{Crapo:Tutte}, however it cannot distinguish all non-isomorphic matroids \cite[Example~3.5]{Derksen:SymQSymMatroids}. It is thus impossible to reconstruct the full Hepp bound function $\Hepp{M,\vec{\ind}}$ from $\Derksen{M}$, but, whenever defined, we find the special value $\Hepp{M}$ at unit indices $\ind_e=1$.
\begin{example}
	Every order on the uniform matroid $\UM{n}{r}$ yields the same rank sequence:
	\begin{equation*}
		\Derksen{\UM{n}{r}}
		= n!\,\buch{\underbrace{1,\ldots,1}_{r},\underbrace{0,\ldots,0}_{n-r}}.
	\end{equation*}
\end{example}
\begin{example}\label{ex:Derksen-K4}
	Consider the complete graph $K_4$. The first $3$ edges $\gamma=\set{\sigma(1),\sigma(2),\sigma(3)}$ of any permutation $\sigma\in\perms{6}$ either form one of the $4$ triangles $\gamma\cong \Cycle{3}$, or one of the $\abs{\STrees{K_4}}=16$ spanning trees. The corresponding rank increments are $\buch{1,1,0,1,0,0}$ and $\buch{1,1,1,0,0,0}$, respectively. Each of these appears $3!\cdot 3!$ times for each fixed $\gamma$, because the edges of $\gamma$ and its complement may be permuted arbitrarily, and we conclude
	\begin{equation*}
		\Derksen{\Graph[0.25]{w3small}}
%		= 36\cdot 4\,\buch{1,1,0,1,0,0} + 36\cdot 16 \,\buch{1,1,1,0,0,0}
		= 144\,\buch{1,1,0,1,0,0} + 576 \,\buch{1,1,1,0,0,0}.
	\end{equation*}
\end{example}
\begin{lemma}\label{lem:Derksen-Hepp}%
	If the Hepp bound \eqref{eq:multi-hepp-from-sectors} is defined for unit indices, then it can be obtained as $\Hepp{M} = h(\Derksen{M})$ from Derksen's invariant, via a linear map $h$ defined on words as
\begin{equation*}
%	h(r_1,\ldots,r_n)
%	\defas \frac{1}{n r_1-1} \frac{1}{n(r_1+r_2)-2} \cdots \frac{1}{n(r_1+\ldots+r_{n-1})-(n-1)}
%	.
	h(\buch{r_1,\ldots,r_n})
	\defas
	\prod_{k=1}^{n-1} \frac{1}{k-\frac{\Dim}{2} \sum_{1\leq j \leq k}(1-r_j)}
	\quad\text{where}\quad
	\frac{\Dim}{2} \defas \frac{n}{\sum_{k=1}^n(1-r_k)}
	.
\end{equation*}
\end{lemma}
\begin{proof}
	With $r_k = \rank{M^{\sigma}_k} - \rank{M^{\sigma}_{k-1}}$ we get $\loops{M^{\sigma}_k} = k - \rank{M^{\sigma}_k}=k-(r_1+\cdots+r_k)$, such that $\sdc{M^{\sigma}_k} = k-\frac{\Dim}{2}\sum_{j=1}^k (1-r_j)$ in \eqref{eq:multi-hepp-from-sectors}.
\end{proof}
\begin{example}
	From $h(1,1,0,1,0,0)=1/4$ and $h(1,1,1,0,0,0)=1/12$ we infer that $\Hepp{K_4}=h(\Derksen{K_4})= 144/4 + 576/12 = 84$ as claimed in \eqref{eq:ws3-hepp}.
\end{example}
Crapo \cite{Crapo:HigherInvariant} defined a non-negative integer $\Crapo{M} \in \Z_{\geq 0}$ for every matroid $M$ as
\begin{equation}
	\Crapo{M}
	= (-1)^{\rank{M}} \sum_{\gamma \subseteq M} (-1)^{\abs{\gamma}} \rank \gamma.
	\label{eq:Crapo}%
\end{equation}
This is the coefficient of $x$ in the Tutte polynomial $T_M(x,y)$ and it also appears as the first coefficient of Speyer's invariant \cite{Speyer:MatroidKtheory}. Some of its remarkable properties are:
\begin{enumerate}
	\item When $M$ has at least two edges, then $\Crapo{M}=0$ precisely when $M$ is disconnected.
	\item If $M$ has at least two edges and $M^{\dual}$ denotes its dual, then $\Crapo{M^{\dual}} = \Crapo{M}$.
	\item For a $2$-sum (see \autoref{def:2sum}), $\Crapo{A \TwoSum{e}{f} B} = \Crapo{A} \Crapo{B}$ is multiplicative \cite{Brylawski:CombinatorialSeriesParallel}.
\end{enumerate}
We already saw that the Hepp bound shares the same vanishing property, and \autoref{sec:symmetries} proves that it also behaves in the same way for duals and 2-sums. This very close analogy suggests that Crapo's invariant is a special value of the Hepp bound, and indeed it is.
\begin{lemma}
	\label{lem:Hepp-Crapo}%
	The Hepp bound of a connected matroid $M$ on $N\geq 2$ edges vanishes to first order on the hyperplane $\set{\ind_1+\cdots+\ind_N=0}$ where $\Dim=0$. Concretely, assume that $\sdc{\gamma} \rightarrow \sum_{e\in \gamma} \ind_e \neq 0$ stays non-zero in the limit $\Dim\rightarrow 0$, for all non-empty $\gamma\subsetneq M$. Then
	\begin{equation}
		\Hepp{M,\vec{\ind}}
		=
		\frac{\Dim/2}{\ind_1\cdots\ind_N} (-1)^{\loops{M}+1}\Crapo{M}
		+ \asyO{\Dim^2}
		\quad\text{as}\quad
		\Dim\rightarrow 0.
		\label{eq:Hepp-Crapo}%
	\end{equation}
\end{lemma}
\begin{proof}
	Since $\ind_{\gamma} \defas \sum_{e\in \gamma} \ind_e \neq 0$, we may expand $1/\sdc{\gamma} = 1/\ind_{\gamma} + \frac{\Dim}{2} \loops{\gamma} /\ind_{\gamma}^2 + \asyO{\Dim^2}$ for small $\Dim$. The Hepp bound \eqref{eq:multi-hepp-from-sectors} thus becomes
	\begin{equation*}
		\Hepp{M,\vec{\ind}}
		= \sum_{\sigma \in \perms{M}} 
		\dChar(\ind_{\sigma(1)},\ldots,\ind_{\sigma(N)}) \left\{ 
			1
			+
			\frac{\Dim}{2}\sum_{k=1}^{N-1} \frac{\loops{\set{M^{\sigma}_k}}}{\ind_{M^{\sigma}_k}}
		\right\}
		+\asyO{\Dim^2}.
	\end{equation*}
	Because of $\ind_1+\cdots+\ind_N = \frac{\Dim}{2} \loops{M}$ and \eqref{eq:Char-shuffle}, the first summand in braces gives
	\begin{equation*}
%		(\ind_1+\cdots+\ind_N)
		\sum_{\sigma \in \perms{M}} 
		(\ind_{\sigma(1)}+\cdots+\ind_{\sigma(N)})
		\Char(\ind_{\sigma(1)},\ldots,\ind_{\sigma(N)})
		= \frac{\Dim}{2} \loops{M} \Char(\buch{\ind_1} \shuffle \ldots \shuffle \buch{\ind_N})
		= \frac{\Dim}{2} \frac{\loops{M}}{\ind_1\cdots\ind_N}.
	\end{equation*}
	We group the remaining summands in braces by the submatroid $M^{\sigma}_k$. To obtain $M^{\sigma}_k = \gamma$, the first $k=\abs{\gamma}$ elements of $\sigma$ must form a permutation $\tau$ of $\gamma$, and the remaining $N-k$ edges $\rho$ permute the complement $M\setminus \gamma$. All these contributions can thus be written as
	\begin{equation*}
		\frac{\Dim}{2} \loops{\gamma}
		\sum_{\tau \in \perms{\gamma}} \Char(\ind_{\tau(1)},\ldots,\ind_{\tau(k)})
		\sum_{\rho \in \perms{M\setminus\gamma}}
		\dChar(\ind_{\gamma},\ind_{\rho(1)},\ldots,\ind_{\rho(N-k)})
		.
	\end{equation*}
	The sum over $\tau$ is a shuffle product and equals $1/(\prod_{e\in \gamma} \ind_e)$ according to \eqref{eq:Char-shuffle}. For $\tau$, we use the antipode \eqref{eq:dChar-apode} to rewrite the summand as $(-1)^{N-k} \Char(\ind_{\rho(N-k)},\ldots,\ind_{\rho(1)})$ and see a shuffle product again. Collecting all contributions, we obtain
	\begin{equation*}
		\Hepp{M,\vec{\ind}}
		=
		\frac{\Dim/2}{\ind_1\cdots\ind_N}
		\sum_{\emptyset \neq \gamma \subseteq M}
		(-1)^{N-\abs{\gamma}}
		\loops{\gamma}
		+
		\asyO{\Dim^2}.
	\end{equation*}
	With $\rank{\gamma} = \abs{\gamma}-\loops{\gamma}$, we recognize Crapo's definition \eqref{eq:Crapo}.
\end{proof}
\begin{example}\label{ex:Crapo-cycle-bond}
	From \eqref{eq:hepp-cycle} and \eqref{eq:hepp-bond}, we see that cycles and bonds have beta invariant $\Crapo{\GMat{\Cycle{n}}} = \Crapo{\GMat{\Bond{n}}} = 1$. More generally, note that $\Crapo{M}=1$ if and only if $M$ is series--parallel \cite[Proposition~8]{Crapo:HigherInvariant}.
\end{example}
In \autoref{def:Hepp-pos} we introduce a variation $\HeppComp{M,\vec{\ind}}$ of the Hepp bound, that evaluates at $\Dim=0$ precisely to $(-1)^{\rank{M}+1} \Crapo{M}$. We can derive the above facts 1.\ through 3.\ for Crapo's invariant from the corresponding symmetries of the Hepp bound $\HeppComp{M,\vec{\ind}}$.

However, this argument does not apply to completion (see \autoref{rem:crapo-completion}), and Crapo's invariant violates this symmetry. For example, the graphs from \autoref{fig:completion} give the values
\begin{equation}
	\Crapo{ \Graph[0.3]{5Rw} }
	= 4
	\neq
	6
	= \Crapo{ \Graph[0.3]{5Rv} }
	.
	\label{eq:crapo-completion-counterexample}%
\end{equation}
These are computed with the contraction-deletion formula $\Crapo{M} = \Crapo{M/e} + \Crapo{M{\setminus}e}$, which applies whenever $e$ is neither a self-loop nor a bridge \cite[Theorem~I]{Crapo:HigherInvariant}.

\section{Flag formulas}
\label{sec:flags}

The formula \eqref{eq:multi-hepp-from-sectors} has $N!$ summands, one for each flag $\emptyset \neq \gamma_1\subsetneq \cdots \subsetneq \gamma_N=G$ of subgraphs $\gamma_k=G^{\sigma}_k$. This is very inefficient and hides the structure and simplicity of results like \eqref{eq:hepp-cycle}. Below we will partition all flags into families of subsets that are easily summed, and thereby derive expressions for the Hepp bound with much fewer terms. 

In \autoref{sec:1pi-flags} we give a formula summing over flags of bridgeless matroids, which is particularly efficient for small loop number $\ell$ and for example gives \eqref{eq:hepp-cycle} on the nose as a single term. It  yields an algorithm that computes the Hepp bound in $\asyO{N^{\ell+2}}$ steps. Dually, flags of flats are most efficient for small ranks, see \autoref{sec:flat-flags}.

On the level of the integral \eqref{eq:hepp}, the flag formulas correspond to a decomposition of the integration domain into fewer sectors, that each combine many individual Hepp sectors. In \autoref{sec:phi4} we apply these sectors to the period itself to get improved bounds.

\subsection{Bridges and ears}
\label{sec:1pi-flags}
\begin{definition}
	A \emph{circuit} $C\subseteq M$ of a matroid is a minimal dependent set, and we write $\Circuits{M}$ for the set of all circuits.
	An edge $e\in\EG{M}$ is called a \emph{bridge} (also \emph{coloop} and \emph{isthmus}) if it is not contained in any circuit; equivalently, if it is contained in every basis.
	We say that a matroid $M$ is \emph{bridgeless} (or \OnePI) when it does not have any bridges.
\end{definition}
Each bridge $e$ corresponds to a direct summand $M|_e\cong \UM{1}{1}$ of $M=(M\setminus e) \oplus M|_e$.
Connected matroids are thus always bridgeless, except for $M\cong\UM{1}{1}\cong\GMat{\gEdge}$. 
Our use of `{\OnePI}' as a synonym for bridgeless stems from particle physics, where graphs $G$ with bridgeless matroids $M=\GMat{G}$ play a special role and are called \emph{1-particle irreducible}, see \cite[Section~5.8]{Borinsky:PhdBook} and \cite{JacksonKempfMorales:RobustLegendre}. 
Note that in our terminology, {\OnePI} does not require connectedness in any sense: direct sums of bridgeless matroids remain bridgeless.

\begin{figure}
	$ K_4\setminus\set{e} = \Graph[0.4]{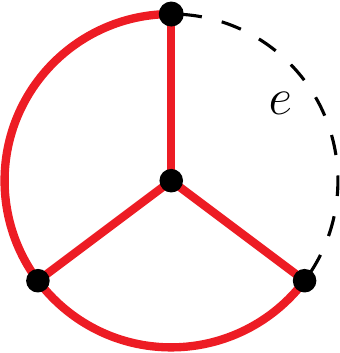}$
	\hfill
	$K_4\setminus\set{v} = \Graph[0.4]{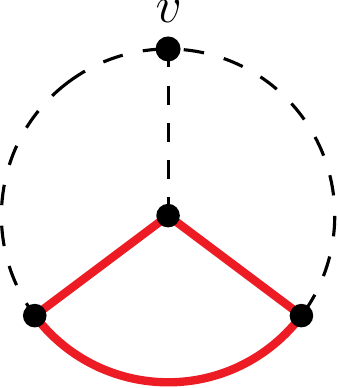}$
	\hfill
	$K_4\setminus\set{e,f} = \Graph[0.4]{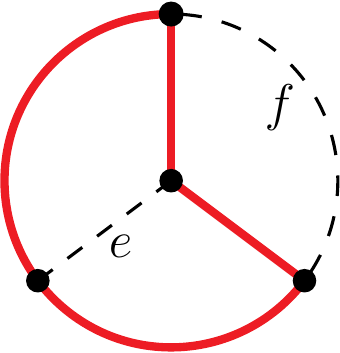}$
	\caption{The three types of bridgeless subgraphs of $K_4$ are highlighted (solid edges).}%
	\label{fig:ws3-flags}%
\end{figure}
A bridge $e$ is characterized by the equivalent conditions $\rank{M\setminus e}=\rank{M}-1$ and $\loops{M\setminus e} = \loops{M}$, and hence a bridgeless matroid is a minimal subset for its loop number:
\begin{equation*}
	\text{$M$ is bridgeless}
	\quad\Leftrightarrow\quad
	\loops{M\setminus e}<\loops{M}\ \text{holds for all $e\in\EG{M}$.}
\end{equation*}
Let $\Bridges{M} \subset \EG{M}$ denote the set of all bridges of $M$. Its complement $\intM{M}$ is the largest bridgeless submatroid of $M$ and it consists of the union of all circuits:
\begin{equation}
	\intM{M}
	\defas M\setminus \Bridges{M}
	= \bigcup_{C\in\Circuits{M}} C
	.
	\label{eq:intM}%
\end{equation}
Bridgeless graphs enter the study of Feynman periods through the desingularization of graph hypersurfaces, where they are referred to as \emph{motic graphs} in \cite[Definition~3.1]{Brown:FeynmanAmplitudesGalois} and \emph{core graphs} in \cite{BlochKreimer:MixedHodge}. They label singular loci and after blowing-up, the \emph{flags} (maximal chains) of bridgeless graphs correspond to the deepest strata of the boundary divisor \cite[Lemma~7.4]{BlochEsnaultKreimer:MotivesGraphPolynomials}. It is therefore not surprising that they can also organize the Hepp bound:
\begin{proposition}\label{prop:hepp-1pi-flags}
	For a connected matroid $M$ on $N$ edges with $\ell=\loops{M}\geq 1$ loops, let
	\begin{equation*}
		\Flags[\OnePI]{M}
		\defas \setexp{
			\emptyset = \gamma_0 \subsetneq \gamma_1 \subsetneq \cdots \subsetneq \gamma_{\ell}=M
		}{
			\text{each $\gamma_k$ is bridgeless with $\loops{\gamma_k}=k$}
		}
	\end{equation*}
	denote the set of flags of bridgeless submatroids of $M$. For any nested subsets $\delta\subseteq \gamma \subseteq M$, let $\ind_{\gamma/\delta} \defas \sum_{e \in \gamma\setminus\delta} \ind_e$ denote the sum of the indices of the additional edges in $\gamma$. Then
	\begin{equation}
		\Hepp{M,\vec{\ind}}
		=
		\frac{1}{\ind_1\cdots \ind_N}
		\sum_{\gamma_{\bullet} \in \Flags[\OnePI]{M}} 
		\frac{
			\ind_{\gamma_1/\gamma_0} \cdots \ind_{\gamma_{\ell}/\gamma_{\ell-1}}
		}{
			\sdc{\gamma_1} \cdots \sdc{\gamma_{\ell-1}}
		}.
		\label{eq:hepp-1pi-flags}%
	\end{equation}
\end{proposition}
\begin{example}\label{ex:K4-1pi-flags}
	Consider the complete graph $K_4=\gKfour$ on four vertices, with unit indices $\ind_1=\cdots=\ind_6=1$ and thus in $\Dim=4$ dimensions.
	The only bridgeless subgraphs are:
	\begin{itemize}
		\item six edge complements $K_4\setminus\set{e} \cong \gKite$ with two loops and $\sdc{\gKite}=5-\frac{4}{2}\cdot 2=1$,
		\item four triangles $\gTri \cong K_4\setminus\set{v}$ with $\sdc{\gTri}=1$ from removing a vertex, and
		\item three squares $\gBox\cong K_4\setminus\set{e,f}$ with $\sdc{\gBox}=2$ by deleting two non-adjacent edges.
	\end{itemize}
	These are illustrated in \autoref{fig:ws3-flags}, and we can form $\big|\Flags[\OnePI]{K_4}\big|=18$ different flags. They come in two types, and their contributions to the Hepp bound $\Hepp{K_4}=84=6\cdot 12 + 2\cdot 6$ are
	\begin{itemize}
		\item $\frac{3\cdot 2 \cdot 1}{1\cdot 1}=6$ for each of the 12 flags $\gTri \subset \gKite \subset \gKfour$, and
		\item $\frac{4\cdot 1 \cdot 1}{2\cdot 1}=2$ for each of the 6 flags $\gBox \subset \gKite \subset \gKfour $.
	\end{itemize}
\end{example}
\begin{remark}
	Every biconnected graph admits a flag of biconnected (hence bridgeless) graphs \cite[Theorem~19]{Whitney:NonSeparablePlanar}, as in the example above. Such flags are called \emph{open ear decompositions}, and this notion generalizes to connected matroids. However, the sum \eqref{eq:hepp-1pi-flags} will typically involve more general bridgeless flags, as in \autoref{fig:perm-to-flag} ($\gamma^{\sigma}_2$ is separable).
\end{remark}
\begin{figure}\centering
	$ G=\Graph[0.9]{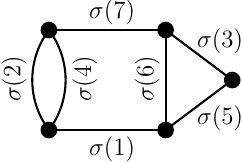}$
	\quad \scalebox{1.6}{$\rightarrow$} \quad
	$ \begin{aligned}
		G^{\sigma}_4 &=\Graph[0.45]{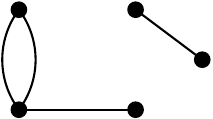} \\
		G^{\sigma}_6 &=\Graph[0.45]{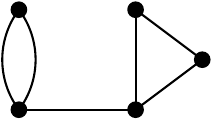}
	\end{aligned} $
	\quad \scalebox{1.6}{$\rightarrow$} \quad
	$ \begin{aligned}
		\gamma_1^{\sigma} &=\Graph[0.45]{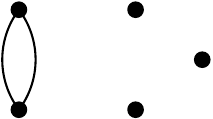} \\
		\gamma_2^{\sigma} &=\Graph[0.45]{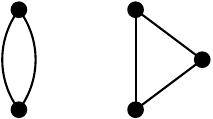} 
	\end{aligned} $
	\caption{In the order $\sigma$ of the edges of the depicted graph $G$, the loop number increases at $i_1=4$, $i_2=6$ and $i_3=7$. The associated flag is $\gamma_1^{\sigma}\subset \gamma_2^{\sigma} \subset \gamma_3^{\sigma}=G$.}%
	\label{fig:perm-to-flag}%
\end{figure}
\begin{proof}[Proof of \autoref{prop:hepp-1pi-flags}]
	Each permutation $\sigma \in \perms{M}$ defines a bridgeless flag as follows (see \autoref{fig:perm-to-flag}): Let $i_1<\ldots<i_{\ell}$ denote the positions of edges that add a loop:
	\begin{equation*}
		\loops{M^{\sigma}_{i_k}} = 1+\loops{M^{\sigma}_{i_{k-1}}}
		\quad\text{for every}\quad
		1 \leq k \leq \ell.
	\end{equation*}
	The corresponding bridgeless subsets $\gamma_k^{\sigma} = \intM{M^{\sigma}_{i_k}} \subseteq M^{\sigma}_{i_k}$ give rise to a map
	\begin{equation*}
		F_M \colon \perms{M} \longrightarrow \Flags[\OnePI]{M}, \quad
		\sigma \mapsto
		\left( \gamma^{\sigma}_{1} \subsetneq \cdots \subsetneq \gamma^{\sigma}_{\ell} \right)
		,
%		\label{eq:perm-to-1pi}%
	\end{equation*}
	and we ask which permutations $\sigma$ lie in the preimage of a given bridgeless flag $\gamma_{\bullet}\in\Flags[\OnePI]{M}$. The last edge $e=\sigma(N) \in S\defas M\setminus M'$ must belong to the complement of $M'\defas \gamma_{\ell-1}$, and all remaining edges $S\setminus\set{e}$ are bridges of $M\setminus e$. Those may appear in any order and at arbitrary positions in $\sigma$, without changing the associated bridgeless flag.
	So if we write $\sigma'$ for the order of the edges of $M'$ as they appear in $\sigma$ and we fix some $\tau \in \perms{M'}$ with $F_{M'}(\tau) = \gamma'_{\bullet} \defas (\gamma_1\subsetneq\cdots\subsetneq M')$, then the set $\setexp{\sigma \in \perms{M}}{\sigma'=\tau\ \text{and}\ \sigma(N)=e}$ is in bijection with the shuffles of $\tau$ and the elements of $S\setminus \set{e}$. The sum over these $\sigma$ is
	\begin{align*}
		\sum_{ F_M(\sigma)=\gamma_{\bullet}}
		\dChar\left(\Dsdc{\sigma}\right)
		&=
		\sum_{ F_{M'}(\tau) = \gamma_{\bullet}'}\ 
		\sum_{e \in S}
		\Char\left(
			\Dsdc{\tau}
			\shuffle
			\bigshuffle_{e \neq f \in S} \ind_f
%			\bigshuffle_{f \in G\setminus (G'\cup e)} \ind_f
		\right)
		\\ &
%		=\sum_{ F(\sigma') = \gamma_{\bullet}' \right)}
%		\sum_{e \in G\setminus G'}
%		\frac{\Char(\Dsdc{\sigma'})}{\prod_{e \neq f \in G\setminus G'} \ind_f}
		=
		\frac{\ind_{M/M'}}{
			\ind^S
		}
		\frac{1}{\sdc{M'}}
		\sum_{F_{M'}(\tau) = \gamma_{\bullet}'}
		\dChar\left(\Dsdc{\tau} \right)
	\end{align*}
	with $\ind^S\defas \prod_{e \in S} \ind_e$, where we used the multiplicativity \eqref{eq:Char-shuffle}. This reduces the sum over $\sigma\in F_M^{-1}(\gamma_{\bullet})$ to the preimages $\tau\in F_{M'}^{-1}(\gamma_{\bullet}')$ of the truncated flag $\gamma_{\bullet}'$ of length $\ell-1$, and iteration of this rule eventually leads to \eqref{eq:hepp-1pi-flags}.
\end{proof}
The length of the bridgeless flags is given by the loop number $\ell$, and hence the formula \eqref{eq:hepp-1pi-flags} tends to be particularly efficient for small $\ell$. In particular, the case $M=\UM{N}{N-1}\cong\GMat{\Cycle{N}}$ of a single loop results in a unique flag and gives directly the result \eqref{eq:hepp-cycle}.
\begin{example}
	The two-loop graphs $D_3(i,j,k)$ from \autoref{fig:k33} consist of three paths with $i$, $j$ and $k$ edges between shared endpoints. Their only bridgeless subgraphs are the three cycles $\Cycle{j+k}$, $\Cycle{i+k}$ and $\Cycle{i+j}$ that are left over after deleting the edges of one of the paths. Hence the sum in \eqref{eq:hepp-1pi-flags} has merely three terms, and for unit indices we obtain
	\begin{equation}
		\Hepp{D_3(i,j,k)}
		= 
			 \frac{i(j+k)}{j+k-\Dim/2}
			+\frac{j(i+k)}{i+k-\Dim/2}
			+\frac{k(i+j)}{i+j-\Dim/2}
		.
		\label{eq:D3ijk}%
	\end{equation}
\end{example}
\begin{figure}
	\centering
	$D_3(i,j,k) = \Graph[0.6]{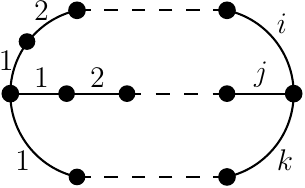} $
	\hfill
	$K_{3,3}^{-} = \Graph[0.5]{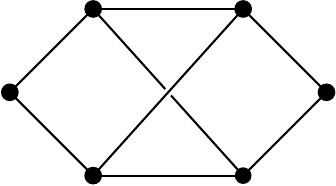}$
	\hfill
	$K_{3,3} = \Graph[0.4]{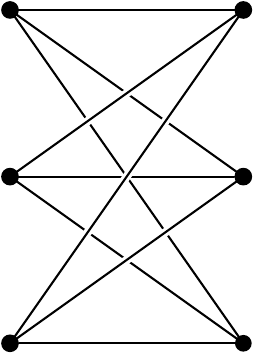}$
	\caption{The family $D_3(i,j,k)$ of all biconnected two-loop graphs, the complete bipartite graph $K_{3,3}$ and its depletion $K_{3,3}^-=K_{3,3}\setminus e$ by one edge.}%
	\label{fig:k33}%
\end{figure}
In more interesting cases, however, the number of bridgeless flags can become huge: Fix a basis $b\in\Bases{M}$ and consider an order $\tau$ of its complement $\EG{M}\setminus b=\set{\tau(1),\ldots,\tau(\ell)}$. 
Let $C_k^{\tau}$ denote the unique circuit contained in $b\cup\set{\tau(k)}$, then 
%$\gamma_k^{\tau} \defas \bigcup_{1\leq i \leq k}C^{\tau}_i$
$\gamma_k^{\tau} \defas C^{\tau}_1\cup\ldots\cup C^{\tau}_k$
is bridgeless for each $1\leq k \leq \ell$, with $\loops{\gamma^{\tau}_k}=k$ loops, and so we obtain a bridgeless flag $\gamma_{\bullet}^{\tau} \in \Flags[\OnePI]{M}$. Since $\set{\tau(k)}=\gamma_k^{\tau}\setminus (\gamma^{\tau}_{k-1} \cup b)$, this construction yields an injection $\perms{\ell} \injects \Flags[\OnePI]{M}$ and we conclude that every matroid has $|\Flags[\OnePI]{M}| \geq \ell!$ bridgeless flags.

But it is not necessary to explicitly enumerate the flags, due to the recursive structure of \eqref{eq:hepp-1pi-flags}: Summing only over the penultimate element $\gamma=\gamma_{\ell-1}$ of the flag, we see that
\begin{equation}
	\Hepp{M,\vec{\ind}}
	= \sum_{\substack{
%		\OnePI\ \gamma \subset G\ \text{with} \\
		\text{bridgeless}\ \gamma \subset M\\
		\text{with}\
		\loops{\gamma}=\loops{M}-1
	}}
		\frac{\ind_{M/\gamma}}{
		%	\prod_{e\notin \gamma} \ind_e
			\ind^{\gamma^c}
		}
		\frac{\Hepp[\Dim]{\gamma,\vec{\ind}}}{\sdc{\gamma}}
	\label{eq:hepp-flag-recursive}%
\end{equation}
where the subscript in $\Hepp[\Dim]{\gamma}$ indicates that this Hepp bound is to be computed in the dimension $\Dim$ determined by $\sdc{M}=0$. This gives a result different from the actual Hepp bound $\Hepp{\gamma}$ of $\gamma$ by itself, since the latter imposes another dimension where $\sdc{\gamma}=0$.
\begin{remark}
	We may expand $\gamma$ in \eqref{eq:hepp-flag-recursive} into its connected components, similar to \eqref{eq:hyperplane-block-recursion}.
\end{remark}
\begin{example}\label{ex:R10-flags}
	The smallest non-graphic regular matroid is called $R_{10}$ \cite{Seymour:DecompositionRegular}. It has rank $5$ and may be represented by the $10$ vectors $\setexp{\uv{i}+\uv{j}+\uv{k}}{1\leq i<j<k\leq 5} \subset \FF{2}^{5}$ over the field $\FF{2}=\Z/2\Z$ with two elements. 
	The complement of every edge $e$ is isomorphic to the graphic matroid
	$R_{10} \setminus e \cong \GMat{K_{3,3}}$
	of the complete bipartite graph $K_{3,3}$. With unit indices, $\sdc{R_{10}}=0$ in $\Dim=4$ dimensions, and hence the Hepp bound of $R_{10}$ is
	\begin{equation*}
		\Hepp{R_{10}} 
		= 10\cdot \frac{\Hepp[4]{K_{3,3}}}{\sdc{K_{3,3}}}
		= 10 \cdot \Hepp[4]{K_{3,3}}
		= 10 \cdot 9 \cdot \frac{\Hepp[4]{}(K_{3,3}^-)}{\sdc{K_{3,3}^{-}}}
		= 45 \cdot \Hepp[4]{}(K_{3,3}^-)
	\end{equation*}
	in terms of the graph $K_{3,3}^- \cong K_{3,3}\setminus e$ depicted in \autoref{fig:k33}. It has two bridgeless subgraphs of the form $D_3(2,2,2)=\Graph[0.2]{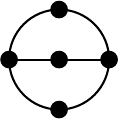}$ and four subgraphs isomorphic to $D_3(3,1,3)=\Graph[0.22]{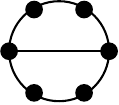}$, hence
	\begin{equation*}
		\Hepp[4]{}\big(K_{3,3}^{-}\big)
		=
		2 \cdot 2 \cdot \frac{\Hepp[4]{\Graph[0.2]{d222}}}{\sdc{\Graph[0.2]{d222}}}
		+
		4 \cdot 1 \cdot \frac{\Hepp[4]{\Graph[0.22]{d313}}}{\sdc{\Graph[0.22]{d313}}}
		= 4 \cdot \frac{12}{2} + 4 \cdot \frac{27/2}{3}
		= 42
	\end{equation*}
	according to \eqref{eq:D3ijk} and we conclude that $\Hepp{R_{10}} = 45\cdot 42 = 1890$.
\end{example}
\begin{remark}\label{rem:Derksen-1pi-recursion}
	The recursion underlying \eqref{eq:hepp-flag-recursive} can also be applied to the Derksen invariant: Indeed, the proof of \autoref{prop:hepp-1pi-flags} readily demonstrates that $\Derksen{M}$ is completely determined by the lattice $\Cycs{M}$ of its bridgeless submatroids, and we have
\begin{equation}
	\Derksen{M} = \sum_{\substack{\text{bridgeless}\ \gamma \subset M \\ \loops{\gamma}=\loops{M}-1}}
	\abs{M/\gamma}! \cdot \left[ 
		\Derksen{\gamma}
		\shuffle
		\buch{1}^{\abs{M/\gamma}-1}
	\right]
	\buch{0}.
	\label{eq:Derksen-from-1PI}%
\end{equation}
%\todo[inline]{\cite{BoninKung:Ginvariant}}
\end{remark}

From an algorithmic point of view, this method to compute the Hepp bound can be implemented as a traversal of the Hasse diagram of the lattice
\begin{equation*}
	\Cycs{M} \defas \setexp{\emptyset \subseteq \gamma \subseteq \EG{M}}{\gamma\ \text{is \OnePI}}
\end{equation*}
of bridgeless submatroids. Starting from its maximum, which is $M$ itself, this lattice can be explored efficiently in a top-down approach as follows:

Given a bridgeless matroid $\gamma$, we call two edges $e$ and $f$ equivalent if $e$ is a bridge of $\gamma\setminus \set{f}$, in other words, if $\rank{\gamma\setminus\set{e,f}}<\rank{\gamma}$. This is an equivalence relation, and we can compute the corresponding partition $\EG{\gamma} = S_1\sqcup\ldots\sqcup S_k$ into equivalence classes $S_i$ using less than $\abs{\gamma}^2$ calls to the rank function. The bridgeless submatroids of $\gamma$ with loop number $\loops{\gamma}-1$, that is the maximal elements below $\gamma$ in $\Cycs{M}$, are precisely the complements $\gamma\setminus S_i$. So the recursion \eqref{eq:hepp-flag-recursive} has $k\leq \abs{\gamma}$ summands.
\begin{corollary}
	Let $K=|\Cycs{M}| \leq 2^N$. Then the Hepp bound of $M$ can be computed in $\asyO{K\cdot N^2}$ many steps, provided that the $K$ values of $\Hepp[\Dim]{\gamma,\vec{\ind}}$ can be stored in memory.
\end{corollary}
We stress that the lattice $\Cycs{M}$ grows slowly from the top down: each element $\gamma$ has at most $\abs{\gamma}$ children in the Hasse diagram. In contrast, there is no such bound in the bottom-up direction: For example, the number of circuits (sets directly above $\emptyset \in \Cycs{M}$) in a connected matroid can be exponentially large (take cycles in complete graphs).
\begin{corollary}\label{cor:Hepp-polynomial-time}
	Every connected matroid with $N$ elements and $\ell$ loops has at most $N(N-1)\cdots (N-\ell+1)$ bridgeless flags. Consequently, the Hepp bound of matroids with bounded loop number is computable in polynomial time in $N$.
\end{corollary}
More formally, if the matroid is given by its rank function as an oracle, then $\asyO{N^{\ell+2}}$ oracle calls are sufficient to determine the Hepp bound. By \autoref{cor:hepp-volume}, this gives a polynomial time algorithm for the calculation of the volume of the polar of the matroid polytope. For matroids polytopes themselves, such a result is well known \cite{DeLoeraHawsKoeppe:Ehrhart}.

\subsection{Flats and cuts}
\label{sec:flat-flags}

\begin{definition}
	A subset $\gamma\subseteq \EG{M}$ of a matroid is called a \emph{flat} (or \emph{closed}) if it is maximal for its rank, so that $\rank{\gamma\cup e}>\rank{\gamma}$ for every $e\in\EG{M}{\setminus}\gamma$. The set of flats of $M$ forms a lattice $\Flats{M}$, and the \emph{span} or \emph{closure} of a subset $\gamma$ of $\EG{M}$ is the unique minimal flat $\closeM{\gamma}$ that contains $\gamma$. 
	The flats $\gamma$ with $\rank{\gamma} = \rank{M}-1$ are called \emph{hyperplanes} (also \emph{copoints}), and the complements $M\setminus \gamma$ of hyperplanes are the \emph{cocircuits}.
\end{definition}
In the case of graphic matroids, a flat is a subgraph $\gamma \subset G$ such that each connected component $\delta$ of $\gamma$ is \emph{(vertex-)induced}, saying that $\delta$ contains all edges of $G$ that have both endpoints in $\delta$. The hyperplanes $\gamma$ of a connected graph $G$ consist of precisely two components and correspond to vertex bipartitions $\VG{G}=S \sqcup T$ (cuts) for which both parts $S,T\neq\emptyset$ induce connected subgraphs. Hence, hyperplane complements (circuits) $G\setminus \gamma$ are the minimal edge-cuts, also called \emph{bonds} \cite{Tsuchiya:BondModular,Tsuchiya:BondLattices}. For $3$-connected $G$, the vertex complements $G\setminus v$ are precisely the connected hyperplanes \cite[Theorem~1]{Sachs:GraphsMatroidsLattices}.
\begin{figure}
	\centering
	\begin{tabular}{cccccc}
	$\Graph[0.6]{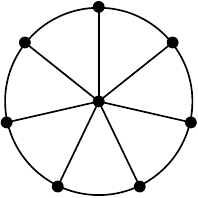}$ &
	$\supset$ &
	$\Graph[0.6]{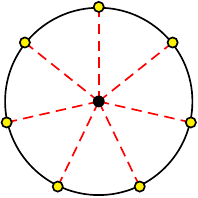}$ &
	$\Graph[0.6]{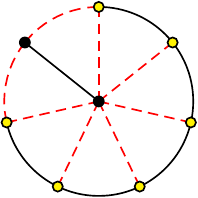}$ &
	$\Graph[0.6]{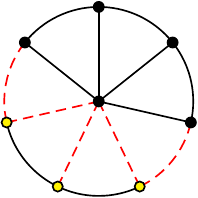}$ &
	$\Graph[0.6]{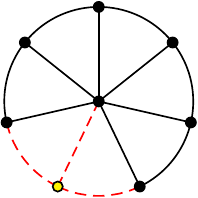}$ \\
	$ \WS{7}$ &
	&
	$\Cycle{7} \sqcup \Path{1}$ &
	$\Path{6} \sqcup \Fan{1}$ &
	$\Path{3} \sqcup \Fan{4}$ &
	$\Path{1} \sqcup \Fan{6}$ 
	\end{tabular}
	\caption{Some cuts of the wheel with seven spokes. The dashed edges are cut and separate the wheel into two parts, indicated by differently drawn vertices.% These parts can be the cycle $\Cycle{7}$, paths $\Path{k}$ or fans $\Fan{k}$.
	}%
	\label{fig:wheel-cuts}%
\end{figure}
\begin{example}
	The wheel graph $\WS{n}$ with $n=\loops{\WS{n}}$ spokes and loops has essentially two types of minimal cuts, see \autoref{fig:wheel-cuts}:
	Either the hub is dissected from the entire rim cycle $\Cycle{n}$, or a path $\Path{k}$ on $k$ vertices in the rim gets separated from a fan $\Fan{n-k}$.
\end{example}
The minimal element of $\Flats{M}$ is the unique flat of rank zero, namely the set $\closeM{\emptyset}$ which consists of the self-loops of $M$. For $M$ connected with rank at least one, this is the empty set. In this case, the set of flags (maximal chains) of flats is
\begin{equation*}
	\Flags[\Flat]{M} 
	\defas \setexp{
		\emptyset = \gamma_0 \subsetneq \gamma_1 \subsetneq \cdots \subsetneq \gamma_{\rank{M}}=M
	}{
		\text{each $\gamma_k$ is a flat with $\rank{\gamma_k}=k$}
	}.
\end{equation*}
These flags are known to encode a matroid in a very interesting way and have received more attention than the bridgeless case \cite{BorovikGelfandVinceWhite:FlatsFlagPolytope,Hampe:IntersectionRing}. Our following observation that the flags of flats directly determine the Hepp bound is very much in the spirit of \cite{Hampe:IntersectionRing,BoninKung:Ginvariant}.
\begin{remark}
	In the position space theory of Feynman integrals \cite{Berghoff:WonderfulQFT,BergbauerBrunettiKreimer:RenResSing}, flats of Feynman graphs are called \emph{saturated graphs} and used to define an arrangement of linear spaces that are blown up to obtain a wonderful compactification.
\end{remark}
\begin{proposition}\label{prop:hepp-flat-flags}
	For a connected matroid $M$ on $N$ edges with rank $r=\rank{M} \geq 1$,
	\begin{equation}
		\Hepp{M,\vec{\ind}}
		= 
%		\frac{1}{(\frac{\Dim}{2} - \ind_1)\cdots(\frac{\Dim}{2} - \ind_N)}
		\frac{1}{\ind_1^{\dual}\cdots \ind_N^{\dual}}
		\sum_{\gamma_{\bullet} \in \Flags[\Flat]{M}}
		\frac{\ind_{\gamma_1}^{\dual} \ind_{\gamma_2/\gamma_1}^{\dual}\cdots \ind_{M/\gamma_{r-1}}^{\dual}}{\sdc{\gamma_1}\cdots\sdc{\gamma_{r-1}}}
%		\frac{\ind_{\gamma_1/\gamma_0}^{\dual} \cdots \ind_{\gamma_r/\gamma_{r-1}}^{\dual}}{\sdc{\gamma_1}\cdots\sdc{\gamma_{r-1}}}
		\label{eq:hepp-flat-flags}%
	\end{equation}
	where $\ind^{\dual}_{\gamma/\delta} \defas \sum_{e\in\gamma\setminus \delta} \ind_e^{\dual}$ denotes the sum of the \emph{dual indices} $\ind_e^{\dual} \defas \frac{\Dim}{2}-\ind_e$ of all edges of $\gamma$ that are not already in $\delta\subset \gamma$ (see also \autoref{sec:duality}).
\end{proposition}
\begin{proof}
	Given an order $\sigma \in \perms{M}$, consider the positions $i_1<\ldots<i_r$ of the edges which increase the rank: $\rank{M^{\sigma}_{i_k}}=1+\rank{M^{\sigma}_{i_{k-1}}}$.
	The flats $\gamma_k^{\sigma} \defas \closeM{M^{\sigma}_{i_k}}$ form a flag
	\begin{equation*}
		(\gamma_1^{\sigma} \subsetneq \cdots \subsetneq \gamma_{r}^{\sigma})
		\in \Flags[\Flat]{M}
	\end{equation*}
	and we like to sum over all permutations $\sigma$ that produce a given, fixed flag $\gamma_{\bullet} \in \Flags[\Flat]{M}$.
	Since $M$ has no self-loops, note $i_1=1$ so $\Dsdc{\sigma}$ starts with $\buch{\ind_e}$ for $e\defas \sigma(1)$.
	To satisfy $\gamma^{\sigma}_1 = \gamma_1$, we must have $e\in\gamma_1$. The remaining edges $f\in\gamma_1\setminus e$ each add a loop and increment the superficial degree of convergence by $\ind_f-\frac{\Dim}{2}$.
	Their positions in $\sigma$ do not affect the flag $\gamma_{\bullet}^{\sigma}$.
	So if we fix the order $\tau = \sigma|_{M\setminus \gamma_1}$ of the edges not in $\gamma_1$ and write $u$ for the subsequence of $\Dsdc{\sigma}$ given by their increments, then the sum over such $\sigma$ contributes
	\begin{equation*}
		\sum_{e\in \gamma_1} \dChar\bigg(
			\buch{\ind_e } 
			\Big( 
				u 
				\shuffle 
				\bigshuffle_{f\in\gamma_1\setminus e} 
				\buch{ \ind_f-\tfrac{\Dim}{2} } 
			\Big)
		\bigg)
		=
		\sum_{e\in \gamma_1} \Char\bigg(
			\apode u
			\shuffle 
			\bigshuffle_{f\in\gamma_1\setminus e} 
			\apode \buch{\ind_f-\tfrac{\Dim}{2}}
		\bigg).
	\end{equation*}
	Here we reversed the order of arguments using \eqref{eq:dChar-apode} and passed from $\dChar$ to $\Char$, which drops the final letter $\apode\buch{\ind_e}=-\buch{\ind_e}$. Exploiting the multiplicativity \eqref{eq:Char-shuffle}, this is
	\begin{equation*}
		=
		\dChar\Big(
			( \apode u )
			\buch{ \sdc{\gamma_1} }
		\Big)
		\sum_{e\in \gamma_1}
		\prod_{f\in\gamma_1\setminus e} \frac{1}{\ind_f^{\dual}}
		=
		\dChar\Big(
			\buch{ \sdc{\gamma_1} }
			u
		\Big)
		\frac{\ind_{\gamma_1}^{\dual}}{\prod_{f\in\gamma_1} \ind_f^{\dual}},
	\end{equation*}
	where we inserted the final letter $\buch{\sdc{\gamma_1}}$ in order to balance the word and apply \eqref{eq:dChar-apode} once more. This argument iterates and proves the claim: At the next step, we consider the first letter of $u=\buch{\ind_{\tau(1)}}u'$ and sum over $\tau(1)\in \gamma_2\setminus\gamma_1$, writing $\dChar\big( \buch{\sdc{\gamma_1}}\buch{\ind_{\tau(1)}}u'\big) = \frac{1}{\sdc{\gamma_1}} \dChar\big(\buch{\sdc{\gamma_1}+\ind_{\tau(1)}} u' \big)=\frac{1}{\sdc{\gamma_1}} \Char(\apode u')$ to then apply the analogous steps as above.
\end{proof}
The formula \eqref{eq:hepp-flat-flags} is most efficient for matroids of low rank. For bonds $\UM{n}{1}$, it gives a single term $(\Dim/2)/\prod_e \ind_e^{\dual}$, which reproduces \eqref{eq:hepp-bond}. 
Dually to \autoref{sec:1pi-flags}, the lattice of flats grows slowly from the bottom up: each flat $\gamma\in\Flats{M}$ is covered by at most $\abs{M\setminus\gamma}$ flats, namely $\closeM{\gamma\cup e_1},\ldots,\closeM{\gamma\cup e_k}$ where $M\setminus\gamma=\set{e_1,\ldots,e_k}$.
\begin{corollary}
	A connected matroid on $N$ elements with rank $r$ has no more than $N(N-1)\cdots(N-r+1)$ flags of flats. The Hepp bound of matroids with bounded rank can be computed in polynomial time in $N$.
\end{corollary}
For computations it is convenient to exploit the recursive structure of \eqref{eq:hepp-flat-flags}: Let us denote by $\HeppFlat[\Dim]{M,\vec{\ind}}$ the result of this formula in a fixed dimension $\Dim$, lifting the constraint $\sdc{M}=0$. Since the penultimate element $\gamma_{r-1}$ of a flag of flats is a hyperplane,
\begin{equation}
	\HeppFlat[\Dim]{M,\vec{\ind}}
	= \sum_{\text{hyperplane}\ \gamma\subset M}
	\frac{\HeppFlat[\Dim]{\gamma,\vec{\ind}}}{\sdc{\gamma}}
	\frac{\ind_{M/\gamma}^{\dual}}{\prod_{e\notin \gamma} \ind_e^{\dual}}.
	\label{eq:hyperplane-recursion}%
\end{equation}
If $M=A\oplus B$ is a direct sum, its flats $\alpha\cup\beta \in \Flats{A\oplus B} \cong \Flats{A} \times \Flats{B}$ are pairs of flats $\alpha,\beta$ of the summands. Consequently, the flags $\Flags[\Flat]{M}$ are in bijection with the shuffles of flags in $\Flags[\Flat]{A}$ with flags in $\Flags[\Flat]{B}$. The multiplicativity \eqref{eq:Char-shuffle} then shows that
\begin{equation}
	\frac{\HeppFlat[\Dim]{A\oplus B,\vec{\ind}}}{\sdc{A\oplus B}}
	= \frac{\HeppFlat[\Dim]{A,\vec{\ind}}}{\sdc{A}}
	\cdot
	  \frac{\HeppFlat[\Dim]{B,\vec{\ind}}}{\sdc{B}}.
	\label{eq:hepp-flats-components}%
\end{equation}
\begin{example}
	For a forest $\gamma\cong\UM{n}{n} \cong (\UM{1}{1})^{\oplus n}$, the formula \eqref{eq:hepp-flat-flags} is easily evaluated to
	\begin{equation}
		\frac{\HeppFlat[\Dim]{\gamma,\vec{\ind}}}{\sdc{\gamma}} = \frac{1}{\ind_1\cdots\ind_n}
		\quad\text{where}\quad
		\sdc{\gamma} = \ind_1+\cdots+\ind_n,
		\label{eq:hepp-flats-forest}%
	\end{equation}
	using either \eqref{eq:hyperplane-recursion} or \eqref{eq:hepp-flats-components}.
	The hyperplanes of the cycle $M=\GMat{\Cycle{n}}\cong\UM{n}{n-1}$ are precisely the forests $\gamma=M\setminus\set{e,f}$ obtained by deleting any pair of edges. So by \eqref{eq:hyperplane-recursion},
	\begin{equation}
		\HeppFlat[\Dim]{\Cycle{n},\vec{\ind}}
		= 
		\sum_{1\leq e<f\leq n}
		\bigg( \prod_{k \neq e,f} \frac{1}{\ind_k} \bigg)
		\frac{\ind_e^{\dual}+\ind_f^{\dual}}{\ind_e^{\dual}\ind_f^{\dual}}
		= \frac{1}{\ind_1\cdots \ind_n} 
		\sum_e \frac{\ind_e}{\ind_e^{\dual}} \sum_{f\neq e} \ind_f
		.
		\label{eq:hepp-flats-cycle}%
	\end{equation}
	Note that the sum over $f$ gives $\sdc{\Cycle{n}}+\ind_e^{\dual}$, such that the double sum can be written as $\frac{\Dim}{2} + \sdc{\Cycle{n}}+\sdc{\Cycle{n}} \sum_e \frac{\ind_e}{\ind_e^{\dual}}$. So we recover \eqref{eq:hepp-cycle} in the dimension where $\sdc{\Cycle{n}}=0$.
\end{example}
\begin{corollary}
	Let $M$ denote a matroid of rank $\rank{M} \geq 1$, and given any submatroid $\gamma \subset M$, write $\gamma = \gamma_1 \oplus \cdots \oplus \gamma_{\nCG{}}$ for its $\nCG{}=\nCG{\gamma}$ connected components. Then
	\begin{equation}
		\HeppFlat[\Dim]{M,\vec{\ind}}
		= \sum_{\text{hyperplane}\ \gamma\subset M}
		\frac{\ind_{M/\gamma}^{\dual}}{\prod_{e\notin \gamma} \ind_e^{\dual}}
		\prod_{k=1}^{\nCG{\gamma}}
		\frac{\HeppFlat[\Dim]{\gamma_k,\vec{\ind}}}{\sdc{\gamma_k}}
		.
		\label{eq:hyperplane-block-recursion}%
	\end{equation}
\end{corollary}
\begin{table}
	\centering
	\begin{tabular}{ccccccc}
		\toprule
		$n$ & $\Dim$ & $f_3$ & $f_4$ & $f_5$ & $f_6$ & $\Hepp{K_n}$ \\
		\midrule
		$3$ & $6$ & $3$ & & & & $3$ \\
		$4$ & $4$ & $3$ & $14$ & & & $84$ \\
		$5$ & $\frac{10}{3}$ & $3$ & $11$ & $\frac{265}{4}$ & & $\frac{5\cdot 3^7\cdot 53}{2^5}$ \\[2pt]
		$6$ & $3$ & $3$ & $10$ & $\frac{130}{3}$ & $312$ & $ 2^{16}\cdot 3^2 \cdot 5 \cdot 13$ \\
		\bottomrule
	\end{tabular}
	\qquad
	$
	\begin{aligned}
		f_3 &= 3 \\
		f_4 &= \tfrac{8n-18}{n-3} \\
		f_5 &= \tfrac{5(n-2)(25n-72)}{6(n-3)(n-4)} \\
		f_6 &= \tfrac{3(n-2)(36n^3-323n^2+948n-900)}{2(n-3)^2(n-4)(n-5)}
	\end{aligned}
	$ %
	\caption{The Hepp bounds of complete graphs with up to $n=6$ vertices.}%
	\label{tab:Hepp-Kn}%
\end{table}
\begin{example}
	The complete graph $K_n$ has $\loops{K_n} = \binom{n-1}{2}$ loops and for unit indices we find $\sdc{K_n} = \binom{n}{2}-\frac{\Dim}{2} \binom{n-1}{2}$.
	Every cut consists of two smaller complete graphs, hence \eqref{eq:hyperplane-block-recursion} yields a quadratic recursion. We can state it as follows: Set $\Dim=\frac{2n}{n-2}$ and define
	\begin{equation*}
		f_2 \defas 1
		\quad\text{and}\quad
		f_k \defas
		\frac{k f_{k-1}}{\sdc{K_{k-1}}}
		+
		\sum_{i=2}^{k-2} i \frac{f_i}{\sdc{K_i}} \frac{f_{k-i}}{\sdc{K_{k-i}}}
%		\sum_{i=1}^{k-1} i f'_i f'_{k-i}
%		\sum_{i=1}^{k-1} i \frac{f_i}{\sdc{K_i}} \frac{f_{k-i}}{\sdc{K_{k-i}}}
%		\frac{k}{2}\sum_{i=2}^{k-2} \frac{f_i}{\sdc{K_i}} \frac{f_{k-i}}{\sdc{K_{k-i}}}
%		\quad\text{where}\quad
%		f_k' = \begin{cases}
%			1 & k=1\ \text{and}\\
%			f_k/\sdc{K_k} & k\geq 2\\
%		\end{cases}.
%		\quad\text{for}\quad k\geq 2,
		\quad\text{for}\quad 3\leq k \leq n,
%		\Hepp{K_n} = \frac{(n-1)!}{(\Dim/2-1)^{\loops{K_n}}}\sum_{i=1}^{n-1} i f_i f_{n-i}
	\end{equation*}
	then $\Hepp{K_n} = (n-1)! (\frac{\Dim}{2}-1)^{-\loops{K_n}} f_n$.
%	The sequence starts with $f_3=3$, $f_4=\frac{8n-18}{n-3}$ and $f_5 = \frac{5(n-2)(25n-72)}{6(n-3)(n-4)}$.
	We give the results for small $n$ in \autoref{tab:Hepp-Kn}.
%	The computation of $f_5=\frac{5}{6} \frac{n-2}{n-3} \frac{25n-72}{n-4}$ gives $\Hepp{K_5}=\frac{579555}{32}$ in $\Dim=\frac{10}{3}$ dimensions. And $\Hepp{K_6}=5!\cdot 2^{10} \cdot f_6|_{n=6} = 38338560=2^{16}\cdot 3^2 \cdot 5 \cdot 13$ in $\Dim=3$ dimensions.
\end{example}
We use the recursion \eqref{eq:hyperplane-block-recursion} also to compute the Hepp bounds of all wheels. Recall that for graphs, the sum over hyperplanes $\gamma$ is a sum over minimal cuts, and the product over $k$ runs over the biconnected components (blocks) of $\gamma$. The computation is tractable because a wheel has only few connected flats (induced biconnected subgraphs).
\begin{proposition}\label{prop:hepp-wheels}%
	The Hepp bounds of the wheel graphs with unit indices are
	\begin{equation}
		\Hepp{\WS{n}} 
		= \frac{2n}{n-2} + \frac{1}{4^{n-1}} \sum_{k=1}^n \binom{2n-2k}{n-k} \binom{2k}{k} k\cdot 9^{n-k}
		\quad\text{for every}\quad
		n \geq 3
		,
		\label{eq:Hepp-WS-explicit}%
	\end{equation}
	and they grow asymptotically like $\Hepp{\WS{n}} \sim \frac{3\cdot 9^n}{8\sqrt{2\pi n}}$ for large $n$.
	Their generating function 
	$
		\gf{W}(z)
		\defas
		\sum_{n=3}^{\infty} \Hepp{\WS{n}} z^n
		= 84 z^3 + 572 z^4 + \tfrac{13240}{3} z^5 + 35463 z^6 + \asyO{z^7}
	$
	has the form
	\begin{equation}
		\gf{W}(z)
		= \frac{2z}{1-z}-4z-14z^2-4z^2 \log(1-z) + \frac{2z}{\sqrt{(1-9z)(1-z)^3}}.
		\label{eq:hepp-wheels}%
	\end{equation}
\end{proposition}
\begin{figure}
	\centering
	\begin{tabular}{ccccc}
		$\Graph[0.8]{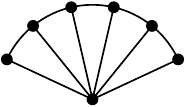}$ &
		$\supset$ &
		$\Graph[0.8]{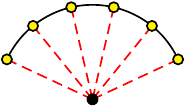}$ &
		$\Graph[0.8]{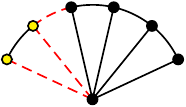}$ &
		$\Graph[0.8]{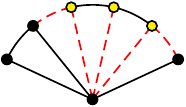}$
		\\
		$\Fan{6}$ & &
		$\Path{6}$ &
		$\Path{2}\sqcup\Fan{4}$ &
		$\Fan{2}\sqcup\Path{3}\sqcup\Fan{1}$
	\end{tabular}%
	\caption{Some cuts of the fan with $6$ spokes and the corresponding block decomposition.}%
	\label{fig:fan-cuts}%
\end{figure}
\begin{proof}
	For unit indices, the wheel $\WS{n}$ is defined in dimension $\Dim=4$ and thus $\ind_e^{\dual}=\ind_e=1$ for all $2n$ edges.
	As illustrated in \autoref{fig:wheel-cuts}, for a wheel the recursion \eqref{eq:hyperplane-block-recursion} gives
	\begin{equation*}
		\Hepp{\WS{n}} = n \frac{\HeppFlat[4]{\Cycle{n}}}{\sdc{\Cycle{n}}}
		+ n \sum_{k=1}^{n-1} (k+2) \HeppFlat[4]{\Fan{n-k}},
		\tag{$\ast$}%
		\label{eq:wheel-hepp-cuts}%
	\end{equation*}
	where the first factor of $n$ is the size of the cut (all spokes) and the factor $n$ in front of the sum over $k$ accounts for the different copies of the fan $\Fan{n-k}$ with $n-k$ spokes obtained by rotations.
	Note that the paths $\Path{k}$ in the rim give the trivial contribution $\HeppFlat[4]{\Path{k}}/\sdc{\Path{k}}=1$ from \eqref{eq:hepp-flats-forest}, and every fan has $\sdc{\Fan{n-k}}=1$. Let us write
	\begin{equation*}
		\gf{C}(z)
		\defas \sum_{n=3}^{\infty} \frac{\HeppFlat[4]{\Cycle{n}}}{\sdc{\Cycle{n}}} z^n
		= \sum_{n=3}^{\infty} \frac{n(n-1)}{n-2} z^n
		= \frac{z^3(4-3z)}{(1-z)^2} - 2z^2 \log(1-z)
	\end{equation*}
	for the generating function of the cycles according to \eqref{eq:hepp-flats-cycle}, and $\gf{F}(z) \defas \sum_{n\geq 1} \HeppFlat[4]{\Fan{n}} z^n$ for the generating series of the fans. Then the recurrence \eqref{eq:wheel-hepp-cuts} can be written as
	\begin{equation*}
		\gf{W}(z) 
		=
		z\frac{\td}{\td z} \left( 
			\gf{C}(z)
			+ \frac{z(3-2z)}{(1-z)^2} \gf{F}(z)
			- 3 z^2
		\right),
		\tag{$\dagger$}%
		\label{eq:wheel-hepp-gf-cuts}%
	\end{equation*}
	where the factor in front of $\gf{F}(z)$ is $\sum_{k=1}^{\infty} (k+2) z^{k}$. To determine $\gf{F}(z)$, consider the cuts of a fan as illustrated in \autoref{fig:fan-cuts} and apply \eqref{eq:hyperplane-block-recursion} to obtain the recurrence
	\begin{equation*}
		\HeppFlat[4]{\Fan{n}}
		= n + 2\sum_{k=1}^{n-1} (n-k+1) \HeppFlat[4]{\Fan{k}}
		+ \sum_{1<j\leq k<n} (k-j+3) \HeppFlat[4]{\Fan{j-1}} \HeppFlat[4]{\Fan{n-k}}
	\end{equation*}
	where the first term stems from dissecting the rim $\Path{n}$ from the hub, the middle sum cuts off a smaller fan $\Fan{k}$ from a path $\Path{n-k}$ in the rim, and the double sum enumerates the cuts that carve out a path in the rim from spoke $j$ to spoke $k$, chopping off two smaller fans $\Fan{j-1}$ and $\Fan{n-k}$.
	For the generating function $\gf{F}(z)$, this recursion reads
	\begin{equation*}
		\gf{F}(z)
		=
		\frac{z}{(1-z)^2}
		+ \frac{2z(2-z)}{(1-z)^2}\gf{F}(z)
		+ \frac{z(3-2z)}{(1-z)^2} \gf{F}^2(z)
%		\frac{
%			z
%			+ 2z(2-z) \gf{F}(z)
%			+ z(3-2z) \gf{F}^2(z)
%		}{(1-z)^2}
		.
	\end{equation*}
	Inserting the solution
	$
		\gf{F}(z)
		=
		\frac{1-6z+3z^2}{2z(3-2z)}-\frac{\sqrt{(1-9z)(1-z)^3}}{2z(3-2z)}
%		= z+6z^2+36z^3+\asyO{z^4}
	$
	into \eqref{eq:wheel-hepp-gf-cuts} confirms \eqref{eq:hepp-wheels}, and we obtain \eqref{eq:Hepp-WS-explicit} using the binomial series.
	The asymptotics for large $n$ can be computed with standard methods, see \cite{FlajoletSedgewick:AnalyticCombinatorics}.
\end{proof}
\begin{remark}
	If all indices $\ind_e=1$ are equal, the partition of all permutations according to the first flat $\gamma=\closeM{\set{\sigma(1)}}$, as described in the proof of \autoref{prop:hepp-flat-flags}, amounts to a recursion for the Derksen invariant of the form
\begin{equation}
	\Derksen{M} =
%	\buch{1} \Big[
	\sum_{\substack{\gamma\in\Flats{M} \\ \rank{\gamma}=1}}
%	\sum_{\rank{\gamma}=1,\gamma\in\Flats{M}}
	\abs{\gamma}! \cdot
	\buch{1}\left[ 
		\buch{0}^{\abs{\gamma}-1}
		\shuffle
		\Derksen{M/\gamma}
	\right]
%	\Big]
	.
	\label{eq:Derksen-from-flats}%
\end{equation}
We conclude that $\Derksen{M}$ is completely determined by the lattice of flats of $M$, because the flats $\Flats{M/\gamma} \cong \{\delta \in \Flats{M}\colon \delta \supseteq \gamma\}$ of a quotient $M/\gamma$ are precisely those flats $\delta$ of $M$ that contain $\gamma$. This dual to \autoref{rem:Derksen-1pi-recursion} was observed in \cite[Section~3]{BoninKung:Ginvariant}.
\end{remark}
\begin{example}
	Starting from the $3$- and $4$-bonds $\Derksen{\Graph[0.3]{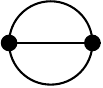}}=6 \buch{1,0,0}$ and $\Derksen{\Graph[0.3]{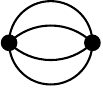}}=24\buch{1,0,0,0}$, and using $\buch{0}\shuffle\buch{1,0,0}=\buch{0,1,0,0}+3\buch{1,0,0,0}$, the recursion \eqref{eq:Derksen-from-flats} gives
	\begin{equation*}
		\Derksen{\Graph[0.45]{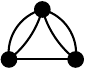}}
		=
		\buch{1}\left[
			\Derksen{\Graph[0.35]{bond4}}
			+
			2\cdot 2!\cdot \buch{0} \shuffle \Derksen{\Graph[0.35]{triple}}
		\right]
		= 96\buch{1,1,0,0,0} + 24\buch{1,0,1,0,0}.
	\end{equation*}
	Inserting this into
	$\Derksen{\gKfour} = 6 \cdot\buch{1}\cdot \Derksen{\Graph[0.35]{tri221}}
	% =576\buch{1,1,1,0,0,0}+144\buch{1,1,0,1,0,0}
	$
	confirms the result from \autoref{ex:Derksen-K4}.
\end{example}

\subsection{Cyclic flats}
\label{sec:cyclic-flats}

In the preceding two sections, we partitioned the total orders of the edges according to the corresponding flags $\gamma_{\bullet}$ of bridgeless or flat matroids. The corresponding iterated quotients are cycles $\gamma_{k}/\gamma_{k-1} \cong \UM{n}{n-1}$ or bonds $\gamma_k/\gamma_{k-1} \cong \UM{n}{1}$ on $n=\abs{\gamma_k/\gamma_{k-1}}$ elements, respectively.

Combining both approaches, we can give a formula in terms of only those submatroids that are simultaneously cyclic (bridgeless) and flat. These form the lattice
\begin{equation}
	\CycFlats{M}
	\defas \Cycs{M} \cap \Flats{M}
	\label{eq:cyclic-flats}%
\end{equation}
of cyclic flats, and in fact this lattice, together with the rank function, determines the matroid \cite{Brylawski:AffineTransversal,BoninMier:CyclicFlats}.
Note that the closure $\closeM{\gamma} \in \CycFlats{M}$ of any bridgeless $\gamma\in \Cycs{M}$, as well as the interior $\intM{\gamma} \in \CycFlats{M}$ of any flat $\gamma\in\Flats{M}$, are cyclic flats.
\begin{proposition}\label{prop:Hepp-cyclic-flats}
	The Hepp bound of a connected matroid $M$ fulfils the recursion
	\begin{equation}
		\HeppFlat[\Dim]{M,\vec{\ind}}
		= \sum_{M\neq Z \in \CycFlats{M}} %}\setminus\set{M}}
		\frac{\HeppFlat[\Dim]{Z,\vec{\ind}}}{\sdc{Z}}
		\CycFlatFactor{ M/Z ,\vec{\ind} }
		\label{eq:cycflat-recursion}%
	\end{equation}
	over cyclic flats $Z\subsetneq M$, in terms of a sum over independent hyperplanes of %the quotient 
	$Q=M/Z$:
	\begin{equation}
		\CycFlatFactor{ Q ,\vec{\ind} }
		\defas
		\sum_{\substack{%
			\text{hyperplane}\ H\subset Q \\
			\text{with}\ \loops{H}=0
		}}
		\ind^{\dual}_{Q/H}
		\Bigg( \prod_{e \in H} \frac{1}{\ind_e} \Bigg)
		\Bigg( \prod_{e\in Q\setminus H} \frac{1}{\ind_e^{\dual}} \Bigg)
		.
		\label{eq:CycFlatFactor}%
	\end{equation}
\end{proposition}
\begin{proof}%[Proof of \autoref{prop:Hepp-cyclic-flats}]
	Given a hyperplane $\gamma$ of $M$, let $Z\defas\intM{\gamma} \in \CycFlats{M}$ denote the cyclic flat that is the interior of $\gamma$. The bridges $B\defas \gamma\setminus Z$ of $\gamma$ form an independent set such that $\gamma = Z \oplus B$ is a direct sum. Using \eqref{eq:hepp-flats-forest} and \eqref{eq:hepp-flats-components}, we find
	\begin{equation*}
		\frac{\HeppFlat[\Dim]{\gamma,\vec{\ind}}}{\sdc{\gamma}}
		= \frac{\HeppFlat[\Dim]{Z,\vec{\ind}}}{\sdc{Z}}
		\prod_{b\in B} \frac{1}{\ind_b}.
	\end{equation*}
	Together with \eqref{eq:hyperplane-recursion} this proves \eqref{eq:cycflat-recursion}: Note that $H\defas \gamma/Z$ is an independent hyperplane of $Q=M/Z$, so $\CycFlatFactor{M/Z,\vec{\ind}}$ defined in \eqref{eq:CycFlatFactor} correctly collects the contributions from all hyperplanes $\gamma$ that give rise to the same interior $Z=\intM{\gamma}$.
\end{proof}
Note that \eqref{eq:cycflat-recursion} sums over all cyclic flats $Z$---in contrast to \eqref{eq:hyperplane-block-recursion} or \eqref{eq:hepp-flag-recursive}, where the rank or loop number are fixed. Correspondingly, iterating the recursion \eqref{eq:cycflat-recursion} expresses the Hepp bound in terms of all chains of cyclic flats, not just the maximal ones:
\begin{equation}
	\Hepp{M,\vec{\ind}}
	= \sum_{k \geq 1}\ 
	\sum_{\substack{
		\emptyset \neq Z_1\subset \cdots \subset Z_k=M \\
		Z_1,\ldots,Z_k \in \CycFlats{M}
	}}
	\frac{\CycFlatFactor{Z_1}\CycFlatFactor{Z_2/Z_1}\cdots\CycFlatFactor{M/Z_{k-1}}}{\sdc{Z_1}\sdc{Z_2}\cdots\sdc{Z_{k-1}}}.
	\label{eq:Hepp-cyclic-flat-chains}%
\end{equation}
\begin{remark}
	Since $M$ is connected and thus {\OnePI}, every quotient $Q=M/Z$ is also {\OnePI}. Therefore, if $Q \cong A \oplus B$ is disconnected, both $A$ and $B$ must have at least one loop: $\loops{A},\loops{B}\geq 1$. But a hyperplane $H$ of $A \oplus B$ must contain all of $A$ or all of $B$, and thus $H$ cannot be independent. This shows that $\CycFlatFactor{Q,\vec{\ind}}=0$ unless $Q$ is connected.
\end{remark}
\begin{example}
	The only cyclic flats of the complete graph $K_4=\gKfour$ are $\emptyset$, the four triangles, and $K_4$ itself. For a triangle $Z=\set{i,j,k}\cong \gTri$, we can write \eqref{eq:hepp-flats-cycle} as
	\begin{equation*}
		\frac{\HeppFlat[\Dim]{Z,\vec{\ind}}}{\sdc{Z}}
		= 
%		\frac{1}{\ind_i^{\dual} \ind_j \ind_k}
%		+\frac{1}{\ind_i \ind_j^{\dual} \ind_k}
%		+\frac{1}{\ind_i \ind_j \ind_k^{\dual}}
%		+ \frac{\ind_i+\ind_j+\ind_k}{\sdc{Z}\ind_i \ind_j \ind_k}
%%		+ \frac{1}{\sdc{Z}} \left( \frac{1}{\ind_j\ind_k} + \frac{1}{\ind_i \ind_k} + \frac{1}{\ind_i \ind_j} \right)
		\frac{1}{\ind_i \ind_j \ind_k} \left(
%			1+\frac{\Dim/2}{\sdc{Z}}
			\frac{\ind_i+\ind_j+\ind_k}{\sdc{Z}}
			+\frac{\ind_i}{\ind_i^{\dual}}
			+\frac{\ind_j}{\ind_j^{\dual}}
			+\frac{\ind_k}{\ind_k^{\dual}}
		\right)
		.
	\end{equation*}
	The corresponding quotient $Q=K_4/Z \cong \Graph[0.3]{triple}$ is a bond with the unique hyperplane $\emptyset$, and \eqref{eq:CycFlatFactor} gives simply $\CycFlatFactor{Q,\vec{\ind}} = \frac{\ind^{\dual}_u+\ind^{\dual}_v +\ind^{\dual}_w}{\ind^{\dual}_u \ind^{\dual}_v \ind^{\dual}_w}$ where $\set{u,v,w}=K_4\setminus Z$ are the complementary edges. Note that the numerator is $\ind_u^{\dual}+\ind_v^{\dual}+\ind_w^{\dual}=\frac{\Dim}{2}-\sdc{Q} = \frac{\Dim}{2}+\sdc{Z}=\ind_i+\ind_j+\ind_k$ due to $0=\sdc{K_4}=\sdc{Z}+\sdc{K_4/Z}$.
	The only other contribution to \eqref{eq:cycflat-recursion} is $\CycFlatFactor{K_4,\vec{\ind}}$ from $Z=\emptyset$, which amounts to summing over the three pairs $H=\set{i,j}\cong\Graph[0.3]{twosticks}$ of non-adjacent edges. In total, we get
	\begin{equation}
		\Hepp{K_4,\vec{\ind}}
		= 
		\sum_{\Graph[0.18]{tri}\cong Z \subset K_4}
		\frac{\ind_Z}{\prod\limits_{e\in Z}\ind_e \prod\limits_{e\notin Z} \ind_e^{\dual}} 
		\left(
			\frac{\ind_Z}{\sdc{Z}}
			+\sum_{e \in Z} \frac{\ind_e}{\ind_e^{\dual}}
		\right)
		+ \sum_{\Graph[0.2]{twosticks}\,\cong Z \subset K_4}
		\frac{\ind^{\dual}_{K_4/Z}}{\prod\limits_{e\in Z}\ind_e \prod\limits_{e\notin Z} \ind_e^{\dual}}
		.
		\label{eq:Hepp-K4-full}%
	\end{equation}
\end{example}
\begin{remark}
	In the case of unit indices $\ind_e=1=\ind_e^{\dual}$ in $\Dim=4$ dimensions, the factor \eqref{eq:CycFlatFactor} becomes just $\CycFlatFactor{Q}=(\loops{Q}+1)\cdot \abs{\set{\text{independent hyperplanes of $Q$}}}$. The count of independent hyperplanes can be retrieved from the lattice $\CycFlats{Q}$ of cyclic flats and their ranks, as worked out in detail in \cite[section~7]{BoninKung:Ginvariant}.
\end{remark}
We can repeat the above discussion starting with bridgeless flags instead of flags of flats: Given $\gamma_{\bullet} \in \Flags[\OnePI]{M}$, we can set $Z_i \defas \closeM{\gamma_i} \in \CycFlats{M}$ and remove duplicates to obtain a chain of cyclic flats. The final formula has the same form as \eqref{eq:Hepp-cyclic-flat-chains}, only the numerator is altered by replacing the factors $\CycFlatFactor{Q,\vec{\ind}}$ with the sum 
\begin{equation*}
	\CycFlatFactorDual{Q,\vec{\ind}}
	= \sum_{C\subset Q} \ind_C
	\Bigg( \prod_{e \in C} \frac{1}{\ind_e} \Bigg)
	\Bigg( \prod_{e\in Q\setminus C} \frac{1}{\ind_e^{\dual}} \Bigg)
\end{equation*}
over spanning circuits (indeed, note that $\gamma_1$ is a spanning circuit of $Z_1$). These are precisely the circuits of rank $\rank{C}=\rank{Q}$, equivalently the circuits with $\abs{C}=\rank{Q}+1$ elements, and they are also called \emph{Hamiltonian circuits}. For unit indices $\ind_e=1$ in $\Dim=4$ dimensions, we conclude $\CycFlatFactorDual{Q} = (\rank{Q}+1)\cdot \abs{\set{\text{Hamiltonian circuits of $Q$}}}$.

The two resulting formulas \eqref{eq:Hepp-cyclic-flat-chains} for $\Hepp{M,\vec{\ind}}$ in terms of chains of cyclic flats, one with $\rho$'s and the other with $\tilde{\rho}$'s in the numerator, can be obtained from each other by duality as in \autoref{sec:duality}.
Namely, the complement of a Hamiltonian circuit is an independent hyperplane of the dual \cite[Theorem~3]{Borowiecki:HamiltonianMatroids}.

\begin{example}
	The uniform matroid $M=\UM{n}{r}$ has only two cyclic flats $\CycFlats{M} = \set{\emptyset, M}$, such that $\Hepp{M,\vec{\ind}} = \CycFlatFactor{M,\vec{\ind}} = \CycFlatFactorDual{M,\vec{\ind}}$.
	As every subset $H \subset M $ of size $\abs{H}=r-1$ is an independent hyperplane and every $C\subset M$ with $\abs{C}=r+1$ is a Hamiltonian circuit,
	\begin{equation}
		\Hepp{\UM{n}{r},\vec{\ind}}
		= \sum_{\abs{C}=r+1} \frac{\ind_{C}}{\left( \prod_{e\in C} \ind_e \right) \left( \prod_{e\notin C} \ind^{\dual}_e \right)}
		= \sum_{\abs{H}=r-1} \frac{\ind^{\dual}_{H}}{\left( \prod_{e\in H} \ind_e \right) \left( \prod_{e\notin H} \ind^{\dual}_e \right)}
		.
		\label{eq:Hepp-Unr-cycflat}%
	\end{equation}
	Note that \autoref{lem:Hepp-Unr} gives yet another formula for this function.
\end{example}

\section{Symmetries}
\label{sec:symmetries}

Different graphs (or matroids) may integrate to the same period. There is no complete combinatorial description of all such pairs of graphs, but several families of identities are known. The simplest of these period relations are:
\begin{enumerate}
	\item \emph{duality:} A planar graph and its dual have the same period.
	\item \emph{product:} The period factorizes for graphs with a $2$-separation.
\end{enumerate}
These hold for arbitrary indices and extend to all matroids; proofs are straightforward in Schwinger parameters. In contrast, the following symmetries are only defined for graphs and furthermore subject to constraints on the indices:
\begin{enumerate}\setcounter{enumi}{2}
	\item \emph{Fourier split \cite{HuSchnetzShawYeats:Further}:}
		Duality may be applied to one side of a $3$-separation.
	\item \emph{completion \cite{Broadhurst:5loopsbeyond}:} If $G$ is regular, then $\Period{G{\setminus}v}$ is the same for all vertices $v$.
	\item \emph{twist \cite{Schnetz:Census}:} A double transposition along a $4$-separation keeps $\Period{G{\setminus}v}$ invariant.
\end{enumerate}
The Fourier split generalises the \emph{uniqueness relations} \cite{Kazakov:Uniqueness,Kazakov:MethodOfUniqueness,Kazakov:TwoLectures} and we include a proof in Schwinger parameters.
For completion and twist, all known proofs \cite{Schnetz:Census} exploit the representation \eqref{eq:period-position-space} of the period in \emph{position space}. Completion and twist are not restricted to regular graphs, but require indices such that each vertex is \emph{conformal} (\autoref{def:excess}).

In this section we demonstrate that the Hepp bound respects all of the above symmetries and fulfils exact analogues of the corresponding identities for periods.

The product, completion and twist identities are most succinctly stated for arbitrary indices if we use a slight variation of the Hepp bound.
\begin{definition}\label{def:Hepp-pos}
	The \emph{position space Hepp bound} $\HeppComp{M,\vec{\ind}}$ of a matroid $M$ on $N$ edges is
	\begin{equation}
		\HeppComp{M,\vec{\ind}}
		\defas 
		\frac{\ind_1^{\dual}\cdots \ind_N^{\dual}}{\Dim/2}
		\Hepp{M,\vec{\ind}}
%		= \frac{\Hepp{M,\vec{\ind}}}{\Hepp{\Bond{N},\vec{\ind}}}
		\quad\text{where}\quad
		\ind_e^{\dual} = \tfrac{\Dim}{2}-\ind_e.
%		\prod_{e\in M} \left( \frac{\Dim}{2}-\ind_e \right).
		\label{eq:hepp-position-space}%
	\end{equation}
\end{definition}
In $\Dim=4$ dimensions with unit indices $\ind_e=1$ on all edges, $\Hepp{M}=2\HeppComp{M}$ differ only by a factor of $2$. Also recall $\Hepp{M,\vec{\ind}}|_{\Dim=0}=0$ from \autoref{cor:hepp(d=0)=0}, so the denominator in \eqref{eq:hepp-position-space} does not create a further pole. Instead, \autoref{lem:Hepp-Crapo} shows that
\begin{equation}
	\HeppComp{M,\vec{\ind}} \rightarrow (-1)^{\rank{M}+1} \Crapo{M}
	\quad\text{for}\quad
	\Dim\rightarrow 0.
	\label{eq:Hepp-Crapo-position}%
\end{equation}
\begin{example}
	For cycles and bonds (\autoref{fig:cycles-bonds}), the results \eqref{eq:hepp-cycle} and \eqref{eq:hepp-bond} translate to
	\begin{equation}
		\HeppComp{\Bond{N},\vec{\ind}} = 1
		\quad\text{and}\quad
		\HeppComp{\Cycle{N},\vec{\ind}} 
		= \frac{\ind_1^{\dual}\cdots \ind_N^{\dual}}{\ind_1 \cdots \ind_N}
		.
		\label{eq:heppcomp-cn-dn}%
	\end{equation}
\end{example}
The motivation for \autoref{def:Hepp-pos} is that the period \eqref{eq:period-mellin} is not exactly a Feynman integral. The Mellin integral lacks a prefactor $\prod_e 1/\Gamma(\ind_e)$ to become the actual Feynman integral in momentum space \cite{Nakanishi:GraphTheoryFeynmanIntegrals}.
A Fourier transform turns this into the integral
\begin{equation}
	\PeriodComp{G,\vec{\ind}}
	\defas \left(
	\prod_{v \neq 0,1} \int_{\R^{\Dim}} \frac{\td[\Dim] \vec{z}_v}{\pi^{\Dim/2}}
	\right)
	\left(
	\prod_{e=\set{v,w} \in \EG{G}} \frac{1}{\norm{\vec{z}_{v}-\vec{z}_{w}}^{2\ind_e^{\dual}}}
	\right)_{\vec{z}_0 = \vec{0},\ \vec{z}_1 = \uv{1}}
	\label{eq:period-position-space}%
\end{equation}
over position vectors $\vec{z}_v \in \R^{\Dim}$ associated with each vertex, similar to \cite[Definition~3.1]{Schnetz:GraphicalFunctions}. Two arbitrary vertices (labelled $0$ and $1$) get fixed positions: $\vec{0}$ and a unit vector $\uv{1}$. The Fourier transform introduces further $\Gamma$-functions \cite{Panzer:PhD}, and the precise relation is
\begin{equation}
	\PeriodComp{G,\vec{\ind}}
	= 
%	\frac{\Gamma(\Dim/2)}{\prod_e \Gamma(\ind_e^{\dual})}
	\frac{\Gamma(\Dim/2)}{\Gamma(\ind_1^{\dual})\cdots\Gamma(\ind_N^{\dual})}
	\Period{G,\vec{\ind}}
	.
	\label{eq:period-to-position}%
\end{equation}
The replacement $\Gamma(s)\mapsto 1/s$ observed in \autoref{ex:hepp-bubble} suggests that \eqref{eq:hepp-position-space} is the correct tropical analogue of the Feynman integral in position space. Indeed, the tropical limit
\begin{equation}
	\HeppComp{G,\vec{\ind}} = \lim_{\varepsilon \rightarrow 0} \PeriodComp{G,\varepsilon \vec{\ind}}
	\label{eq:tropical-limit-positionspace}%
\end{equation}
follows from \eqref{eq:tropical-limit}. The symmetries of the Hepp bound are therefore implied by the symmetries of periods as functions of the indices $\vec{\ind}$. We give independent proofs that reduce completion and twist to a purely combinatorial classification of the poles of the rational functions $\HeppComp{G,\vec{\ind}}$, combined with the factorization \eqref{eq:hepp-residue} of their residues.

Note that, in contrast to the Hepp bound $\Hepp{G}=2\HeppComp{G}$, the period in $\Dim=4$ dimensions with unit indices $\ind_e=1$ is unchanged in position space: $\Period{G} = \PeriodComp{G}$.

\subsection{Duality}
\label{sec:duality}
\begin{figure}
	\centering%
	$ G=\Graph[0.6]{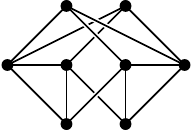} = \Graph[0.4]{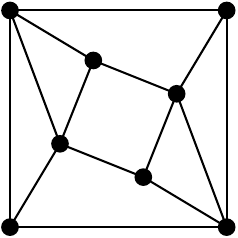} \quad\rightarrow\quad \Graph[0.4]{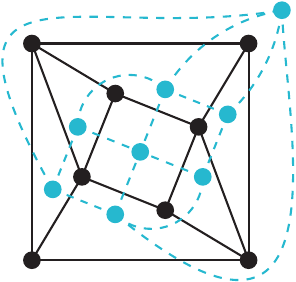} \quad\rightarrow\quad G^{\dual}=\Graph[0.35]{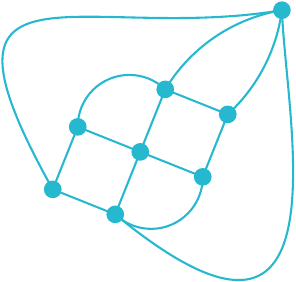}=\Graph[0.5]{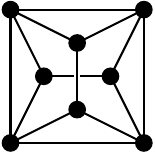}$ %
	\caption{A planar graph $G$, a planar embedding of the same graph, the dual indicated by dashed lines, and another (non-planar) drawing of this dual graph $G^{\dual}$.}%
	\label{fig:duality}%
\end{figure}

Every matroid $M$ (with corank function $\ell$) has a \emph{dual matroid} $\DMat{M}$ \cite{Oxley:MatroidTheory} on the same ground set $E=E_M=E_{\DMat{M}}$, but with the corank function $\ell^{\dual}\colon 2^E \longrightarrow \Z_{\geq 0}$ defined by
\begin{equation}
	\ell^{\dual}(\gamma)
	= \loops{E\setminus \gamma} + \abs{\gamma} -\loops{E}.
	\label{eq:dual-corank}%
\end{equation}
\begin{example}
	The dual of the uniform matroid $\UM{n}{r}$ is $\UM{n}{n-r}$.
\end{example}
\begin{example}
For planar graphs $G$, choose a planar drawing. Construct the planar dual $G^{\dual}$ by assigning a vertex to each face and connecting neighbouring faces with edges, as in \autoref{fig:duality}. Then the cycle matroid $\GMat{G^{\dual}}$ is the dual of $\GMat{G}$.
\end{example}
\begin{proposition}\label{lem:hepp-duality}
	Let $M$ denote any matroid and set $\ind^{\dual}_e \defas \frac{\Dim}{2}-\ind_e$ for each edge. Then
	\begin{equation}
		\Hepp{M^{\dual},\vec{\ind}}
		= \Hepp{M,\vec{\ind}^{\dual}}.
		\label{eq:hepp-duality}%
	\end{equation}
\end{proposition}
\begin{proof}
	The superficial degree of convergence of a subset $\gamma \subseteq E$ in the dual matroid is
\begin{equation*}
	\sdc[\vec{\ind}]{\DMat{\gamma}}
	= \sum_{e\in \gamma} \ind_e - \tfrac{\Dim}{2} \Big( \abs{\gamma}-\loops{M} + \loops{M{\setminus}\gamma} \Big)
	= -\sum_{e\in \gamma} \ind_e^{\dual}+\tfrac{\Dim}{2} \loops{M} - \tfrac{\Dim}{2} \loops{M{\setminus}\gamma}
\end{equation*}
according to \eqref{eq:dual-corank}. The case $\gamma=E$ shows that $\sdc[\vec{\ind}]{\DMat{M}} = -\sdc[\vec{\ind}^{\dual}]{M}$ both vanish in the same dimension. Substitute $\frac{\Dim}{2} \loops{M}=\sum_{e\in M} \ind_e^{\dual}$ in the equation above to conclude that 
	$\sdc[\vec{\ind}]{\DMat{\gamma}}
	=\sdc[\vec{\ind}^{\dual}]{M{\setminus}\gamma}$.
	The contribution to $\Hepp{\DMat{M},\vec{\ind}}$ from any flag $\gamma_1\subsetneq \cdots \subsetneq \gamma_N$ in \eqref{eq:multi-hepp-from-sectors} therefore matches precisely the contribution of the complementary flag $E{\setminus}\gamma_{N-1} \subsetneq \ldots \subsetneq E{\setminus}\gamma_1 \subsetneq E$ to the Hepp bound $\Hepp{M,\vec{\ind}^{\dual}}$.
\end{proof}
\begin{example}
	The cycles $\Cycle{n}$ and the bonds $\Bond{n}$ in \autoref{fig:cycles-bonds} are duals of each other, and indeed their Hepp bounds \eqref{eq:hepp-cycle} and \eqref{eq:hepp-bond} are related by \eqref{eq:hepp-duality}.
\end{example}
\begin{remark}\label{rem:dual-1pi-flat}
	We can read \eqref{eq:dual-corank} as $\ell^{\dual}(\gamma) = \rank{M} - \rank{M{\setminus}\gamma}$. It follows that the bridgeless subsets of $\DMat{M}$ are precisely the complements of flats in $M$ and vice-versa. Hence under duality, the bridgeless flag formula \eqref{eq:hepp-1pi-flags} becomes the sum \eqref{eq:hepp-flat-flags} over flags of flats.
\end{remark}
The identity \eqref{eq:hepp-duality} is also very easy to prove with the Mellin integral: The bases of the dual are the complements 
$\Bases{\DMat{M}} = \setexp{E\setminus T}{T\in\Bases{M}}$
of the bases of $M$, such that
\begin{equation}
	\PsiTrop_{\DMat{M}}(\SP)
	= \max_{T \in \Bases{M}} \prod_{e \in T} \SP_e
	= \left( \prod_{e\in M} \SP_e \right) \max_{T \in \Bases{M}} \prod_{e \notin T} \frac{1}{\SP_e}
	= \left( \prod_{e\in M} \SP_e \right) \PsiTrop_M\left( (1/\SP_e)_{e\in M} \right).
	\label{eq:psitrop-dual}%
\end{equation}
Inversion of the Schwinger parameters $x_e \rightarrow 1/x_e$ thus transforms the integrands \eqref{eq:hepp-mellin} of $\Hepp{\DMat{M},\vec{\ind}}$ and $\Hepp{M,\vec{\ind}^{\dual}}$ into each other. This proves \eqref{eq:hepp-duality}, and upon replacing $\PsiTrop$ with $\PsiPol$, the same argument also shows that
\begin{equation}
	\Period{\DMat{M},\vec{\ind}}
	= \Period{M,\vec{\ind}^{\dual}}
	\label{eq:period-duality}%
\end{equation}
This well-known relation for periods is called \emph{Fourier identity} in \cite{Schnetz:Census}, because it corresponds to a Fourier transform in the position space integral \eqref{eq:period-position-space}.

\subsection{Products and 2-sums}
\begin{definition}\label{def:2sum}
	Suppose we are given two connected matroids $A$ and $B$, each with at least $3$ elements. Let further $e\in A$ and $f\in B$ denote a choice of edges.
	Then the \emph{2-sum} $M=A \TwoSum{e}{f} B$ is the matroid on the disjoint union $\EG{M} = \EG{A}{\setminus}e \sqcup \EG{B}{\setminus}f$ with bases
\begin{equation}
	\Bases{M}
	\defas
%	\setexp{S{\setminus}e \sqcup T{\setminus}f}{
%		S \in \Bases{A}\ \text{and}\ T \in \Bases{B}\ \text{and}\ 
%		\abs{(S\cup T)\cap \set{e,f}}=1
%	}
	\big\{S \sqcup T\colon\ 
%		S \in \Bases{A} \ \text{and} \ T\cup\set{f} \in \Bases{B}
		S \in \Bases{A{\setminus}e}\ \text{and}\ T\in \Bases{B/f}
		\quad \text{or}\quad
%		S\cup\set{e} \in \Bases{A} \ \text{and} \ T \in \Bases{B}
		S \in \Bases{A/e}\ \text{and}\ T\in \Bases{B{\setminus}f}
	\big\}
	.
	\label{eq:2sum-bases}%
\end{equation}
%	where $S\subseteq A{\setminus}e$ and $T\subseteq B{\setminus}f$.
%	The circuits (minimal dependent sets) of the $2$-sum are
%	\begin{equation*}
%		\Circuits{M}%{A \TwoSum{e}{f} B} 
%		= \Circuits{A{\setminus}{e}} \ \cup\  \Circuits{B{\setminus}{f}}
%		\ \cup\ 
%		\setexp{C{\setminus}{e} \sqcup D{\setminus}{f}}{e\in C \in \Circuits{A}\ \text{and}\  f\in D \in \Circuits{B}}.
%	\end{equation*}
\end{definition}
\begin{figure}
	\centering
	$\GGMat{\Graph[0.5]{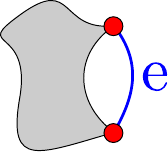}} \ \TwoSum{e}{f}\ \GGMat{\Graph[0.5]{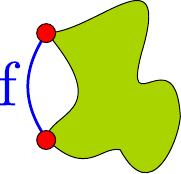} } \ =\ \GGMat{ \Graph[0.5]{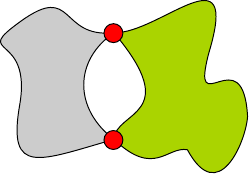} } = \GGMat{ \Graph[0.5]{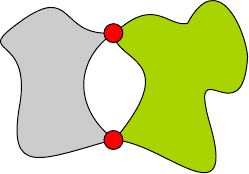} }$
%	$\Graph[0.6]{2cutL} \quad\TwoSum{e}{f}\quad \Graph[0.6]{2cutR} \quad=\quad \Graph[0.6]{2cut}$
	\caption{%
		The $2$-sum of two graphs.%
		%$S=\Graph[0.2]{2cutL}$ and $T=\Graph[0.2]{2cutR}$
	}%
	\label{fig:product}%
\end{figure}
We illustrate the $2$-sum for graphical matroids in \autoref{fig:product}.
It amounts to taking the disjoint union of $A{\setminus}e$ and $B{\setminus}f$, followed by the identification $v\sim v'$ and $w\sim w'$ of the endpoints of $e=\set{v,w}$ and $f=\set{v',w'}$. As the figure shows, we can flip one side ($v'\leftrightarrow w'$) and thus obtain two different graphs with the same cycle matroid \cite{Whitney:2isomorphic,Truemper:OnWhitney}.
\begin{example}[\autoref{fig:series-parallel}]
	The $2$-sum $M \TwoSum{e}{f} \UM{k+1}{k}$ with a cycle $\UM{k+1}{k} \cong \GMat{\Cycle{k+1}}$ replaces $e$ in $M$ by a path with $k$ edges. This is called a \emph{series operation}. A \emph{parallel operation} is the $2$-sum $M \TwoSum{e}{f} \UM{k+1}{1}$ with a bond, which replaces $e$ by $k$ parallel edges.
\end{example}
From \eqref{eq:2sum-bases} we see that the rank of a $2$-sum is $\rank{M} = \rank{A}+\rank{B}-1$, and thus
\begin{equation}
	\loops{A \TwoSum{e}{f} B}
	= \loops{A} + \loops{B} -1.
	\label{eq:loops-2sum}%
\end{equation}
Suppose we have assigned indices $\vec{\ind}$ to $S \defas A{\setminus}e \subset M$ and indices $\vec{b}$ to $T \defas B{\setminus}f \subset M$.
They determine the dimension $\Dim$ where $\sdc{M}=0$.
If we set $\ind_e \defas \frac{\Dim}{2} - \sdc{S}$, then
\begin{equation*}
	\sdc{A} = \sdc{S} + \sdc{A/S} 
	= \sdc{S} + \ind_e - \tfrac{\Dim}{2} = 0
\end{equation*}
and similarly, we get $\sdc{B}=0$ for $b_f \defas \frac{\Dim}{2} - \sdc{T}$. Note that $b_f = \sdc{S}$ and $\ind_e = \sdc{T}$, because \eqref{eq:loops-2sum} shows that $\sdc{M} = \sdc{S}+\sdc{T} - \frac{\Dim}{2}$.
\begin{proposition}\label{thm:hepp-product}
	Given a 2-sum $M=A \TwoSum{e}{f} B$, let $\vec{a}$ and $\vec{b}$ denote variables indexed by the edges in $A{\setminus}{e}$ and $B{\setminus}{f}$, respectively. Set $\ind_e \defas \sdc{B{\setminus}f}$ and $b_f \defas \sdc{A{\setminus}e}$.
	Then
	\begin{equation}
		\HeppComp{}(M,\vec{a},\vec{b})
		= \HeppComp{A,\vec{a},a_e} \cdot \HeppComp{}(B,\vec{b},b_f)
		.
		\label{eq:Hepp-product}%
	\end{equation}
\end{proposition}
\begin{proof}
	By our definition, $A$ and $B$ are connected, and it follows that $M$ is connected. As an identity of rational functions, it suffices to prove \eqref{eq:Hepp-product} locally, so we may assume that the indices lie in the convergence cone of the Mellin integral \eqref{eq:hepp-mellin}.

	Let $\psi_A \defas \PsiTrop_{A{\setminus}e}$ and $\phi_A \defas \PsiTrop_{A/e}$ such that $\PsiTrop_A = \max\set{\SP_e \psi_A,\phi_A}$. We can write
	\begin{equation}
		\Hepp{A,\vec{\ind},\ind_e}
		=
		\int_{\Projective^{A{\setminus}e}} \frac{\Omega(\vec{\ind})}{(\phi_A)^{\Dim/2}}
		\int_0^{\infty} \frac{
			\SP_e^{\ind_e-1} \td \SP_e
		}{
			(\max\set{\SP_e \psi_A/\phi_A,1})^{\Dim/2}
		}
		= \frac{\Dim/2}{\ind_e \ind_e^{\dual}}
		\int_{\Projective^{A{\setminus}e}} \frac{\Omega(\vec{\ind})}{\psi_A^{\ind_e} \phi_A^{\ind_e^{\dual}}}
		\label{eq:hepp-mellin-after-one}%
	\end{equation}
	where $\ind_e^{\dual} = \frac{\Dim}{2} - \ind_e$, by performing the integral over $\lambda \defas \SP_e \psi_A/\phi_A$ as in \autoref{ex:hepp-bubble}:
	\begin{equation*}
		\int_0^{\infty} \frac{
			\lambda^{\ind_e-1} \td \lambda
		}{
			(\max\set{\lambda,1})^{\Dim/2}
		}
		=
		\int_0^{1} \lambda^{\ind_e-1}\td \lambda
		+\int_{1}^{\infty} \lambda^{\ind_e-\Dim/2-1}\td \lambda
		= \frac{1}{\ind_e} + \frac{1}{\Dim/2-\ind_e}
		= \frac{\Dim/2}{\ind_e \ind_e^{\dual}}.
	\end{equation*}
	Set also $\psi_B \defas \PsiTrop_{B{\setminus}f}$ and $\phi_B \defas \PsiTrop_{B/f}$, then according to \eqref{eq:2sum-bases}, we can write the tropical matroid polynomial of the $2$-sum as
	$ \PsiTrop_M = \max\set{\psi_A \phi_B,\phi_A\psi_B}= \psi_A \phi_B \max\set{1,\lambda}$ in terms of the coordinate $\lambda \defas (\phi_A/\psi_A)(\psi_B/\phi_B)$ on $\Projective^M$.
	Combining $\lambda$ with the forgetful maps $x\mapsto y\defas [x_i]_{i\in A{\setminus}e}$ and $x\mapsto z\defas [x_i]_{i\in B{\setminus}f}$, we obtain a change of variables
	\begin{equation*}
		\Projective^M \longrightarrow \Projective^{A{\setminus}e} \times \Projective^{B{\setminus}f} \times \R_{>0},
		\quad x \mapsto (y,z,\lambda).
		\tag{$\ast$}%
		\label{eq:product-coordinates}%
	\end{equation*}
	Since $\phi_A$ and $\psi_A$ are homogeneous of degrees $\loops{A}$ and $\loops{A}-1$, respectively, $\lambda$ is homogeneous of degree one in the variables $y$. Consequently, \eqref{eq:product-coordinates} is invertible with
	\begin{equation*}
		x(y,z,\lambda)
		= \left[ 
%			\lambda \frac{\psi_A(y)\phi_B(z)}{\phi_A(y)\psi_B(z)} y
			\frac{\lambda}{\lambda'(y,z)}y
			\colon z
		\right] \in \Projective^M
		\quad\text{for}\quad
		\lambda'(y,z) \defas \frac{\phi_A(y)\psi_B(z)}{\psi_A(y)\phi_B(z)}.
	\end{equation*}
	Under this rescaling, $\psi_A(x) = \psi_A(y) (\lambda/\lambda')^{\loops{A{\setminus}e}}$ and the volume form \eqref{eq:volume-form} is of degree $\sum_{i \in A{\setminus}e} \ind_e = \sdc{A{\setminus}e} + \frac{\Dim}{2} \loops{A{\setminus}e}$ in $y$. With $b_f = \sdc{A{\setminus}e}$, the integrand factorizes as
	\begin{equation*}
		\frac{\Omega(\vec{\ind},\vec{b})}{(\PsiTrop_M(x))^{\Dim/2}}
		= \frac{\Omega(\vec{\ind})}{\psi_A(y)^{\Dim/2}}
		\left( \frac{\psi_A(y)}{\phi_A(y)} \right)^{b_f}
		\wedge
		\frac{\Omega(\vec{b})}{\phi_B(z)^{\Dim/2}}
		\left( \frac{\phi_B(z)}{\psi_B(z)} \right)^{b_f}
		\wedge
		\frac{\lambda^{b_f-1}\td \lambda}{(\max\set{1,\lambda})^{\Dim/2}}.
	\end{equation*}
	The integrals over $z\in \Projective^{B{\setminus}f}$ and $\lambda$ produce precisely $\Hepp{}(B,\vec{b},b_f)$ by the equivalent of \eqref{eq:hepp-mellin-after-one} for $B$.
	Since $b_f=\frac{\Dim}{2}-\ind_e=\ind_e^{\dual}$, the integral of the first term over $y$ gives $\Hepp{A,\vec{\ind},\ind_e}$ times $\ind_e^{\dual} b_f^{\dual}/(\Dim/2)$, again using \eqref{eq:hepp-mellin-after-one}. We have thus shown that
	\begin{equation*}
		\Hepp{}(M,\vec{\ind},\vec{b})
		=
		\ind_e^{\dual} \Hepp{A,\vec{\ind},\ind_e}
		\cdot
		b_f^{\dual} \Hepp{}(B,\vec{b},b_f)
		\cdot 2/\Dim,
	\end{equation*}
	which becomes the claim \eqref{eq:Hepp-product} in position space \eqref{eq:hepp-position-space}.
\end{proof}
\begin{example}\label{ex:K4-product}
	Recall that $\Hepp{K_4}=84$ and $\HeppComp{K_4}=42$ for unit indices in $\Dim=4$ dimensions. All $2$-sums of $K_4$ with itself are isomorphic, and we find that
	\begin{equation*}
		\HeppComp{\Graph[0.3]{5Rw}}
		=
		\HeppComp{\Graph[0.25]{w3small} \TwoSum{e}{f} \Graph[0.25]{w3small}}
		= \HeppComp{\Graph[0.25]{w3small}}^2
%		= 42^2 
		= 1764
		\quad\text{and}\quad
		\Hepp{\Graph[0.3]{5Rw}}
		= 3528.
	\end{equation*}
\end{example}
\begin{figure}
	\centering
%	$ \GGMat{\Graph[0.6]{2cutLsmall} } \TwoSum{e}{3} \GGMat{ \Graph[0.6]{c3glue} } = \GGMat{\Graph[0.6]{2cutseries}} $
	$ \Graph[0.7]{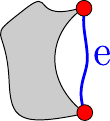} \ \TwoSum{e}{3} \ \Graph[0.7]{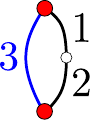} \ =\ \Graph[0.7]{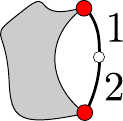}$
	\qquad\qquad
	$ \Graph[0.7]{2cutLsmall} \ \TwoSum{e}{3}\ \Graph[0.7]{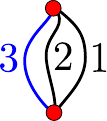} \ =\ \Graph[0.7]{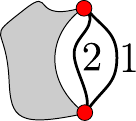} $
	\caption{Series and parallel operations on edge $e$, illustrated for graphical matroids.}%
	\label{fig:series-parallel}%
\end{figure}
\begin{corollary}[series--parallel reductions]
	If the matroid $S=M \TwoSum{e}{3} \UM{3}{2}$ is obtained from $M$ by replacing an edge $e$ with two edges $1$ and $2$ in series (\autoref{fig:series-parallel}), then
	\begin{equation}
		\Hepp{S,\ind_1,\ind_2,\vec{\ind}}
		= \frac{\ind_1+\ind_2}{\ind_1 \ind_2} \Hepp{M,\ind_1+\ind_2,\vec{\ind}}
		.
		\label{eq:series-hepp}%
	\end{equation}
	If $P=M \TwoSum{e}{3} \UM{3}{1}$ is obtained from $M$ by replacing $e$ with a pair of parallel edges, then
	\begin{equation}
		\Hepp{P,\ind_1,\ind_2,\vec{\ind}}
%		= \frac{\Dim-\ind_1-\ind_2}{(\Dim/2-\ind_1)(\Dim/2-\ind_2)} \Hepp{M,\ind_1+\ind_2-\tfrac{\Dim}{2},\vec{\ind}}
		= \frac{\ind_1^{\dual}+\ind_2^{\dual}}{\ind_1^{\dual}\ind_2^{\dual}} \Hepp{M,\ind_1+\ind_2-\tfrac{\Dim}{2},\vec{\ind}}
		.
		\label{eq:parallel-hepp}%
	\end{equation}
	In particular, the Hepp bound in position space is invariant under parallel operations:
	\begin{equation}
		\HeppComp{P,\ind_1,\ind_2,\vec{\ind}}
		=
		\HeppComp{M,\ind_1+\ind_2-\tfrac{\Dim}{2},\vec{\ind}}
		.
		\label{eq:parallel-heppcomp}%
	\end{equation}
\end{corollary}
\begin{proof}
	Apply \autoref{thm:hepp-product} to $B=\UM{3}{2}$ or $B=\UM{3}{1}$ and use \eqref{eq:heppcomp-cn-dn}.
\end{proof}
\begin{corollary}\label{cor:Crapo-2sum}
	Crapo's invariant of a $2$-sum is $\Crapo{A \TwoSum{e}{f} B} = \Crapo{A} \Crapo{B}$ \cite{Brylawski:CombinatorialSeriesParallel}, and
	series--parallel operations do not affect Crapo's invariant \cite[Propositions~4~and~5]{Crapo:HigherInvariant}.
\end{corollary}
\begin{proof}
	We use \eqref{eq:Hepp-Crapo-position} and apply the limit $\Dim\longrightarrow 0$ to \eqref{eq:Hepp-product}. Because of $\rank{M}-1 = \rank{A}-1 + \rank{B}-1$ from \eqref{eq:loops-2sum}, this proves the first claim. The series--parallel invariance is the special case $B=\GMat{\Cycle{n}}$ or $B=\GMat{\Bond{n}}$, with $\Crapo{B}=1$ due to \autoref{ex:Crapo-cycle-bond}.
\end{proof}
The product \eqref{eq:Hepp-product} is the exact analogue of the well-known relation \cite[Proposition~40]{Brown:PeriodsFeynmanIntegrals}
\begin{equation}
	\PeriodComp{}\big(A \TwoSum{e}{f} B,\vec{\ind},\vec{b}\,\big)
	= \PeriodComp{A,\vec{\ind},\ind_e}
	\cdot
	\PeriodComp{}(B,\vec{b},b_f)
	\label{eq:period-product}%
\end{equation}
for the period in position space. It can be proven in the same way as above; the only difference arises because $\PsiPol_M = \PsiPol_{A{\setminus}e} \PsiPol_{B/f} + \PsiPol_{A/e} \PsiPol_{B{\setminus}f}$ is a sum and not the maximum. Hence the $\lambda$-integrals become Euler beta functions, and the analogue of \eqref{eq:hepp-mellin-after-one} reads
\begin{equation*}
	\Period{A,\vec{\ind},\ind_e}
	= \frac{\Gamma(\ind_e) \Gamma(\ind_e^{\dual})}{\Gamma(\Dim/2)}
	\int_{\Projective^{A{\setminus}e}} \frac{\Omega(\vec{\ind})}{\PsiPol_{A{\setminus}e}^{\ind_e} \PsiPol_{A/e}^{\ind_e^{\dual}}}
	.
\end{equation*}

The product formulas suggest a unique factorization for matroids with respect to $2$-sums, and indeed this was achieved in \cite{CunninghamEdmonds:Decomposition}. We review this result and related terminology.
\begin{definition}
	A \emph{(Tutte) $k$-separation} of a matroid $M$ is a partition $\EG{M} = S \sqcup T$ of its edges such that $\abs{S},\abs{T} \geq k$ and
	$
		\rank{S} + \rank{T} \leq \rank{M} + k-1
	$. We say that $M$ is \emph{$n$-connected} if it has no $k$-separation where $1\leq k<n$.
\end{definition}
With this definition, every matroid is $1$-connected, and a \emph{connected} matroid according to \autoref{def:direct-sum} is called $2$-connected.
For a connected graph $G$, a $k$-separation of $\GMat{G}$ is an edge partition with $\abs{S},\abs{T} \geq k$ such that $S$ and $T$ meet in at most $k$ vertices \cite[Theorem~3]{Cunningham:OnMatroidCon}. So \autoref{def:separation} describes the special case of $1$-separations, and \autoref{fig:product} shows $2$-separations.

In \cite{CunninghamEdmonds:Decomposition}, $2$-separations are called \emph{splits}, and a matroid is called \emph{prime} if it is $3$-connected. Every $2$-sum decomposition $M\cong A \TwoSum{e}{f} B$ implies the existence of a split with $S=A{\setminus}e$ and $T=B{\setminus f}$. Conversely, every split arises in this way and implies a decomposition of $M$ into minors $A$ and $B$ of $M$. Therefore, a prime matroid admits no $2$-sum decomposition.

It follows that a connected matroid is decomposable by $2$-sums into prime matroids. Apart from the order of performing the splits, this decomposition is almost unique. The only ambiguity arises from cycles $\Cycle{n+m} \cong \Cycle{n+1} \TwoSum{e}{f} \Cycle{m+1}$, and similarly bonds, which allow several decompositions. The unique factorization result is \cite[Theorem~18]{CunninghamEdmonds:Decomposition}:
\begin{theorem}\label{thm:2sum-decomposition}%
	Every connected matroid has a unique minimal $2$-sum decomposition into bonds, cycles, and $3$-connected matroids.
\end{theorem}
The Hepp bound implements this decomposition as an actual factorization of rational functions. Cycles and bonds \eqref{eq:heppcomp-cn-dn} have only linear factors in the numerator and denominator, which may be partitioned in several ways. In contrast, the numerator of $\Hepp{M,\vec{\ind}}$ for a $3$-connected matroid does not seem to factorize.

\subsection{Completion}
\label{sec:completion}

The completion symmetry is the invariance of the integral \eqref{eq:period-position-space} in position space under a conformal transformation $\vec{z}\mapsto \vec{z}/\norm{\vec{z}}^2$. It is very useful in particle physics as it equates the periods of many non-isomorphic graphs \cite[section~5]{Broadhurst:5loopsbeyond}. Completion applies only to graphs and requires that, at each vertex, the dual indices sum up to the dimension.
\begin{definition}\label{def:excess}
	Let $\ind_e^{\dual} = \frac{\Dim}{2} - \ind_e$ denote the dual indices.
The \emph{excess} at a vertex $v$ is
\begin{equation}
	\exc{v} \defas \Dim -
	\sum_{e\colon v\in e} 
	\ind_e^{\dual}
	%\left( \frac{\Dim}{2} - \ind_e \right)
	\label{eq:excess}%
\end{equation}
where the sum runs over all edges incident to $v$. We call $v$ \emph{conformal} if $\exc{v}=0$, and a \emph{conformal graph} is a graph $G$ together with indices $\vec{\ind}$ such that all vertices are conformal.
\end{definition}
\begin{example}\label{ex:Kn-conformal}
	Every $k$-regular graph with $k\geq 3$ and unit indices $\ind_e=1$ on all edges is conformal in $\Dim=2k/(k-2)$ dimensions. For $k=4$, this dimension is $\Dim=4$.
\end{example}
\begin{example}\label{ex:completion}
	For any graph $G$ with logarithmic indices, its \emph{completion} \cite{Schnetz:Census,Schnetz:GraphicalFunctions} is a conformal graph $H$ with one additional vertex `$\infty$' such that $G=H{\setminus}\infty$. To make $v\in\VG{G}$ with $\exc{v}\neq 0$ conformal, $H$ has an edge $e$ from $v$ to $\infty$ with $\ind_{e} \defas \Dim/2-\exc{v}$. The new vertex $\infty$ is conformal by virtue of $\exc{\infty}=-\sdc{G}=0$ from \eqref{eq:sdc-exc} below.
\end{example}
We can express the convergence degree of a subgraph in terms of excesses. Let $G[S]$ denote the subgraph of $G$ that is \emph{induced} by the vertex set $S\subseteq \VG{G}$. It contains precisely those edges of $G$ which have both endpoints in $S$.
\begin{lemma}\label{lem:sdc-exc}%
	Suppose that $\gamma=G[S]\subseteq G$ is a connected induced subgraph. Let $C\subset \EG{G}$ denote the edges with precisely one endpoint in $S$. Then
	\begin{equation}
		\sdc{\gamma}
		= \frac{1}{2} \left(
			\sum_{v\in S} \exc{v} +\sum_{e\in C} \ind_e^{\dual} - \Dim%\nCG{\gamma}
		\right)
		.
		\label{eq:sdc-exc}%
	\end{equation}
\end{lemma}
\begin{proof}
	The sum $\sum_{v\in S} \sum_{e: v\in e} \ind_e^{\dual}$ counts edges in $\gamma$ twice and edges in $C$ once. Hence,
	\begin{equation*}
		\sum_{v\in S} \exc{v} - \Dim \abs{S} + \sum_{e\in C} \ind_e^{\dual}
		= -2\sum_{e\in \gamma} \ind_e^{\dual}
		= -\Dim \abs{\gamma} + 2\sum_{e\in\gamma} \ind_e
		= 2\sdc{\gamma} + \Dim\left( \loops{\gamma}-\abs{\gamma} \right)
	\end{equation*}
	and we conclude using \eqref{eq:euler}.
\end{proof}
In particular, a conformal graph has $\sdc{G}=-\Dim/2$ and is not logarithmic unless $\Dim=0$.
However, the complement $G{\setminus}v$ of any vertex in a conformal graph is always logarithmic with $\sdc{G{\setminus}v}=-\exc{v}/2=0$. We can therefore consider the periods of such complements.
\begin{theorem}
	\label{thm:period-completion}%
	If $G$ is conformal, then $\PeriodComp{G{\setminus}v,\vec{\ind}}$ is the same for all vertices $v$.
\end{theorem}
The proof for positive integer dimensions in \cite[Definition~and~Theorem~2.2]{Schnetz:Census} applies the inversion $\vec{z}\mapsto \vec{z}/\norm{z}^2$ to the vertex coordinates in position space \eqref{eq:period-position-space}. In non-integer dimensions, these integrals over $\vec{z}\in\R^{\Dim}$ can be defined using \emph{dimensional regularization} as explained for example in detail in \cite{Collins}. The inversion proof then still applies.

We will prove the same invariance for the Hepp bound. Recall that the Mellin integrals \eqref{eq:hepp-mellin} and \eqref{eq:period-mellin} have poles. The equalities $\PeriodComp{G{\setminus}v,\vec{\ind}} = \PeriodComp{G{\setminus}w,\vec{\ind}}$ of \autoref{thm:period-completion} are to be understood as identities of meromorphic functions on the vector space
\begin{equation*}
	\bigcap_{u\in \VG{G}} \setexp{\vec{\ind}}{\exc{u}=0}
	\subset \C^{\abs{G}}
\end{equation*}
of conformal indices on $G$. This space has dimension $\abs{G}-\abs{\VG{G}} + 1=\loops{G}$, because one of the constraints determines the dimension $\Dim$. Even though a complement $G{\setminus}v$ has fewer edges than $G$, we still write $\vec{\ind}$ to denote its indices and to suggest that the indices of $G{\setminus}w$ are determined (through conformality in $G$) by the indices of $G{\setminus}v$.
\begin{figure}
	\centering
	$K_4=\Graph[0.45]{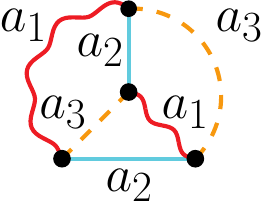}$
	\quad
	$K_4{\setminus}v = \Graph[0.45]{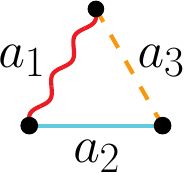}$
	\quad
	$\HeppComp{K_4{\setminus}v,\vec{a}} =\displaystyle \frac{(a_1+a_2)(a_1+a_3)(a_2+a_3)}{a_1 a_2 a_3}$%
	\caption{The complete graph $K_4$ with the most general conformal indices.}%
	\label{fig:K4-rainbow}%
\end{figure}%
\begin{theorem}\label{thm:hepp-completion}
	For any two vertices $v,w$ of a biconnected graph $G$, the Hepp bounds
	$\HeppComp{G{\setminus}v,\vec{\ind}}=\HeppComp{G{\setminus}w,\vec{\ind}}$ agree as rational functions on the space of conformal indices.
\end{theorem}
\begin{corollary}
	If $G$ is regular and $G{\setminus}v$, $G{\setminus}w$ are {\PlogDiv}, then $\Hepp{G{\setminus}v}=\Hepp{G{\setminus}w}$.
\end{corollary}%
\begin{example}\label{ex:completion-K4}
	The complete graph $K_4$ is conformal precisely when the indices $\ind_e=\ind_f$ of all non-adjacent edge pairs $\set{e,f}$ coincide. The conformal indices thus define a \emph{rainbow colouring}, where every vertex touches exactly one edge of each index (see \autoref{fig:K4-rainbow}). All complements give the same cycle $\Cycle{3}$, so \autoref{thm:hepp-completion} is trivial for $G=K_4$.
\end{example}
To prove the theorem, we first reduce it to the case of complete graphs $K_n$. If $G$ has $n$ vertices, define indices $\vec{c}$ on $K_n$ as follows: The edge between vertices $i$ and $j$ receives
\begin{equation}
%	c_{ij}\defas \tfrac{\Dim}{2}-c_{ij}^{\dual}
%	\quad\text{for}\quad
	c_{ij}^{\dual} 
	= \tfrac{\Dim}{2}-c_{ij}
	\defas \ind_{e_1}^{\dual} + \cdots + \ind_{e_k}^{\dual},
	\label{eq:indices-to-Kn}%
\end{equation}
where $e_1,\ldots,e_k \in \EG{G}$ denote all edges in $G$ between $i$ and $j$ (there may be none, one, or several such edges). This assignment ensures that $\vec{c}$ are conformal indices for $K_n$. So if we know $\HeppComp{K_n{\setminus}v,\vec{c}}=\HeppComp{K_n{\setminus}w,\vec{c}}$, then the theorem for $G$ follows from the identity
\begin{equation}
	\HeppComp{G{\setminus}v,\vec{\ind}} = \HeppComp{K_n{\setminus}v,\vec{c}}
	\label{eq:hepp-from-Kn}%
\end{equation}
and its analogue for $w$. We already saw in \eqref{eq:parallel-heppcomp} that several parallel edges ($k\geq 2$) can be replaced by a single edge with the dual index \eqref{eq:indices-to-Kn}. Nothing changes for edges with $k=1$, so \eqref{eq:hepp-from-Kn} is clear once we recognize that adding an edge with weight $c_{ij} = \Dim/2$ in the case $k=0$ has no effect on the Hepp bound. Indeed, due to the factors $\ind_e^{\dual}$ in \eqref{eq:hepp-position-space}, the pole for deletion of an edge in \eqref{eq:hepp-res-e} amounts to the finite limit
\begin{equation*}
	\HeppComp{G\sqcup\set{e},\vec{\ind},\Dim/2}
	=
	\lim_{\ind_e \rightarrow \Dim/2} \HeppComp{G\sqcup\set{e},\vec{\ind},\ind_e}
	= \HeppComp{G,\vec{\ind}}.
\end{equation*}
\begin{lemma}\label{lem:completion-poles}
	For a biconnected graph $G$, the poles of the Hepp bound $\HeppComp{G,\vec{\ind}}$ in position space are in bijection with subsets $S\subsetneq \VG{G}$ of $\abs{S}\geq 2$ vertices with the property that the induced subgraph $\gamma=G[S]$ and its quotient $G/\gamma$ are biconnected.
\end{lemma}
\begin{proof}
	According to \autoref{cor:hepp-poles} and \autoref{lem:cycle-matroid-connected}, the poles of $\Hepp{G,\vec{\ind}}$ are identified with subgraphs $\gamma$ such that $\gamma$ and $G/\gamma$ are nonseparable. If $\gamma$ is not induced, every edge $e\in G{\setminus}\gamma$ with both endpoints in $\gamma$ becomes a self-loop in the quotient $G/\gamma$. To avoid a separation, $G/\gamma\cong\Graph[0.3]{1rose}$ can only consist of one self-loop on its own, so $\gamma=G{\setminus}e$ is an edge complement. But the corresponding poles at $\sdc{G{\setminus}e} = \ind_e^{\dual}=0$ are cancelled in position space by the numerators in \eqref{eq:hepp-position-space}. So only induced $\gamma$ contribute poles to $\HeppComp{G,\vec{\ind}}$.
\end{proof}
In a complete graph $G=K_n$, all induced $\gamma=G[S]$ and their quotients are biconnected, so $\HeppComp{K_n,\vec{\ind}}$ has precisely $2^n-n-2$ poles, one for each subset $S\subsetneq \set{1,\ldots,n}$ with $\abs{S}\geq 2$.
To identify poles in other graphs, it may be useful that the quotient $G/G[S]$ is biconnected if and only if the complement $G{\setminus}S = G[\VG{G}{\setminus}S]$ is a connected graph.
\begin{figure}
	\centering
	$G=\Graph[0.7]{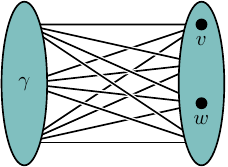}$
	\quad
	$H=\Graph[0.7]{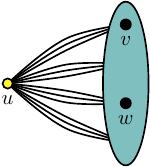}$
	\qquad
	$G=\raisebox{1.7mm}{$\Graph[0.7]{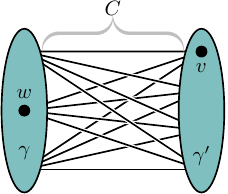}$}$
	\quad
	$H=\Graph[0.7]{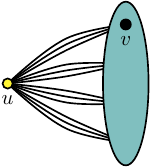}$
	\caption{An induced subgraph $\gamma=G[S]\subsetneq G{\setminus}v$ in the cases $w\notin S$ and $w\in S$.}%
	\label{fig:completion-proof}%
\end{figure}
\begin{proof}[Proof of \autoref{thm:hepp-completion}]
	We perform an induction over the number $n$ of vertices of $G$. The start at $n=3$ is simple: All vertex complements of $K_3=\gTri$ are single edges $K_3{\setminus}v \cong \gEdge$ with the same Hepp bound $\HeppComp{\Bond{1},\vec{\ind}}=1$ from \eqref{eq:heppcomp-cn-dn}.

	Now consider the complete graph $G=K_{n+1}$ and suppose that we already established the theorem for all graphs with $\leq n$ vertices. Call $v=n+1$ such that $G{\setminus}v=K_n$. Recall from \autoref{ex:completion} that the indices on $K_n$ are unconstrained and parametrize the entire space of conformal indices on $G$, so $\HeppComp{K_n,\vec{\ind}}$ has simple poles in bijection with subsets $S\subsetneq \set{1,\ldots,n}$ of size $\abs{S}\geq 2$. 
	In position space, the factorization \eqref{eq:hepp-residue} reads
	\begin{equation*}
		\Res_{\sdc{\gamma}=0} \HeppComp{K_n,\vec{\ind}}
%		= \tfrac{\Dim}{2} \cdot \HeppComp{\gamma,\vec{\ind}}|_{\sdc{\gamma}=0} \cdot \HeppComp{K_n/\gamma,\vec{\ind}}|_{\sdc{\gamma}=0}
		= \tfrac{\Dim}{2} \left[\,
			\HeppComp{\gamma,\vec{\ind}}
			\cdot 
			\HeppComp{K_n/\gamma,\vec{\ind}}
		\right]_{\sdc{\gamma}=0}
		\tag{$\sharp$}%
		\label{eq:completion-proof-residue}%
	\end{equation*}
	and we first consider the case $\gamma=G[S]$ with $w\notin S$ (see \autoref{fig:completion-proof}). Then $\gamma$ is also a subgraph of $G{\setminus}w$, and we set $H\defas G/\gamma$ such that $K_n/\gamma = H{\setminus}v$. The excess of the vertex $u$ representing $\gamma$ in $H$ is $-2\sdc{\gamma}$ due to \eqref{eq:sdc-exc}. So on the hyperplane $\sdc{\gamma}=0$, the graph $H$ is itself conformal. Since $H$ has at most $n$ vertices, we know by induction that $\HeppComp{H{\setminus}v,\vec{\ind}}=\HeppComp{H{\setminus}w,\vec{\ind}}$. But $H{\setminus}w=(G{\setminus}w)/\gamma$, so we learn
	\begin{equation*}
		\Res_{\sdc{\gamma}=0} \HeppComp{G{\setminus}v,\vec{\ind}}
		= \Res_{\sdc{\gamma}=0} \HeppComp{G{\setminus}w,\vec{\ind}}
%		\tag{$\ast$}%
%		\label{eq:completion-proof-residues-equal}%
	\end{equation*}
	by comparing \eqref{eq:completion-proof-residue} with its analogue for $G{\setminus}w$ instead of $G{\setminus}v=K_n$.
	Now suppose that $w\in S$, then $\gamma=G[S]$ is not anymore a subgraph of $G{\setminus}w$. To see the pole, set $\gamma'\defas G[T]$ for the complement $T\defas\set{1,\ldots,n+1}\setminus S$ and note that $\sdc{\gamma} = \sdc{\gamma'} = (\sum_{e\in C} \ind_e^{\dual} -\Dim)/2$ by \eqref{eq:sdc-exc}, where $C$ are the edges between $S$ and $T$. The residue formula gives
	\begin{equation*}
		\Res_{\sdc{\gamma}=0} \HeppComp{G{\setminus}w,\vec{\ind}}
		= \Res_{\sdc{\gamma'}=0} \HeppComp{G{\setminus}w,\vec{\ind}}
%		= \tfrac{\Dim}{2} \cdot \HeppComp{\gamma',\vec{\ind}}|_{\sdc{\gamma'}=0} \cdot \HeppComp{G/\gamma',\vec{\ind}}|_{\sdc{\gamma'}=0}
		= \tfrac{\Dim}{2} \left[\,
			\HeppComp{\gamma',\vec{\ind}}\cdot
			\HeppComp{G{\setminus}w/\gamma',\vec{\ind}}
		\right]_{\sdc{\gamma'}=0}
		.
		\tag{$\flat$}%
		\label{eq:completion-proof-residue2}%
	\end{equation*}
	This time, we delete from $H=G/\gamma$ the vertex $u$ that represents $\gamma$, to conclude by induction and $H{\setminus}u=G{\setminus}S=G[T]=\gamma'$ that $\HeppComp{G{\setminus}v/\gamma,\vec{\ind}}=\HeppComp{H{\setminus}v,\vec{\ind}} = \HeppComp{\gamma',\vec{\ind}}$.
	Similarly, set $H' \defas G/\gamma'$ such that $H'{\setminus}u'=G{\setminus}T = G[S]=\gamma$ for the vertex $u'$ that is $\gamma'$ in $H'$.
	By induction, we find that $\HeppComp{G{\setminus}w/\gamma',\vec{\ind}} = \HeppComp{H'{\setminus}w,\vec{\ind}}=\HeppComp{\gamma,\vec{\ind}}$ and conclude that \eqref{eq:completion-proof-residue} and \eqref{eq:completion-proof-residue2} coincide.

	In summary, we have shown that both Hepp bounds $\HeppComp{G{\setminus}v,\vec{\ind}}$ and $\HeppComp{G{\setminus}w,\vec{\ind}}$ have the exact same residues on all poles. Therefore, their difference
	\begin{equation*}
		\Delta \defas
		\HeppComp{G{\setminus}v,\vec{\ind}}
		-
		\HeppComp{G{\setminus}w,\vec{\ind}}
	\end{equation*}
	is a rational function without poles, thus a polynomial. We also know that $\Delta$ is homogeneous of degree zero, so it must in fact be a constant rational number $\Delta \in \Q$. To show that $\Delta=0$, we specialize to unit indices $\ind_e=1$ on all edges as in \autoref{ex:Kn-conformal}, then we get trivially that $\HeppComp{G{\setminus}v}=\HeppComp{K_n}=\HeppComp{G{\setminus}w}$.
\end{proof}
The case with $4$ vertices is almost trivial as discussed in \autoref{ex:completion-K4}, but already for $n=5$ the completion relation gives an involved identity of rational functions.
\begin{figure}
	\centering
	$K_4=K_5{\setminus}5 = \Graph[0.45]{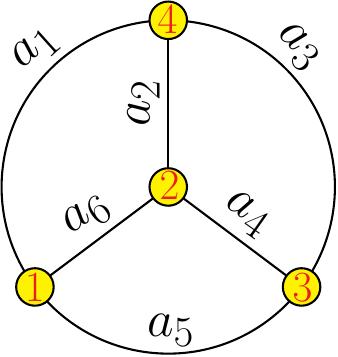}$
	\quad
	$K_5{\setminus}2 = \Graph[0.45]{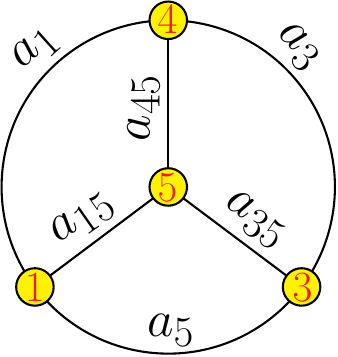}$
	\quad
	$\begin{aligned}
	\ind_{45} &= \tfrac{2(\ind_4+\ind_5+\ind_6)-\ind_1-\ind_2-\ind_3}{3}
	\\
	\ind_{35} &= \tfrac{2(\ind_1+\ind_2+\ind_6)-\ind_3-\ind_4-\ind_5}{3}
	\\
	\ind_{15} &= \tfrac{2(\ind_2+\ind_3+\ind_4)-\ind_1-\ind_5-\ind_6}{3}
	\end{aligned}$
	\caption{Two uncompletions of $K_5$ and the relationship of their indices.}%
	\label{fig:completion-K5}%
\end{figure}
\begin{example}\label{ex:completion-K5}
	The graph $K_4=K_5{\setminus}5$ with arbitrary indices $\ind_1,\ldots,\ind_6$ is logarithmic in dimension $\Dim=2(\ind_1+\cdots+\ind_6)/3$. The completion in \autoref{ex:completion} determines the indices of the edges connected to vertex $5=\infty$ such that $K_5$ is conformal. With the labels as in \autoref{fig:completion-K5}, the excess at vertex $4$ is $\exc{4} = \ind_1+\ind_2+\ind_3-\Dim/2$ and therefore
	\begin{equation*}
		\ind_{45} = \tfrac{\Dim}{2} - \exc{v_4}
%		= \Dim-\ind_1-\ind_2-\ind_3
		=\tfrac{2\ind_4+2\ind_5+2\ind_6-\ind_1-\ind_2-\ind_3}{3}
		= \sdc{K_4{\setminus}4}.
	\end{equation*}
	Similarly we find $\ind_{35}$ and $\ind_{15}$ as given in \autoref{fig:completion-K5}.
%	$\ind_{35} = \frac{2(\ind_1+\ind_2+\ind_6)-\ind_3-\ind_4-\ind_5}{3}$ 
%	and
%	$\ind_{15} = \frac{2(\ind_2+\ind_3+\ind_4)-\ind_1-\ind_5-\ind_6}{3}$.
	The identity $\HeppComp{K_5{\setminus}5,\vec{\ind}}=\HeppComp{K_5{\setminus}2,\vec{\ind}}$ is an explicit functional equation for the Hepp bound \eqref{eq:Hepp-K4-full} of $K_4$,
%	The completion symmetry $\HeppComp{K_5{\setminus}5,\vec{\ind}}=\HeppComp{K_5{\setminus}2,\vec{\ind}}$ can thus be explicitly written as
	\begin{equation*}
		\HeppComp{K_4,\ind_1,\ind_2,\ind_3,\ind_4,\ind_5,\ind_6}
		= \HeppComp{K_4,\ind_1,\ind_{45},\ind_3,\ind_{35},\ind_5,\ind_{15}}.
	\end{equation*}
\end{example}
\begin{remark}\label{rem:HeppComp-2cut-zero}
	We only get non-zero Hepp bounds from conformal graphs that are $3$-vertex-connected.
	For suppose that $G$ can be disconnected by deleting two vertices $v$ and $w$ (this implies a \emph{vertical} $2$-separation \cite{Cunningham:OnMatroidCon}). Then $G{\setminus}v$ has an articulation point $w$ and is therefore not biconnected, so $\HeppComp{G{\setminus}v,\vec{\ind}}=0$. If $G$ is conformal, this implies that the Hepp bounds of \emph{all} vertex complements $G{\setminus}u$ vanish, even those that \emph{are} biconnected. In these cases, the corresponding Hepp bound function $\HeppComp{G{\setminus}u,\vec{\ind}}$ vanishes on the solution space to the conformality constraints. For example,
	\begin{equation*}
		\HeppComp{\Graph[0.4]{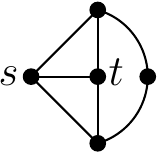},\vec{\ind}}_{\exc{s}=\exc{t}=0}
		= \HeppComp{\Graph[0.4]{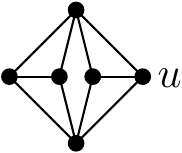}{\setminus}u,\vec{\ind}}
		= \HeppComp{\Graph[0.4]{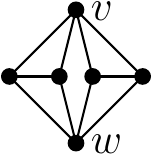}{\setminus}v,\vec{\ind}}
		= \HeppComp{\Graph[0.4]{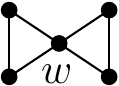},\vec{\ind}}
		= 0.
	\end{equation*}
\end{remark}
\begin{remark}\label{rem:crapo-completion}
	The completion symmetry $\HeppComp{G{\setminus}v,\vec{\ind}}=\HeppComp{G{\setminus}w,\vec{\ind}}$ does \emph{not} descend to Crapo's invariant, see the counterexample in \eqref{eq:crapo-completion-counterexample}.
	We cannot apply formula \eqref{eq:Hepp-Crapo-position} for the limit $\Dim\rightarrow 0$, because the restriction to conformal indices forces some subgraphs to diverge at $\Dim=0$. It follows from \eqref{eq:sdc-exc} that the complement $G{\setminus}I$ of any set $I\subset \VG{G}$ of independent (non-adjacent) vertices is divergent, violating the premise of \autoref{lem:Hepp-Crapo}.
%	Geometrically, the intersection of the conformality constraints with $\Dim=0$ lies in the singular locus:
%	\begin{equation*}
%		\set{\Dim=0} \cap \bigcap_{v\in\VG{G}} \set{\exc{v}=0}
%		\subset \bigcup_{\emptyset\neq\gamma\subsetneq\VG{G}} \set{\sdc{\gamma}=0}.
%	\end{equation*}
\end{remark}

\subsection{Twist}
\label{sec:twist}

Whenever a conformal graph can be disconnected by deleting $2$ or $3$ vertices, the Hepp bounds $\HeppComp{G{\setminus}v,\vec{\ind}}$ are forced to vanish (\autoref{rem:HeppComp-2cut-zero}) or they factorize by \autoref{thm:hepp-product}. If $G$ is $4$-vertex-connected, no such simplification applies, but the \emph{twist} from \cite[Section~2.6]{Schnetz:Census} provides identities between different graphs with a $4$-separation.
\begin{figure}
	\centering
	$\Graph[0.5]{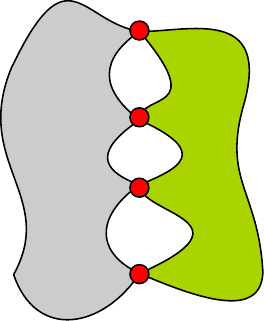}\quad \longrightarrow\quad
	\Graph[0.5]{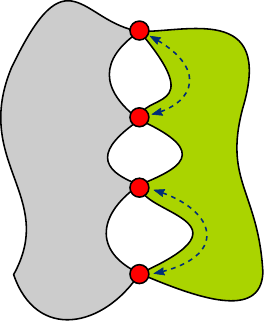}\quad \longrightarrow\quad 
	\Graph[0.5]{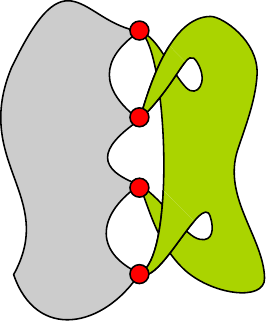} \cong
	\Graph[0.5]{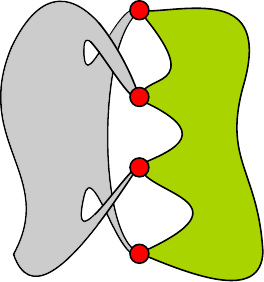} $
	\caption{A $4$-separation with a choice of a double transposition and the resulting twist.}%
	\label{fig:twist}%
\end{figure}
\begin{definition}\label{def:twist}%
	Suppose a graph $G$ has an edge bipartition $\EG{G}=S \sqcup T$ into two subgraphs with precisely four vertices $\set{p,q,r,s}$ in common (see \autoref{fig:twist}). The graph obtained by a double transposition $p\leftrightarrow q$ and $r \leftrightarrow s$ on $T$ (or $S$) is called a \emph{twist} of $G$.
\end{definition}
A single graph can have a lot of twists. Even for a fixed intersection set $\set{p,q,r,s}$, we can consider three different double transpositions, and whenever there are edges with both ends in $\set{p,q,r,s}$, we may distribute those edges arbitrarily among $S$ and $T$. Note that the construction of a twist $G'$ from $G$ gives a bijection of the edges $\EG{G'} \cong \EG{G}$.
\begin{theorem}\label{thm:twist-period}
	Consider a graph $G$ and a twist $G'$ with indices $\vec{\ind}$ that are conformal for $G$ and $G'$. Then the periods $\PeriodComp{G{\setminus}v,\vec{\ind}} = \PeriodComp{G'{\setminus}w,\vec{\ind}}$ coincide for all vertices $v,w$.
\end{theorem}
%By completion invariance (\autoref{thm:period-completion}) we conclude that $\PeriodComp{G{\setminus}v,\vec{\ind}} = \PeriodComp{G'{\setminus}w,\vec{\ind}}$ for any two vertices $v$ and $w$. 
This is an identity of meromorphic functions on the vector space of indices $\vec{\ind}$ that make both graphs $G$ and $G'$ conformal at the same time. We can specialize the indices to particular values, provided those stay away from singularities.
\begin{corollary}
	Suppose that $G_1$ and $G_2$ are {\PlogDiv} with unit indices. If their completions are twists of each other, then $\Period{G_1}=\Period{G_2}$.
\end{corollary}
A proof in $\Dim=4$ dimensions is given in \cite[Theorem~2.2]{Schnetz:Census}. With \eqref{eq:period-position-space} in dimensional regularization \cite{Collins}, the invariance \cite[Equation~(1.18)]{Schnetz:GraphicalFunctions} of graphical functions under double transposition shows that \autoref{thm:twist-period} holds for arbitrary dimensions, see the paragraph after \cite[Theorem~3.20]{Schnetz:GraphicalFunctions}. We demonstrate the analogue for the Hepp bound.
\begin{theorem}\label{thm:twist-hepp}%
	Suppose that $G'$ is a twist of $G$. Then the rational functions $\HeppComp{G{\setminus}v,\vec{\ind}}$ and $\HeppComp{G'{\setminus}w,\vec{\ind}}$ coincide on the space of indices $\vec{\ind}$ that are conformal for both $G$ and $G'$.
\end{theorem}
For the proof, we follow the strategy used above to establish completion invariance. We fill in missing edges with both ends in $S$ or $T$ by giving them the index $\Dim/2$, and we replace parallel edges within $S$ or $T$ by a single edge. It thus suffices to consider the case where $S\cong K_n$ and $T\cong K_m$ are complete, simple graphs on $n,m\geq 4$ vertices.

The corresponding graphs $G\cong G'$ are isomorphic to the quotient $G_{n,m}$ of the disjoint union $K_n\sqcup K_m$ by the identification of four vertex pairs. Any two of these shared vertices $Q=\set{p,q,r,s}$ are connected by precisely two edges in $G_{n,m}$, one coming from $S$ and the other coming from $T$.
\begin{figure}
	\centering%
	\renewcommand{\arraystretch}{1.5}%
%	$G_{4,4} = \Graph[0.4]{G44}$ \hfill
%	$S=\Graph[0.4]{Spqrs}$ \hfill
%	$T=\Graph[0.4]{Tpqrs}$ \hfill
%	$G'_{4,4} = \Graph[0.4]{G44twist}$%
	\begin{tabular}{c@{\qquad}c@{\qquad}c@{\qquad}c}
	$\Graph[0.4]{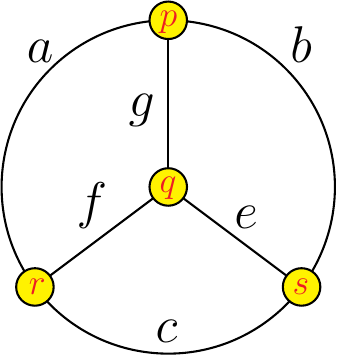}$ &
	$\Graph[0.4]{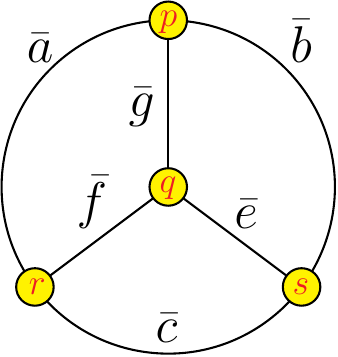}$ &
	$\Graph[0.4]{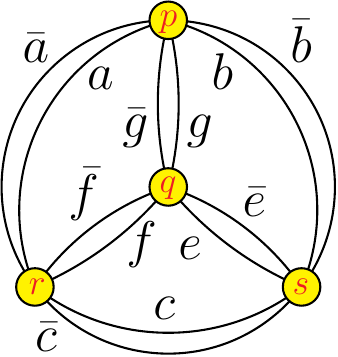}$ &
	$\Graph[0.4]{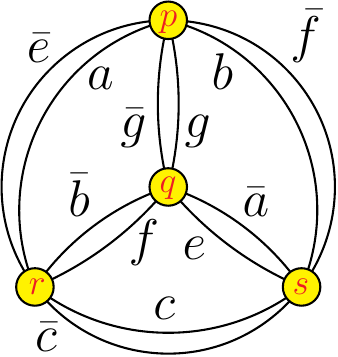}$ \\%[12mm]
	$S$ & $T$ & $G_{4,4}$ & $G_{4,4}'$ 
	\end{tabular}%
	\caption{The most general twist $p\leftrightarrow q$, $r\leftrightarrow s$ of a graph with four vertices.}%
	\label{fig:twist-44}%
\end{figure}
\begin{example}\label{ex:G44twist}
	The graph $G_{4,4}$ has only the four separating vertices and twelve edges. With the labels in \autoref{fig:twist-44}, the twist identity between $G_{4,4}{\setminus}q$ and $G'_{4,4}{\setminus}q$ claims that
	\begin{equation*}
		\HeppComp{\gTri,
			a+\bar{a}-\tfrac{\Dim}{2},
			b+\bar{b}-\tfrac{\Dim}{2},
			c+\bar{c}-\tfrac{\Dim}{2}
		}
		=\HeppComp{\gTri,
			a+\bar{e}-\tfrac{\Dim}{2},
			b+\bar{f}-\tfrac{\Dim}{2},
			c+\bar{c}-\tfrac{\Dim}{2}
		}.
		\tag{$\ast$}%
		\label{eq:G44twist}%
	\end{equation*}
	This is wrong for generic indices, but the conformality $\exc{r}=0$ in $G_{4,4}$ in $G_{4,4}'$ enforces
	\begin{equation*}
		a+c+f+\bar{a}+\bar{c}+\bar{f} = 2\Dim = a+c+f+\bar{b}+\bar{c}+\bar{e}
	\end{equation*}
	and therefore $\bar{a}+\bar{f}=\bar{b}+\bar{e}$. Similarly, we find $\bar{a}+\bar{b}=\bar{e}+\bar{f}$ from the constraints $\exc{p} = 0$. We conclude that $\bar{a}=\bar{e}$ and $\bar{b}=\bar{f}$, hence \eqref{eq:G44twist} is clearly true.
\end{example}
\begin{figure}
	\centering
	$\Graph[1.35]{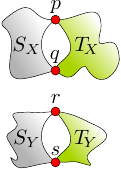}$ {\Large $\mapsto$}
	$\Graph[1.35]{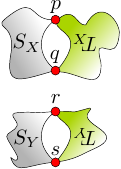}$
	\qquad\quad
	$\Graph[1.35]{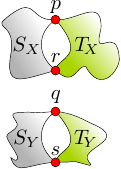}$ {\Large $\mapsto$}
	$\Graph[1.35]{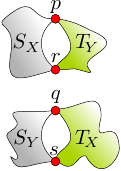}$
	\caption{The left shows the subgraphs $G[X]$ and $G[Y]$ in $G$ and in the twist $G'$, when $\set{p,q}\subseteq X$ and $\set{r,s}\subseteq Y$. In the case $\set{p,r}\subseteq X$ and $\set{q,s}\subseteq Y$ on the right, the twist $p\leftrightarrow q$, $r\leftrightarrow s$ swaps the parts $T_X$ and $T_Y$.}%
	\label{fig:twistproof}%
\end{figure}
\begin{proof}[Proof of \autoref{thm:twist-hepp}]
	As for completion, we perform an induction over the number of vertices of $G$. The minimum of $4$ vertices is \autoref{ex:G44twist}. In the induction step, it suffices to consider $G=G_{n,m}$ with $n+m-4\geq 5$ vertices, and we may assume that the twist identity is already proven for all graphs with fewer than $n+m-4$ vertices. 

	The function $F_G(\vec{\ind}) \defas \HeppComp{G{\setminus}v,\vec{\ind}}$ is independent of $v$ by \autoref{thm:period-completion}. According to \autoref{lem:completion-poles}, the poles of $F_G$ correspond to vertex bipartitions $\VG{G} = X\sqcup Y$ where $\abs{X},\abs{Y} \geq 2$. Let $\gamma \defas G[X]$ and recall from \eqref{eq:hepp-residue} that the corresponding residue is
	\begin{equation*}
		\Res_{\sdc{\gamma}=0} F_G(\vec{\ind})
		= \tfrac{\Dim}{2} %\Big[ 
		\HeppComp{\gamma,\vec{\ind}}
		F_{G/\gamma}(\vec{\ind})
%		\Res_{\sdc{G[X]}=0} F_G(\vec{\ind})
%		= \tfrac{\Dim}{2} %\Big[ 
%		\HeppComp{G[X],\vec{\ind}}
%		F_{G/G[X]}(\vec{\ind})
		%\Big]_{\sdc{G[X]}=0}
%		= \tfrac{\Dim}{2}
%		\HeppComp{G[X],\vec{\ind}}
%		\HeppComp{G[Y],\vec{\ind}}
		.
		\tag{$\dagger$}%
		\label{eq:res-twist-proof}%
	\end{equation*}
	Let $X$ denote the part of the partition that contains $p$. If $X$ intersects $Q\defas\set{p,q,r,s}$ only in $p$, then removing $p$ disconnects $\gamma$ if $X{\setminus}Q$ contains vertices on both sides $S$ and $T$ of the $4$-separation. For a non-zero residue we must have $\gamma\subseteq S$ or $\gamma \subseteq T$. In this situation, $\gamma\cong \gamma'\defas G'[X]$ is a subgraph of both $G$ and $G'$. Furthermore, the quotient $G'/\gamma'$ is a twist of $G/\gamma$, so we know that $F_{G/\gamma} = F_{G'/\gamma'}$ by induction. It then follows from \eqref{eq:res-twist-proof} that $F_G$ and $F_{G'}$ have the same residue at $\sdc{\gamma}=0$.

	If $\abs{X\cap Q}=3$, then we can apply the symmetric argument to $Y$ to show that the residues of $F_G$ and $F_{G'}$ at $\sdc{\gamma}=0$ coincide; recall that $\sdc{G[X]} = \sdc{G[Y]}$ by \eqref{eq:sdc-exc}.

	The case $X\supseteq Q$ means $Y\cap Q=\emptyset$, and we proceed in the same way. Connectedness of $Y$ implies that $Y\subseteq S$ or $Y\subseteq T$, and $G[Y]=G'[Y]$ are literally the same graphs.

	It remains to compare the residues when $\abs{X\cap Q} = \abs{Y\cap Q}=2$. Compute $F_{G/\gamma}$ by deleting the special vertex corresponding to $\gamma$, so we can rewrite \eqref{eq:res-twist-proof} symmetrically as
	\begin{equation*}
		\Res_{\sdc{G[X]}=0} F_G(\vec{\ind})
		= \tfrac{\Dim}{2}
		\HeppComp{G[X],\vec{\ind}}
		\HeppComp{G[Y],\vec{\ind}}
		.
		\tag{$\ddagger$}%
		\label{eq:res-twist-proof-XY}%
	\end{equation*}
	If $X\cap Q = \set{p,q}$, then $G[X]$ and $G'[X]$ differ only by turning around one of the two sides $S_X\defas \gamma\cap S$ and $T_X\defas\gamma \cap T$ (\autoref{fig:twistproof}). This operation shown on the right in \autoref{fig:product} does not change the cycle matroid \cite{Whitney:2isomorphic}. Hence the Hepp bounds of $G[X]$ and $G'[X]$ agree, similarly for $Y$. So again, $F_G$ and $F_{G'}$ have the same residue.

	Finally consider the case where $X\cap Q = \set{p,r}$ and $Y\cap Q=\set{q,s}$, illustrated on the right in \autoref{fig:twistproof}. Set $S_Y\defas S\cap G[Y]$ and $T_Y\defas T\cap G[Y]$ and write $S_Y^e$ for the graph $S_Y$ with one extra edge $e$ between $q$ and $s$; similarly $S_X^e$ is $S_X$ with an extra edge $e$ connecting $p$ and $r$. Define $T_X^f$ and $T_Y^f$ analogously. Then \eqref{eq:res-twist-proof-XY} becomes the product
	\begin{equation*}
		\Res_{\sdc{G[X]}=0} F_G(\vec{\ind})
		= \tfrac{\Dim}{2}
		\HeppComp{S_X^e,\vec{\ind}}
		\HeppComp{T_X^f,\vec{\ind}}
		\HeppComp{S_Y^e,\vec{\ind}}
		\HeppComp{T_Y^f,\vec{\ind}}
	\end{equation*}
	by applying \eqref{eq:Hepp-product} to $G[X] = S_X^e \TwoSum{e}{f} T_X^f$ and $G[Y] = S_Y^e \TwoSum{e}{f} T_Y^f$. We get the same product as the residue of the twist $F_{G'}$ at $\sdc{G'[X']}=0$ if we set
	\begin{equation*}
		X' \defas (X\cap \VG{S}) \cup (Y\cap \VG{T}{\setminus}Q)
		\quad\text{and}\quad
		Y' \defas (Y\cap \VG{S}) \cup (X\cap \VG{T}{\setminus}Q),
	\end{equation*}
	because then $G'[X'] = S_X^e \TwoSum{e}{f} T_Y^f$ and $G'[Y'] = S_Y^e \TwoSum{e}{f} T_X^f$ (see \autoref{fig:twistproof}). This shows that $F_G$ and $F_{G'}$ have the same residue at $\sdc{G[X]}=0$, because $\sdc{G[X]}=\sdc{G'[X']}$. Note that the cut edges $C=C_S\sqcup C_T$ connecting $G[X]$ with $G[Y]$ decompose into the cut $C_S=C\cap S$ joining $S_X$ with $S_Y$ and the cut $C_T=C\cap T$ between $T_X$ and $T_Y$. The same cut edges separate $G'[X']$ from $G'[Y']$, so by \eqref{eq:sdc-exc} we get
	\begin{equation*}
		\sdc{G'[X']}= \frac{1}{2}\left( \sum_{e\in C} \ind_e^{\dual} - \Dim \right)
		=\sdc{G[X]}.
	\end{equation*}

	In conclusion, we have shown that $F_G-F_{G'}$ has no poles, so it must be a constant rational number just as in the proof of \autoref{thm:hepp-completion}. Now pick the particular values
	\begin{equation*}
%		\ind_e=\tfrac{\Dim}{2}\tfrac{n-3}{n-1}
		\ind_e^{\dual}=\tfrac{\Dim}{n-1}
		,\qquad
%		\ind_f=\tfrac{\Dim}{2}\tfrac{m-3}{m-1}
		\ind_f^{\dual}=\tfrac{\Dim}{m-1}
		\quad\text{and}\quad
%		\ind_h = \tfrac{\Dim}{6}\big( 2 + \tfrac{n-4}{n-1} + \tfrac{m-4}{m-1}\big)
		\ind_h^{\dual} = \tfrac{\Dim}{6}\big( 1 - \tfrac{n-4}{n-1} - \tfrac{m-4}{m-1}\big)
	\end{equation*}
	uniformly for all edges $e\in S=K_n$ and $f\in T=K_m$ with at most one end in $Q$, and the 12 edges $h$ with both ends in $Q=\set{p,q,r,s}$. These indices are conformal and convergent: If $\abs{X\cap Q}=2$ and $\abs{X\cap \VG{S}}=2+k$ and $\abs{X\cap \VG{T}}=2+l$, we find that
	\begin{equation*}
		\sdc{G[X]} = \tfrac{\Dim}{2}\Big( 
			\tfrac{k(n-4-k)}{n-1}
			+\tfrac{l(m-4-l)}{m-1}
			+\tfrac{1}{3}
			+\tfrac{2}{3}\tfrac{n-4}{n-1}
			+\tfrac{2}{3}\tfrac{m-4}{m-1}
		\Big)
		> 0
	\end{equation*}
	because $k\leq \abs{\VG{S}{\setminus}Q}=n-4$ and $l\leq m-4$. In the case with $\abs{X\cap Q} \leq 1$ and say $X\subseteq \VG{S}$, we get $\sdc{G[X]}=\frac{\Dim}{2} \frac{k-1}{n-1} (n-1-k)>0$ where $k\defas \abs{X}$ and thus $2\leq k \leq n-3$.

	We may therefore evaluate $F_G$ and $F_{G'}$ at the indices $\vec{\ind}$ defined above. By construction, the indices on $G'\cong G=G_{n,m}$ are the same and therefore $F_G(\vec{\ind})=F_{G'}(\vec{\ind})$ is trivial.
\end{proof}
\begin{figure}\centering
	$\Graph[0.6]{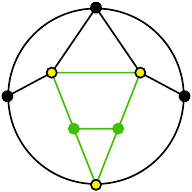}
	\quad\xleftarrow{\text{delete $p$}}\quad
	\Graph[0.6]{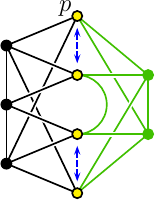}
%	\quad\longleftrightarrow\quad
	\quad\xleftrightarrow{\text{twist}\ }\quad
	\Graph[0.6]{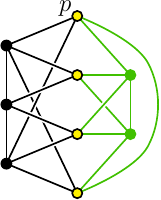}
	\quad\xrightarrow{\text{delete $p$}}\quad
	\Graph[0.6]{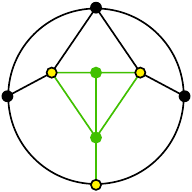}$
	\caption{The twist of $P_{7,4}$ (left) is $P_{7,7}$ (right) for the graphs listed in \cite{Schnetz:Census}.}%
	\label{fig:twist-P74-P77}%
\end{figure}
\begin{example}
	The vertex complements $G{\setminus}p$ and $G'{\setminus}p$ of a twist pair differ by replacing one side of a $3$-separation along $\set{q,r,s}$, see \autoref{fig:twist-P74-P77} for an explicit example. The twists where $T$ has this particular shape are studied as \emph{magic identities} in \cite{DrummondHennSmirnovSokatchev:MagicConformal4}.
\end{example}
%\begin{equation*}
%	G=P_{7,4}{\setminus}9 = \Graph[0.4]{P74}
%	\quad\rightarrow\quad
%	\Graph[0.4]{P74dual}
%	\quad\rightarrow\quad
%	G^{\dual}=P_{7,7}{\setminus}9=\Graph[0.4]{P77fromdual}
%\end{equation*}

\subsection{Fourier split and uniqueness}
\label{sec:Fourier-split}

The recently described Fourier split \cite{HuSchnetzShawYeats:Further} is a vast generalization of the old `uniqueness' relations \cite{Kazakov:Uniqueness,Kazakov:MethodOfUniqueness,Kazakov:TwoLectures}. It takes the planar dual on one side of a $3$-separation.
\begin{figure}
	\centering
	$G^{\TriG} = \Graph{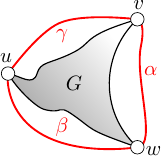}$
	\quad
	$G^{\StarG} = \Graph{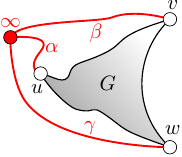}$
	$(G^{\StarG})^{\dual} = \Graph{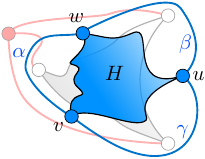} = H^{\TriG}$
	\caption{The triangle $G^{\TriG}$, star $G^{\StarG}$ and dual $H$ of an externally planar graph $G$.}%
	\label{fig:star-triangle-completion}%
\end{figure}
\begin{definition}%[\autoref{fig:star-triangle-completion}]
	Given a graph $G$ with three marked vertices $\set{u,v,w}$, the \emph{star over $G$} is the graph $G^{\StarG}$ obtained by adding a vertex `$\infty$' with edges $\alpha=u\infty$, $\beta=v\infty$ and $\gamma=w\infty$ as in \autoref{fig:star-triangle-completion}.
%	Given a graph $G$ with three marked vertices $\set{u,v,w}$, write $G^{\StarG}$ for the graph obtained by adding a vertex `$\infty$' with edges $\alpha=u\infty$, $\beta=v\infty$ and $\gamma=w\infty$ as in \autoref{fig:star-triangle-completion}.
%	Given a graph $G$ with three marked vertices $\set{u,v,w}$, the \emph{star over $G$} is the graph $G^{\StarG}$ obtained by connecting the marked vertices to an additional vertex `$\infty$' with edges $\alpha=u\infty$, $\beta=v\infty$ and $\gamma=w\infty$. % as in \autoref{fig:star-triangle-completion}.
	The \emph{triangle over $G$} is the graph $G^{\TriG}$ created by adding three edges $\alpha=vw$, $\beta=wu$ and $\gamma=uv$ to $G$.
%	Let $G^{\TriG}$ denote the graph created by adding three edges $\alpha=vw$, $\beta=wu$ and $\gamma=uv$ to $G$.
	We call $G$ \emph{externally planar} if $G^{\StarG}$ is planar, and then a \emph{dual} of $G$ is an externally planar graph $H$ such that $H^{\TriG} = (G^{\StarG})^{\dual}$.
\end{definition}
\begin{example}
	The star $G=\Graph[0.3]{starup}$ with terminals $\set{u,v,w}$ is externally planar with $G^{\StarG}=\Graph[0.23]{d222}$ and dual $(G^{\StarG})^{\dual} = \Graph[0.3]{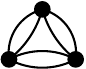}=H^{\TriG}$ for $H=\Graph[0.25]{tri}$. So, as externally planar graphs, the star $\Graph[0.3]{starup}$ and the triangle $\Graph[0.25]{tri}$ are dual to each other.
\end{example}
A marked graph $G$ is externally planar if and only if it admits a planar embedding such that $u$, $v$ and $w$ lie on the exterior face. The construction of a dual $H$ is illustrated in \autoref{fig:star-triangle-completion}. Note that the marked vertices of $H$ are in a well-defined sense `opposite' to those of $G$. For example, the vertex $u$ in $H$ corresponds to the face $vw\infty$ in $G^{\StarG}$.
\begin{definition}
	Suppose that the graph $G$ has an edge bipartition $\EG{G}=S\sqcup T$ that meets in precisely three vertices $\set{u,v,w}$, and that $S$ is externally planar with dual $S'$. The corresponding \emph{Fourier split} is the graph obtained by gluing $S'$ and $T$ along $\set{u,v,w}$.
\end{definition}
The Fourier split is shown schematically in \autoref{fig:fourier-split}; see \cite[Figure~3]{HuSchnetzShawYeats:Further} for an explicit example. We assume that $S$ and $T$ are connected, such that $\loops{G} = \loops{S}+\loops{T}+2$ and
\begin{equation*}
	0=\sdc{G} = \sdc{S} + \sdc{T} - \Dim.
\end{equation*}
When both $S$ and $T$ have the same degree of convergence $\sdc{S}=\sdc{T}=\Dim/2$, the Fourier split identity \cite[Theorem~2.8]{HuSchnetzShawYeats:Further} relates the periods of $G$ and its Fourier split.
\begin{figure}
	\centering
	$G=\Graph{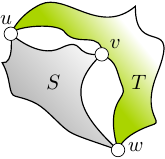}
	\quad\longrightarrow\quad
	\Graph{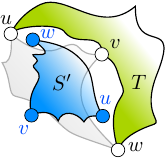}
	\quad\longrightarrow\quad
	G'=\Graph{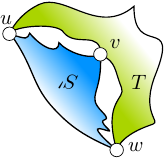}$
	\caption{Fourier split $G'$ of $G$ along a $3$-separation $S\sqcup T$.}%
	\label{fig:fourier-split}%
\end{figure}
\begin{theorem}
	\label{thm:fourier-split-period}%
	Let $G'$ denote a Fourier split of $G$ for a bipartition $\EG{G}=S\sqcup T$. Write $\vec{\ind}$ and $\vec{b}$ for the indices on $T$ and $S$, respectively. Then $\tPeriod{G',\vec{\ind},\vec{b}} = \tPeriod{G,\vec{\ind}^{\dual},\vec{b}}$ holds on the space of indices such that $\sdc{S}=\sdc{T}$.
\end{theorem}
\begin{example}
	A $3$-valent vertex $p$ is called \emph{unique} \cite{Kazakov:Uniqueness} if the dual indices at $p$ sum to $\Dim$ (we call this \emph{conformal}, $\exc{p}=0$, in \autoref{def:excess}). The three edges $S$ between $p$ and its neighbours $\set{u,v,w}$ form one side of a $3$-separation with $\sdc{S}=\Dim/2$, so the period is unchanged if we replace the star $S$ by a triangle with the dual indices.
	This special case of a Fourier split amounts to a star-triangle ($\Delta-Y$) transformation $G\rightarrow G'$ and the corresponding identities are known as \emph{uniqueness relations} \cite{Kazakov:Uniqueness,Kazakov:MethodOfUniqueness,Kazakov:TwoLectures}.
\end{example}
The proof in \cite[Theorem~2.8]{HuSchnetzShawYeats:Further} uses duality for graphical functions \cite{GolzPanzerSchnetz:GfParam}. Our derivation below stays entirely in Schwinger parameters, and it adapts easily to the Hepp bound.
\begin{proof}[Proof of \autoref{thm:fourier-split-period}]
	If we expand the graph polynomial of $S^{\StarG}$ as a function of the parameters of the additional three edges $\set{\alpha,\beta,\gamma}$, we get \cite[Example~32]{Brown:PeriodsFeynmanIntegrals}
\begin{equation}
	\PsiPol_{S^{\StarG}}
	=(\alpha\beta+\alpha\gamma+\beta\gamma)\PsiPol_S
	+(\beta+\gamma)\Phi_S^u 
	+(\alpha+\gamma)\Phi_S^v 
	+(\alpha +\beta)\Phi_S^w 
	+\Phi_S
	\label{eq:psipol-star}%
\end{equation}
in terms of \emph{spanning forest polynomials} \cite{BrownYeats:SpanningForestPolynomials}. For example, $\Phi_S = \sum_F \prod_{e\in S{\setminus}F} \SP_e$ is a sum over all spanning forests $F$ of $S$ with precisely three connected components, each containing one of the marked vertices. We only use that $\Phi_S$ is homogeneous of degree $\sdc{S}+2$ and that $\Phi^u_S,\Phi^v_S,\Phi^w_S$ are homogeneous of degree $\sdc{S}+1$. Similarly, we have
\begin{equation*}
	\PsiPol_{S^{\TriG}}
	=\alpha\beta\gamma\PsiPol_S
	+\alpha(\beta+\gamma)\Phi_S^u 
	+\beta (\alpha+\gamma)\Phi_S^v 
	+\gamma(\alpha +\beta)\Phi_S^w 
	+(\alpha+\beta+\gamma)\Phi_S 
\end{equation*}
with the same spanning forest polynomials \cite[Example~33]{Brown:PeriodsFeynmanIntegrals}. Let $\lambda \defas \alpha\beta+\alpha\gamma+\beta\gamma$ and set $z_e \defas \lambda/\SP_e$ for $e\in S$. Consider the dual $H$ of $S$ such that $(H^{\StarG})^{\dual} = S^{\TriG}$, then
\begin{equation}
	\PsiPol_{H^{\StarG}}(\vec{z},\alpha,\beta,\gamma)
	= 
	\lambda^{-\loops{S}-1}
	\Big(\prod_{e\in S} z_e\Big)
	\PsiPol_{S^{\StarG}}(\vec{\SP},\alpha,\beta,\gamma)
	\tag{$\ast$}%
	\label{eq:Hstar-Gstar}%
\end{equation}
follows from duality \eqref{eq:psitrop-dual} and comparison with \eqref{eq:psipol-star}. Now specialize to $\alpha \defas \Phi^u_T/\PsiPol_T$, $\beta \defas \Phi^v_T/\PsiPol_T$ and $\gamma \defas \Phi^w_T/\PsiPol_T$. It is shown in \cite[Proposition~22]{BrownYeats:SpanningForestPolynomials} that $\alpha\beta+\alpha\gamma+\beta\gamma=\lambda$ is then equal to $\lambda=\Phi_T/\PsiPol_T$. Therefore, \eqref{eq:psipol-star} specializes to
%With the identity $\PsiPol_T\Phi_T = \Phi_T^u\Phi_T^v + \Phi_T^u\Phi_T^w + \Phi_T^v \Phi_T^w$ from \cite[Proposition~22]{BrownYeats:SpanningForestPolynomials} we can write
\begin{equation*}
	\PsiPol_T \PsiPol_{S^{\StarG}}
%	\left( 
%		\tfrac{\Phi_T^u}{\PsiPol_T},
%		\tfrac{\Phi_T^v}{\PsiPol_T},
%		\tfrac{\Phi_T^w}{\PsiPol_T}
%	\right)
	= \PsiPol_S \Phi_T
	+\Phi_S^u \Phi_T^v
	+\Phi_S^u \Phi_T^w
	+\Phi_S^v \Phi_T^u
	+\Phi_S^v \Phi_T^w
	+\Phi_S^w \Phi_T^u
	+\Phi_S^w \Phi_T^v
	+\Phi_S\PsiPol_T.
\end{equation*}
This expression is equal to $\PsiPol_G$, see \cite[Theorem~23]{BrownYeats:SpanningForestPolynomials}. We can thus write \eqref{eq:Hstar-Gstar} as
%After multiplication with $\Phi_T^{\loops{S}+1}$,  thus takes the form
\begin{equation}
	\Phi_T^{\loops{S}+1} \PsiPol_{G'}(\vec{z},\vec{y})
	= \PsiPol_T^{\loops{S}+1} \Big( \prod_{e\in S} z_e \Big) \PsiPol_G(\vec{\SP},\vec{y})
	\label{eq:psipol-fourier-twist}%
\end{equation}
where $\vec{y}=(y_e)_{e\in T}$ denotes the Schwinger parameters of the edges in $T$, $\vec{\SP}=(\SP_e)_{e\in S}$ are the parameters of $S$ and $z_e=\lambda/\SP_e$. The Fourier split identity follows directly from \eqref{eq:psipol-fourier-twist} by the change of variables $\vec{\SP} \rightarrow \vec{z}$ in the Mellin integral \eqref{eq:period-mellin}, because
\begin{equation*}
	\frac{1}{[\PsiPol_G(\vec{\SP},\vec{y})]^{\Dim/2}}
	\prod_{e\in S} \SP_e^{\ind_e}
%	=\frac{1}{[\lambda^{\loops{S}+1}\PsiPol_{G'}(\vec{z},\vec{y})]^{\Dim/2}}
%	\prod_{e\in S} z_e^{\Dim/2}\Big(\tfrac{\lambda}{z_e}\Big)^{\ind_e}
	= \frac{\lambda^{\sdc{S}-\Dim/2}}{[\PsiPol_{G'}(\vec{z},\vec{y})]^{\Dim/2}}
	\prod_{e\in S} z_e^{\Dim/2-\ind_e}.
	\qedhere
\end{equation*}
\end{proof}
\begin{theorem}
	\label{thm:fourier-split-hepp}%
	Let $G'$ denote a Fourier split of $G$ for a bipartition $\EG{G}=S\sqcup T$. Write $\vec{\ind}$ and $\vec{b}$ for the indices on $T$ and $S$, respectively. Then $\Hepp{}(G',\vec{\ind},\vec{b}) = \Hepp{}(G,\vec{\ind}^{\dual},\vec{b})$ holds on the space of indices such that $\sdc{S}=\sdc{T}$.
\end{theorem}
\begin{proof}
	Multiply \eqref{eq:psipol-fourier-twist} by $\PsiPol_T^{\loops{T}+1}$ to clear denominators, then we obtain the identity
	\begin{equation*}
		\PsiPol_T(\vec{y})^{\loops{T}+1}
		\Phi_T(\vec{y})^{\loops{S}+1}
		\PsiPol_{G'}(\vec{z},\vec{y})
		= \Big( \prod_{e \in S} z_e \Big)
		\PsiPol_G\left( 
			\tfrac{\Phi_T(\vec{y})}{z_1},
			\ldots,
			\tfrac{\Phi_T(\vec{y})}{z_{\abs{S}}}
			,
			\PsiPol_T(\vec{y}) \cdot \vec{y}
		\right)
	\end{equation*}
	of polynomials in the variables $\vec{y}$ and $\vec{z}$. Each monomial in $\vec{y}$ and $\vec{z}$ appears with the same coefficient on the left and on the right, and spanning forest polynomials have only positive coefficients. So the maximum monomial that appears on either side is equal to
	\begin{equation*}
		\PsiTrop_T(\vec{y})^{\loops{T}+1}
		\PhiTrop_T(\vec{y})^{\loops{S}+1}
		\PsiTrop_{G'}(\vec{z},\vec{y})
		= \Big( \prod_{e \in S} z_e \Big)
		\PsiTrop_G\left( 
			\tfrac{\PhiTrop_T(\vec{y})}{z_1},
			\ldots,
			\tfrac{\PhiTrop_T(\vec{y})}{z_{\abs{S}}},
			\PsiTrop_T(\vec{y}) \cdot \vec{y}
		\right),
	\end{equation*}
	where $\PhiTrop_T$ %\defas \max_F\prod_{e\inT{\setminus}F} y_e$ 
	denotes the maximum monomial of $\Phi_T$. Set $\lambda^{\trop} \defas \PhiTrop_T/\PsiTrop_T$ and change the definition of $\vec{z}$ to $z_e \defas \lambda^{\trop}/\SP_e$, then we can write this tropical version of \eqref{eq:psipol-fourier-twist} as
	\begin{equation*}
		(\lambda^{\trop})^{\loops{S}+1} \PsiTrop_{G'}(\vec{z},\vec{y})
		= \Big(\prod_{e\in S} z_e \Big) \PsiTrop_G(\vec{\SP},\vec{y}).
	\end{equation*}
	Changing variables $\vec{\SP}\rightarrow \vec{z}$ in the Mellin integral \eqref{eq:hepp-mellin} as before proves the theorem.
\end{proof}
\begin{example}
	Consider the complete graph $K_4$ with labels as in \autoref{fig:completion-K5}. Its Hepp bound \eqref{eq:Hepp-K4-full} simplifies on the hyperplanes where a vertex or a triangle becomes unique. For example, on the subspace $H=\setexp{\vec{\ind}}{\ind_2+\ind_4+\ind_6=\Dim/2}$ we can replace the triangle $\set{1,3,5}$ by a unique star. The series \eqref{eq:series-hepp} reduces $G'=\Graph[0.2]{d222}$ to the dipole $\Bond{2}$ from \eqref{eq:hepp-bond}:
	\begin{align*}
		\Hepp{K_4,\vec{\ind}}|_{H}
		&= \Hepp{\Graph[0.2]{d222},\ind_1^{\dual},\ind_2,\ind_3^{\dual},\ind_4,\ind_5^{\dual},\ind_6}|_H
		\\
		&=
		\frac{
			\Dim/2
			(\ind_1^{\dual}+\ind_4)(\ind_3^{\dual}+\ind_6)(\ind_5^{\dual}+\ind_2)
		}{
			\ind_1^{\dual} \ind_2 \ind_3^{\dual} \ind_4 \ind_5^{\dual} \ind_6
			(\ind_1-\ind_4)(\ind_3-\ind_6)(\ind_5-\ind_2)
		}.
	\end{align*}
\end{example}

\section{\texorpdfstring{$\field^4$}{phi4} theory}
\label{sec:phi4}

The inequality $\Period{G} \leq \Hepp{G}$ from \eqref{eq:hepp-period-bound} was the initial motivation to study the Hepp bound. In this section we investigate this relation and discuss improved bounds.

For simplicity we will only consider {\PlogDiv} graphs in $\Dim=4$ dimensions with unit indices $\ind_e=1$ on each edge, that is, graphs $G$ such that
\begin{itemize}
	\item $\abs{E_G} = 2\loops{G}$ and
	\item $\abs{\gamma}>2\loops{\gamma}$ for every subgraph $\emptyset \neq \gamma \subsetneq G$.
\end{itemize}
These graphs are particularly interesting, since their periods contribute to the beta function of scalar field theories in $4$ dimensions of space-time. The best understood case is scalar $\field^4$ theory \cite{KleinertSchulteFrohlinde:CriticalPhi4,KompanietsPanzer:phi4eps6,Schnetz:NumbersAndFunctions}, which consists of graphs with degree at most four at each vertex. More than a thousand periods of these \emph{$\field^4$-graphs} are known \cite{PanzerSchnetz:Phi4Coaction,Schnetz:NumbersAndFunctions}.

\subsection{Period correlation}
Every {\PlogDiv} $\field^4$-graph $G$ is a vertex complement $G=H{\setminus}v$ in a $4$-regular graph $H=\widehat{G}$, called the \emph{completion} of $G$ (\autoref{ex:completion}). Graphs with the same completion have the same period and Hepp bound, see \autoref{sec:completion}. We call a $4$-regular graph $H$ \emph{primitive} if its vertex complements are {\PlogDiv}, so that $\Period{H{\setminus}v}$ and $\Hepp{H{\setminus}v}$ are finite (and independent of $v$). This condition is equivalent to $H$ being cyclically 6-edge-connected, that is, the only $4$-edge cuts of $H$ are those separating off a single vertex \cite[Proposition~2.1]{Schnetz:Census}.

The primitive $4$-regular graphs $H=P_{\ell,k}$ are enumerated in \cite{Schnetz:Census}, where $\ell=\loops{G}=\loops{H{\setminus}v}=\loops{H}-3$ refers to the loop number of the uncompleted graphs. For example, $P_{3,1}=K_5$ represents the graph $K_4=\Graph[0.15]{w3small}$ with $3$ loops from \eqref{eq:ws3-period} and \eqref{eq:ws3-hepp}.
With the recursion \eqref{eq:hepp-flag-recursive} we computed the Hepp bounds of all primitives with $\ell \leq 11$ loops.\footnote{These results are provided in the file \texttt{HeppBoundsPhi4.txt} with DOI \href{https://doi.org/10.5287/bodleian:ZVDqZnz78}{10.5287/bodleian:ZVDqZnz78}. It can be obtained from \url{https://doi.org/10.5287/bodleian:ZVDqZnz78} or from the \href{https://arxiv.org/abs/1908.09820}{arXiv}.} Periods are only known completely through $7$ loops \cite{PanzerSchnetz:Phi4Coaction}. \autoref{tab:hepp-7loop} summarizes the comparison.
\begin{table}
	\centering
	\begin{tabular}[t]{clrr}\toprule
		$\loops{G}$ & $\widehat{G}$ & $\Period{G}$ & $\Hepp{G}$ \\
		\midrule
		$1$ & $P_{1,1}$ & $1\phantom{.0}$ & $2$ \\
		$3$ & $P_{3,1}$ & $7.2$ & $84$ \\
		$4$ & $P_{4,1}$ & $20.7$ & $572$ \\
		\midrule
		$5$ & $P_{5,1}$ & $55.6$ & $3702$ \\
		    & $P_{3,1} \cdot P_{3,1}$ & $52.0$ & $3528$ \\
		\midrule
		$6$ & $P_{6,1}$ & $168.3$ & $26220$ \\
		    & $P_{3,1}\cdot P_{4,1}$ & $149.6$ & $24024$ \\
		    & $P_{6,2}$ & $132.2$ & $21912$ \\
		    & $P_{6,3}$ & $107.7$ & $18828$ \\
		    & $P_{6,4}$ & $71.5$ & $13968$ \\
		\bottomrule
	\end{tabular}\qquad
	\begin{tabular}[t]{clrr}\toprule
		$\loops{G}$ & $\widehat{G}$ & $\Period{G}$ & $\Hepp{G}$ \\
		\midrule
		$7$ & $P_{7,1}$ & $527.7$ & $190952$ \\
		    & $P_{4,1}\cdot P_{4,1}$ & $430.1$ & $163592$ \\
		    & $P_{3,1}\cdot P_{5,1}$ & $400.9$ & $155484$ \\
		    & $P_{7,2}$ & $380.9$ & $149426$ \\
		    & $P_{3,1}\cdot P_{3,1} \cdot P_{3,1}$ & $375.2$ & $148176$ \\
		    & $P_{7,3}$ & $336.1$ & $136114$ \\
		    & $\set{P_{7,4},P_{7,7}}$ & $294.0$ & $123260$ \\
		    & $P_{7,6}$ & $273.5$ & $116860$ \\
		    & $\set{P_{7,5},P_{7,10}}$ & $254.8$ & $110864$ \\
		    & $P_{7,9}$ & $216.9$ & $98568$ \\
		    & $P_{7,11}$ & $200.4$ & $92984$ \\
		    & $P_{7,8}$ & $183.0$ & $87088$ \\
		\bottomrule
	\end{tabular}%
	\caption{The periods and Hepp bounds for all primitive completed $\field^4$-graphs with up to $7$ loops, as illustrated in \autoref{fig:hepp-period}. The graphs $P_{7,4}$ and $P_{7,7}$ are Fourier dual and thus share the same period and Hepp bound; similarly for $P_{7,5}$ and $P_{7,10}$.}%
	\label{tab:hepp-7loop}%
\end{table}

For completed graphs $H$, the product \eqref{eq:Hepp-product} from a $2$-separation of $H{\setminus}v$ corresponds to a $3$-separation of $H$ \cite[Section~2.5]{Schnetz:Census}. So the entry $P_{3,1}\cdot P_{3,1}$ in \autoref{tab:hepp-7loop} corresponds to \autoref{ex:K4-product}. The only identities between different completed graphs up to $7$ loops are two dualities, which connect the planar uncompletions of $P_{7,4}$ and $P_{7,7}$, as well as the planar uncompletions of $P_{7,5}$ and $P_{7,10}$. The latter duality is shown in \autoref{fig:duality}, where $G=P_{7,5}{\setminus}9$ and $G^{\dual}=P_{7,10}{\setminus}v$ for any $v$ ($P_{7,10}$ is vertex transitive).

\autoref{tab:hepp-7loop} shows that the Hepp bound is much larger than the period, in fact by several orders of magnitude at higher loop orders (see \autoref{sec:improved-bounds} for improvements). Apart from an overall scale, however, the \emph{relative} variations of the Hepp bound follow the period surprisingly closely. This is illustrated at $7$ loops in \autoref{fig:hepp-period}.
To compare different loop orders, we use logarithmic coordinates. Recall that $\HeppComp{G}=\Hepp{G}/2$, so we set
\begin{equation*}
	\xi(G) \defas \left(
		\frac{\ln (\Hepp{G}/2)}{\loops{G}-1},
		\frac{\ln \Period{G}}{\loops{G}-1}
	\right)
\end{equation*}
because $\loops{G}-1$ is additive under $2$-sums \eqref{eq:loops-2sum}. This choice of variables linearizes the product \eqref{eq:Hepp-product}: If $G=G_1 \TwoSum{e}{f} G_2$ is a $2$-sum, then the point
$\xi(G)
	= \lambda \xi(G_1)
	+ (1-\lambda)\xi(G_2)
$
lies on the straight line between $\xi(G_1)$ and $\xi(G_2)$ at $\lambda=(\loops{G_1}-1)/(\loops{G}-1)$. 

\autoref{fig:correlation} of all known periods up to $11$ loops contains more than a thousand points and demonstrates the persistence of the correlation and its uniformity across loop orders.
\begin{figure}
	\centering%
	\includegraphics[width=0.87\textwidth]{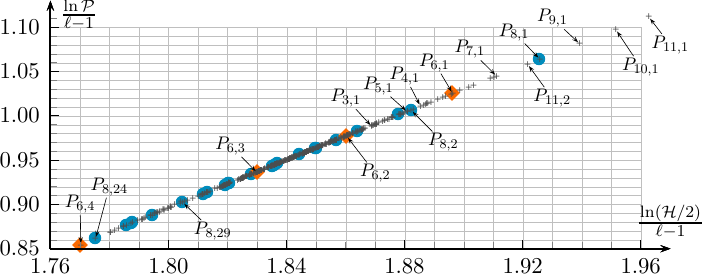}%
	\caption{$\field^4$ periods from \cite{PanzerSchnetz:Phi4Coaction} and their Hepp bounds (products not included). Graphs with $6$ and $8$ loops are highlighted with orange diamonds and blue circles.% Some graphs are labelled, in particular the zigzag graphs $P_{\ell,1}$, which give the largest periods at each loop order.
}%
	\label{fig:correlation}%
\end{figure}
\begin{remark}
	Despite the strong correlation evident in \autoref{fig:hepp-period} and \autoref{fig:correlation}, we cannot strengthen \autoref{con:faithful} to the equivalence $\Period{G_1}\leq \Period{G_2} \Leftrightarrow \Hepp{G_1} \leq \Hepp{G_2}$ of inequalities. In other words, the plots are not monotone---though one needs to zoom in closely to notice this. For example, in \cite{PanzerSchnetz:Phi4Coaction} we find that
	\begin{align*}
		5548.00\approx\Period{P_{10,48}{\setminus}v}
		&<
		\Period{P_{10,255}{\setminus}v} \approx 5549.93,
		\quad\text{but}\\
		32743060 = \Hepp{P_{10,48}{\setminus}v} 
		&> 
		\Hepp{P_{10,255}{\setminus}v} = 32740360
		.
	\end{align*}
\end{remark}

\subsection{Unexplained identities}

\begin{figure}
	\centering\setlength{\tabcolsep}{13pt}
	\begin{tabular}{cccc}
		$\Graph[0.42]{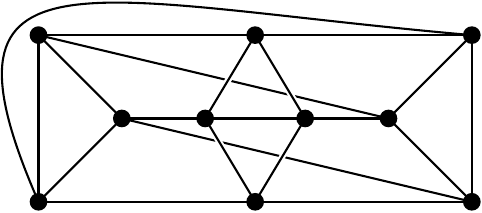}$ &
		$\Graph[0.45]{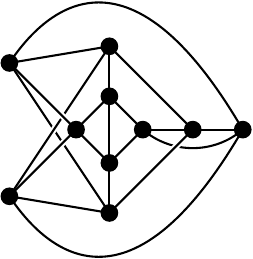}$ &
		$\Graph[0.5]{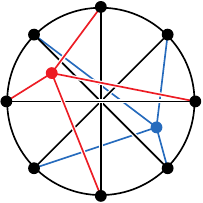} = \Graph[0.32]{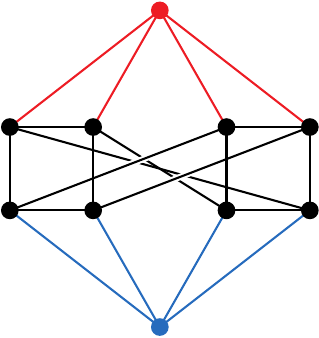}$ &
		$\Graph[0.5]{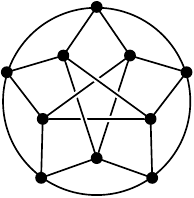}$ \\
		$P_{8,30}$ & $P_{8,31}$ & $P_{8,35}$ & $P_{8,36}$
	\end{tabular}
	\caption{The smallest $\field^4$ primitives with unexplained Hepp bound equalities.}%
	\label{fig:8-loop-hepp-coincidences}%
\end{figure}
The first examples of equal Hepp bounds in $\field^4$ theory that do not follow from any known symmetry are $\Hepp{P_{8,30}{\setminus}v} = \Hepp{P_{8,36}{\setminus}v}$ and $\Hepp{P_{8,31}{\setminus}v} = \Hepp{P_{8,35}{\setminus}v}$ as stated in \eqref{eq:hepp-8loop-pairs}. These four graphs are depicted in \autoref{fig:8-loop-hepp-coincidences} and structurally rather different:
\begin{center}
\begin{tabular}{rccccc}\toprule
	graph $\widehat{G}$ & $\Aut(\widehat{G})$ & triangles & ancestor & uncompletions & $c_2(G)$ \\
	\midrule
	$P_{8,30}$ & $\Z/2\Z$ & 4 & $P_{7,11}$ & 5 & $-z_3$ \\
	$P_{8,31}$ & $\Z/2\Z \times \Z/2\Z$ & 3 & $P_{7,8}$  & 7 & $-z_2$ \\
	$P_{8,35}$ & $D_8$    & 4 & $P_{8,35}$ & 2 & $-z_2$ \\
	$P_{8,36}$ & $D_5$    & 5 & $P_{8,36}$ & 2 & $-z_3$ \\
	\bottomrule
\end{tabular}
\end{center}
Both $P_{8,30}$ and $P_{8,31}$ have a double triangle and are thus \emph{descendants} of $P_{7,11}$ and $P_{7,8}$, respectively, in the terminology of \cite{Schnetz:Census}. Their Hepp bound partners $P_{8,35}$ and $P_{8,36}$ have no double triangles and larger, dihedral symmetry groups: $P_{8,35}$ arises by attaching two vertices to the circulant $C^8_{1,4}$ and $P_{8,36}$ features a cyclic arrangement of five triangles.

Apart from the equality of Hepp bounds, it is notable that also the $c_2$-invariants and the permanents coincide within both pairs of graphs. In particular the latter indicates a strong relationship of these pairs, because the permanent is a very rich invariant---except for relations following from the four symmetries of \cite{Schnetz:Census}, the permanents of two different graphs coincide only in very few cases \cite[Appendix~A]{Crump:ExtendedPermanent}.

We are thus led to \autoref{con:faithful}, motivated by the expectation that all these identities are not accidents but rather a consequence of some underlying combinatorial relationship of the graphs, like the symmetries in \autoref{sec:symmetries}. This structure remains to be identified, but we hope that a mechanism explaining equality of Hepp bounds, permanents and $c_2$ invariants might also force the periods to be equal.

All data from \cite{PanzerSchnetz:Phi4Coaction} is compatible with \autoref{con:faithful} and further corroborated by \cite{HuSchnetzShawYeats:Further}. We computed the Hepp bounds for all primitive $\field^4$ graphs with $\ell \leq 11$ loops, and whenever two graphs share the same Hepp bound and their periods are known, then the periods indeed coincide. The same holds for the known $c_2$ invariants \cite{PanzerSchnetz:Phi4Coaction,HuSchnetzShawYeats:Further} and the permanents computed in \cite{Crump:ExtendedPermanent}. Note that this is necessary for \autoref{con:faithful} to be viable, because it is conjectured that graphs with equal period have the same $c_2$ invariant \cite[Conjecture~5]{BrownSchnetz:K3phi4} and the same permanent \cite[Conjecture~1]{CrumpDeVosYeats:Permanent}.

\begin{table}
	\centering
	\begin{tabular}{rc@{$\;=\;$}c@{$\;+\;$}c@{$\;+\;$}c@{$\;+\;$}c}
		\toprule
		$\ell$ & primitives & \parbox[m]{14mm}{\centering\small Hepp bounds} & \parbox[m]{13mm}{\centering\small Fourier \& twist} & \parbox[m]{12mm}{\centering\small Fourier split} & \parbox[m]{20mm}{\centering\small unexplained identities} \\
		\midrule
			 7 &   11 &    9 &    2 &  0 &   0 \\
			 8 &   41 &   29 &   10 &  0 &   2 \\
			 9 &  190 &  129 &   55 &  1 &   5 \\ %6
			10 & 1182 &  776 &  346 & 13 &  47 \\ %60
			11 & 8687 & 6030 & 2411 & 55 & 191 \\ %246
		\bottomrule
	\end{tabular}%
	\caption{Tally of known and unexplained identities of Hepp bounds of $\field^4$ graphs.}%
	\label{tab:phi4-identities}%
\end{table}
Out of the four graphs in \autoref{fig:8-loop-hepp-coincidences}, only $P_{8,31}$ could be computed in \cite{PanzerSchnetz:Phi4Coaction}.
The number of unknown $8$-loop $\field^4$ periods would drop from 8 in \cite{PanzerSchnetz:Phi4Coaction} to 6 if the conjectures
\begin{equation*}
	\Period{P_{8,30}{\setminus}v} \stackrel{?}{=} \Period{P_{8,36}{\setminus}v}
	\quad\text{and}\quad
	\Period{P_{8,31}{\setminus}v} \stackrel{?}{=} \Period{P_{8,35}{\setminus}v}
\end{equation*}
were to be verified.
At higher loop orders, there are many more period relations that would follow through \autoref{con:faithful} from unexplained identities of Hepp bounds. The summary in \autoref{tab:phi4-identities} starts out with the number of irreducible primitives, i.e.\ completed primitive $\field^4$ graphs that do not have a three-vertex cut. 
In this setup of \cite{Schnetz:Census}, product relations are absent and the completion symmetry is automatically taken care of.

The number of different Hepp bounds (third column) is smaller, because different completion classes can evaluate to the same Hepp bound. The bulk of these identities is explained by the twist and Fourier relations from \cite{Schnetz:Census}, and a few more identities follow from the Fourier split discussed in \cite{HuSchnetzShawYeats:Further}. From 8 loops onwards, there are leftover identities of Hepp bounds that are not explained by any of these symmetries (last column), the first examples of which are \eqref{eq:hepp-8loop-pairs}.

\subsection{Improved bounds}
\label{sec:improved-bounds}
The wheel graphs $\WS{\ell}$ illustrate the huge gap between the Hepp bound and the period: According to \autoref{prop:hepp-wheels}, their Hepp bounds grow by a factor of $9$ per loop, whereas the periods only gain a factor of $4$ per loop \cite{Broadhurst:Wheels}:
\begin{equation*}
	\Period{\WS{\ell}} = \binom{2\ell-2}{\ell-1} \mzv{2\ell-3}
	\sim \frac{4^{\ell-1}}{\sqrt{\pi \ell}}.
\end{equation*}
To improve the bound, we can drop the numerator $\abs{\gamma_1}\cdots \abs{G/\gamma_{\ell-1}}$ in the formula \eqref{eq:hepp-1pi-flags} that sums over flags $\gamma_1 \subsetneq \cdots\subsetneq \gamma_{\ell}=G$ of bridgeless subgraphs (see \autoref{thm:HeppOne}). This follows from an approximation of the graph polynomial $\PsiPol_G$ by products of one-loop graphs, within the sector associated to the flag,
\begin{equation}
	\HeppSec[\OnePI]{\gamma_{\bullet}}
	\defas \bigg\{
		0
		<
		\max_{e\in \gamma_1} \SP_e
		<
		\max_{e\in\gamma_2/\gamma_1} \SP_e
		<
		\cdots
		<
		\max_{e \in G/\gamma_{\ell-1}} \SP_e
	\bigg\}
	\subset \R_+^{N}.
	\label{eq:1pi-sector}%
\end{equation}
\begin{theorem}
	The period $\Period{G} \leq \HeppOneLoop{G}$ is bounded by 
	\begin{equation}
		\HeppOneLoop{G}
		\defas \sum_{\gamma_{\bullet} \in \Flags[\OnePI]{G}}
		\frac{
			\cOne{\abs{\gamma_1}} \cOne{\abs{\gamma_2\setminus \gamma_1}} \cdots \cOne{\abs{G\setminus \gamma_{\ell-1}}}
		}{
			\sdc{\gamma_1} \cdots \sdc{\gamma_{\ell-1}}
		},
		\label{eq:HeppOneLoop}%
	\end{equation}
	where the positive reals $\cOne{n} \in \R$ are defined for all integers $n\geq 1$ by
	\begin{equation}
		\cOne{n} \defas
		\int_{[0,1]^{n}}
		\frac{ \td \SP_1 \cdots \td \SP_{n}}{(\SP_1+\cdots+\SP_{n})^2} \delta\left( 1-\max_{1\leq i \leq n} \SP_i \right)
		.
		\label{eq:cOne}%
	\end{equation}
\end{theorem}
\begin{example}
	With \autoref{tab:cOne}, the flags of $K_4$ from \autoref{ex:K4-1pi-flags} give the bound
	\begin{equation*}
		\HeppOneLoop{K_4} 
		= 12 \frac{\cOne{3} \cOne{2} \cOne{1}}{1\cdot 1} + 6 \frac{\cOne{4} \cOne{1} \cOne{1}}{2\cdot 1}
		= 72\ln 3-96 \ln 2
		\approx 12.56.
	\end{equation*}
\end{example}
\begin{proof}
	For any subgraph $\gamma \subseteq G$, the spanning trees $T\in\STrees{G}$ with the property that $F=\gamma \cap T$ is a spanning forest of $\gamma$ are easily seen to stand in bijection with pairs $(F,T')$ of spanning forests $F$ of $\gamma$ and spanning trees $T'$ of the quotient $G/\gamma$ \cite[Proposition~2.2]{Brown:FeynmanAmplitudesGalois}. Applied to a bridgeless flag $\gamma_{\bullet} \in \Flags[\OnePI]{G}$ with $\ell=\loops{G}$ elements, this gives an injection
	\begin{equation*}
		\STrees{\gamma_1} \times \STrees{\gamma_2/\gamma_1} \times \cdots \times \STrees{G/\gamma_{\ell-1}}
		\injects
		\STrees{G},\qquad
		(T_1,\ldots,T_{\ell}) \mapsto T_1 \cup \cdots \cup T_{\ell}
	\end{equation*}
	covering precisely those spanning trees $T\in\STrees{G}$ such that $T\cap \gamma_k$ is a spanning forest of $\gamma_k$, for each $1\leq k \leq \ell$.
	In terms of the graph polynomials, this gives the inequality
	\begin{equation}
		\PsiPol_G \geq \PsiPol_{\gamma_1} \PsiPol_{\gamma_2/\gamma_1} \cdots \PsiPol_{G/\gamma_{\ell-1}},
		\label{eq:Psi-flag-ineq}%
	\end{equation}
	which goes back to \cite[Proposition~3.5]{BlochEsnaultKreimer:MotivesGraphPolynomials} and is crucial for the desingularization of Feynman integrals \cite{Brown:FeynmanAmplitudesGalois,Schultka:ToricFeynman}.
	The quotients $\gamma_k/\gamma_{k-1}$ are cycles, such that $\PsiPol_{\gamma_k/\gamma_{k-1}} = \sum_{e\in \gamma_k\setminus\gamma_{k-1}} \SP_e$, and we can therefore estimate the period integral over the flag sector \eqref{eq:1pi-sector} by
	\begin{align*}
		\int_{\HeppSec[\OnePI]{\gamma_{\bullet}}} \frac{\Omega}{\PsiPol_G^2}
		&= \int_{\HeppSec[\OnePI]{\gamma_{\bullet}}} \frac{\Omega}{\PsiPol_G^2} \prod_{k=1}^{\ell} \int_0^{\infty} \td \lambda_k \ \delta\left( \lambda_k-\max_{e\in\gamma_k/\gamma_{k-1}} \SP_e \right)
		\\
		&\leq
		\int_{0<\lambda_1<\cdots<\lambda_{\ell}=1} \prod_{k=1}^{\ell-1} \lambda_k^{\abs{\gamma_k/\gamma_{k-1}}} \frac{\td \lambda_k}{\lambda_k^3}
		\prod_{k=1}^{\ell} \cOne{\abs{\gamma_k/\gamma_{k-1}}}
	\end{align*}
	by changing variables to $\SP_e = \lambda_k y_e$ for $e \in \gamma_k$. The iterated integral over the $\lambda_k$'s generates the product $\prod_{k=1}^{\ell-1}(\abs{\gamma_k}-2k) = \prod_{k=1}^{\ell-1} \sdc{\gamma_k}$ in the denominator of \eqref{eq:HeppOneLoop}.
\end{proof}
\begin{table}
	\centering
	\begin{tabular}{rccccc}\toprule
		$n$ & 1 & 2 & 3 & 4 & 5 \\\midrule
		$\cOne{n}$ & 1 & 1 & $6 \ln 2-3\ln 3$ & $36 \ln 3 - 56 \ln 2$ & $360 \ln 2 - 135\ln 3 - \frac{125}{2} \ln 5$ \\
		 &  &  & $\approx 0.863$ & $\approx 0.734$ & $\approx 0.630$ \\ \bottomrule
	\end{tabular}%
	\caption{The coefficients \eqref{eq:cOne} that appear in the numerator of the improved bound \eqref{eq:HeppOneLoop}.}%
	\label{tab:cOne}%
\end{table}
We can relabel the maximal Schwinger parameter in \eqref{eq:cOne} to be $\SP_1$, so that
	\begin{equation}
		\cOne{n}
		=
		n \int_{[0,1]^{n-1}}
		\frac{ \td \SP_2 \cdots \td \SP_{n}}{(1+\SP_2+\cdots+\SP_{n})^2}
		.
		\label{eq:cOne-1}%
	\end{equation}
	For small $n$, these integrals are listed in \autoref{tab:cOne}. For large $n$, they grow asymptotically like $\cOne{n} \sim 4/n+\asyO{n^{-2}}$, but we will not use this. From $n\geq 3$, they involve logarithms,
	\begin{equation}
		\cOne{n} = \frac{1}{(n-3)!} \sum_{k=2}^n \binom{n}{k} (-1)^{k+n+1} k^{n-2} \ln k,
		\label{eq:cOne-as-logs}%
	\end{equation}
	which follows in the limit $\rho\rightarrow -3$ from the elementary integral (for non-integral $\rho$)
	\begin{equation*}
		\prod_{e=2}^n \int_0^{1} \td \SP_e
		\left(1+\sum_{i=2}^n\SP_i\right)^{\rho+1}
		= \frac{1}{(\rho+2)\cdots(\rho+n)} \sum_{I \subseteq \set{2,\ldots,n}} (-1)^{n-\abs{I}} \left( 1+\abs{I} \right)^{\rho+n}
		.
	\end{equation*}
\begin{remark}
	The improved bound $\HeppOneLoop{G}$ does not respect any of the period symmetries exactly. For example, the two uncompletions in \autoref{fig:completion} differ slightly:
	\begin{equation*}
		\HeppOneLoop{\Graph[0.25]{5Rv}}
		= 6(\cOne{6}{+}3\cOne{5}{+}9\cOne{4}{+}6\cOne{3}{+}10\cOne{3}^2{+}7\cOne{3}\cOne{4})
		\approx 156.63
		\neq
		\HeppOneLoop{\Graph[0.25]{5Rw}}
%		4\cOne{6}+48\cOne{4}+88\cOne{3}^2+24\cOne{5}+24\cOne{3}+28 \cOne{3}\cOne{4}
%		4(\cOne{6}+12\cOne{4}+22\cOne{3}^2+6\cOne{5}+6\cOne{3}+7 \cOne{3}\cOne{4})
		\approx 156.54.
	\end{equation*}
\end{remark}
\begin{lemma}
	The sequence $c_n$ defined in \eqref{eq:cOne} starts with $c_1=c_2=1$ and is strictly decreasing thereafter. In particular, $c_n \leq 1$ for all $n$.
\end{lemma}
\begin{proof}
	Let $y_i \defas \SP_1+\cdots+\SP_{i-1}+\SP_{i+1}+\cdots+\SP_n$, then the convexity of $\SP\mapsto \SP^{-2}$ implies
	\begin{equation*}
		\left( \sum_{i=1}^n \SP_i \right)^{-2}
		= \left( \frac{1}{n-1} \sum_{i=1}^n y_i \right)^{-2}
		= \left( \frac{n-1}{n} \right)^2  \left( \frac{1}{n} \sum_{i=1}^n y_i \right)^{-2}
		\leq \frac{(n-1)^2}{n^3} \sum_{i=1}^n \frac{1}{y_i^2}.
	\end{equation*}
	Now set $\SP_1 \defas 1$ and integrate over $\SP_2,\ldots,\SP_{n}$, then using \eqref{eq:cOne-1} we find that
	\begin{equation*}
		\cOne{n}
		\leq
		\frac{(n-1)^2}{n^2}
		\sum_{i=1}^{n} \int_{[0,1]^{n-1}} \frac{\td\SP_2\cdots\td\SP_n}{y_i^2}
		= 
		\frac{(n-1)^2}{n^2} \left( \cOne{n-1} + \frac{\cOne{n-1}}{n-3} \right).
		\tag{$\ast$} \label{eq:cOne-ineq}%
	\end{equation*}
	Here we exploited that each of the $n-1$ summands with $i>1$ just integrates to $\frac{\cOne{n-1}}{n-1}$, whereas the first summand with $i=1$ contributes the term $\frac{\cOne{n-1}}{n-3}$ according to \eqref{eq:cOne} and
	\begin{equation*}
		\int_{[0,1]^{n-1}} \frac{\td \SP_2\cdots \td\SP_n}{(\SP_2+\cdots+\SP_{n})^2}
		= \int_0^1 \frac{t^{n-1}}{t^2}\frac{\td t}{t} \int_{[0,1]^{n-1}} \frac{\td y_2\cdots\td y_n}{(y_2+\cdots+y_{n})^2} \delta\left( 1-\max_{2\leq i \leq n} y_i \right),
	\end{equation*}
	upon changing variables to $t\defas \max_{2\leq i \leq n} \SP_i$ and $\SP_i=t y_i$. From \eqref{eq:cOne-ineq} then we see that
	\begin{equation*}
		\frac{\cOne{n}}{\cOne{n-1}} \leq 1- \frac{n^2-5n+2}{n^2(n-3)}<1
	\end{equation*}
	for all $n\geq 5$, and the remaining relations $\cOne{4}< \cOne{3} < \cOne{2}$ are demonstrated in \autoref{tab:cOne}.
\end{proof}
\begin{corollary}\label{thm:HeppOne}
	The period $\Period{G} \leq \HeppOne{G}$ is bounded by the rational number
	\begin{equation}
		\HeppOne{G}
		\defas \sum_{\gamma_{\bullet} \in \Flags[\OnePI]{G}}
		\frac{
			1
		}{
			\sdc{\gamma_1} \cdots \sdc{\gamma_{\loops{G}-1}}
		}
		\in \Q.
		\label{eq:HeppOne}%
	\end{equation}
\end{corollary}
\begin{table}
	\centering
	\begin{tabular}[t]{lrrr}\toprule
		$\widehat{G}$ & $\Period{G}$ & $\Hepp{G}$ & $\HeppOne{G}$ \\
		\midrule
		$P_{1,1}$ & $1\phantom{.0}$ & $2$ & $1$ \\
		$P_{3,1}$ & $7.2$ & $84$ & $15$\\
		$P_{4,1}$ & $20.7$ & $572$ & $59$\\
		\midrule
		$P_{5,1}$ & $55.6$ & $3702$ & $224$\\
		$P_{3,1} \cdot P_{3,1}$ & $52.0$ & $3528$ & $216$ \\
		\bottomrule
	\end{tabular}\qquad
	\begin{tabular}[t]{lrrr}\toprule
		$\widehat{G}$ & $\Period{G}$ & $\Hepp{G}$ & $\HeppOne{G}$ \\
		\midrule
		$P_{6,1}$ & $168.3$ & $26220$ & $909.5$\\
		$P_{3,1}\cdot P_{4,1}$ & $149.6$ & $24024$ & $852\phantom{.0}$ \\
		$P_{6,2}$ & $132.2$ & $21912$ & $795.5$\\
		$P_{6,3}$ & $107.7$ & $18828$ & $709.5$\\
		$P_{6,4}$ & $71.5$ & $13968$ & $567\phantom{.0}$\\
%		$P_{6,1}$ & $168.3$ & $26220$ & $1819/2$\\
%		$P_{6,2}$ & $132.2$ & $21912$ & $1591/2$\\
%		$P_{6,3}$ & $107.7$ & $18828$ & $1419/2$\\
%		$P_{6,4}$ & $71.5$ & $13968$ & $567$\\
		\bottomrule
	\end{tabular}%
	\caption{Comparison of the Hepp bound and its improvement \eqref{eq:HeppOne} for $\ell \leq 6$ loops.}%
	\label{tab:hepp-improved}%
\end{table}
This slightly bigger, but rational bound follows from \eqref{eq:HeppOneLoop} because the products of $\cOne{k}$ in the numerators are $\leq 1$.
For \autoref{ex:K4-1pi-flags}, $\HeppOne{K_4} = 12\cdot \frac{1}{1\cdot 1} + 6\cdot \frac{1}{2\cdot 1} = 15$ is much closer to $\Period{K_4} \cong 7.2$ than $\Hepp{K_4}=84$. More comparisons are given in \autoref{tab:hepp-improved}.

In contrast to $\HeppOneLoop{G}$, the rational bound $\HeppOne{G}$ \emph{does} respect completion symmetry $\HeppOne{G{\setminus}v}=\HeppOne{G{\setminus}w}$, and the same proof as for \autoref{thm:hepp-completion} applies. But the other symmetries fail: Duality is violated by the pair $(P_{7,5}{\setminus}9)^{\dual} \cong P_{7,10}{\setminus}v$ from \autoref{fig:duality}:
\begin{equation*}
	\HeppOne{P_{7,5}{\setminus}9}
	= \tfrac{5015}{2}
	\neq
	2517
	=\HeppOne{P_{7,10}{\setminus}v}.
\end{equation*}
For the twist from \autoref{fig:twist-P74-P77}, we find the values
\begin{equation*}
	\HeppOne{P_{7,4}{\setminus}p} = \tfrac{10855}{4}
	\neq \tfrac{5443}{2} = \HeppOne{P_{7,7}{\setminus}p},
\end{equation*}
and there is no product relation between $225=\HeppOne{K_4}^2$ and the $2$-sum
\begin{equation*}
	\HeppOne{\Graph[0.25]{5Rw}}
	=
	\HeppOne{\Graph[0.25]{w3small} \TwoSum{e}{f} \Graph[0.25]{w3small}}
	= 216
	.
\end{equation*}
\begin{remark}
	An analogous improvement applies to the formula \eqref{eq:hepp-flat-flags} expressing the Hepp bound as a sum over flags of flats (induced subgraphs). The integration sector is then
	\begin{equation*}
		\HeppSec[\Flat]{\gamma_{\bullet}}
		\defas \bigg\{
			0
			<
			\min_{e\in \gamma_1} \SP_e
			<
			\min_{e\in\gamma_2/\gamma_1} \SP_e
			<
			\cdots
			<
			\min_{e \in M/\gamma_{r-1}} \SP_e
		\bigg\}
		\subset \R_+^{N},
%		\label{eq:flat-sector}%
	\end{equation*}
	and we arrive at the improved rational bound $\HeppOneFlat{G} \geq \Period{G}$ given by
	\begin{equation*}
		\HeppOneFlat{G}
		\defas \sum_{\gamma_{\bullet} \in \Flags[\Flat]{G}}
		\frac{
			1
		}{
			\sdc{\gamma_1} \cdots \sdc{\gamma_{\rank{G}-1}}
		}
		\in \Q.
	\end{equation*}
	For planar graphs, \autoref{rem:dual-1pi-flat} shows $\HeppOneFlat{G} = \HeppOne{G^{\dual}}$, but otherwise the improved bounds are unrelated. The bound $\HeppOneFlat{G}$ violates all symmetries, even completion:
	\begin{equation*}
		\HeppOneFlat{\Graph[0.25]{5Rv}}
		= 218
		\neq
		216
		= \HeppOneFlat{\Graph[0.25]{5Rw}}
		= \HeppOne{\Graph[0.25]{5Rw}}
		= \HeppOne{\Graph[0.25]{5Rv}}.
	\end{equation*}
\end{remark}
\begin{proposition}
	The improved Hepp bounds of the wheel graphs with $n$ loops are
	\begin{equation*}
		\HeppOne{\WS{n}} = \HeppOneFlat{\WS{n}}
		= 
		-\frac{n(n-3)}{2(n-2)}
%		\frac{1}{n-2}-\frac{n-1}{2}
		+\frac{2}{4^n} \sum_{k=1}^n \binom{2n-2k}{n-k} \binom{2k}{k} k \cdot 5^{n-k}
		\sim \frac{5^{n+1/2}}{8\sqrt{\pi n}}.
	\end{equation*}
\end{proposition}
\begin{proof}
	Note that $\WS{n}^{\dual}=\WS{n}$ is self-dual. For the improved bound $\HeppOneFlat{\WS{n}}$, the numerator $\ind_{M/\gamma}^{\dual}$ disappears from the recursion \eqref{eq:hyperplane-block-recursion}. With this modification, the calculation as in the proof of \autoref{prop:hepp-wheels} produces the generating function
	\begin{equation*}
		\sum_{n=3}^{\infty} \HeppOne{\WS{n}} z^n
		= \frac{z^3(z-2)}{2(1-z)^2} - 4z^2\log(1-z) + \frac{z}{\sqrt{(1-5z)(1-z)^3}}. \qedhere
	\end{equation*}
\end{proof}
The numerator $\abs{\gamma_1}\cdots\abs{G/\gamma_{\ell-1}}$ in the flag formula \eqref{eq:hepp-1pi-flags} is at most $2^{\ell}$, because the factors add to $\abs{\EG{G}}=2\ell$. The improved bound is therefore at least $\HeppOne{G} \geq \Hepp{G}/2^\ell$, so the gain of $5/9$ per loop in the ratio $\HeppOne{\WS{n}}/\Hepp{\WS{n}}$ is almost as good as possible.

Up to the overall scale, the improved bounds show a very similar correlation as in \autoref{fig:correlation} (see \autoref{tab:hepp-improved}). The most notable difference is that a single period branches into several data points with slightly different abscissa, due to the violation of symmetries.

\section{Polyhedral geometry}
\label{sec:geometry}

In the spirit of tropical geometry, let us change variables from the Schwinger parameters $\SP_i=e^{-\SPlog_i}$ to their logarithms $\SPlog_i=-\log\SP_i$. Then $\prod_{i \notin T} \SP_i = \exp\left( -\vec{\SPlog}\cdot \uv{T^c} \right)$ in terms of the characteristic vector 
$\uv{T^c} = \sum_{i \notin T} \uv{i} \in \R^{N}$
of the complement $T^c$ of any spanning tree $T$.
In the affine chart $\SP_1=1$, the Hepp bound integrand from \eqref{eq:hepp-mellin} then becomes
\begin{equation*}
	\left.\frac{\Omega(\vec{\ind})}{(\PsiTrop_G)^{\Dim/2}}\right|_{\SP_1=1}
	= \exp\left[ 
		-\vec{\SPlog}\cdot\vec{\ind} 
		- \frac{\Dim}{2} \max_{T\in\STrees{G}} \left( -\vec{\SPlog} \cdot \uv{T^c} \right)
	\right]_{\SPlog_1=0}
	\td \SPlog_2 \cdots \td \SPlog_N.
	%\td \SPlog_2 \wedge \ldots \wedge \td \SPlog_N
\end{equation*}
\begin{definition}\label{def:sdc-vec}
	The exponent defines a continuous, piecewise linear function
	\begin{equation}
		\R^N \ni \vec{\SPlog} \mapsto 
		\sdc[\vec{\ind}]{\vec{\SPlog}}
		\defas
		\vec{\SPlog}\cdot\vec{\ind} 
		+\frac{\Dim}{2} \max_{T\in\STrees{G}} \left( -\vec{\SPlog} \cdot \uv{T^c} \right)
		\label{eq:sdc-vec}%
	\end{equation}
	where, as always, $\frac{\Dim}{2}=\frac{\ind_1+\cdots+\ind_N}{\loops{G}}$ is fixed by the constraint $\sdc[\vec{\ind}]{G} = 0$.\footnote{%
		In the (uninteresting) case $\loops{G}=\abs{T^c}=0$, we get $\uv{T^c}=0$ and thus $\sdc[\vec{\ind}]{\vec{\SPlog}}=\vec{\SPlog}\cdot \vec{\ind}$.
	}
	In particular, locally this function is a homogeneous bilinear form in $\vec{\SPlog}$ and $\vec{\ind}$.
\end{definition}
\begin{corollary}\label{cor:hepp-exp}
	The Hepp bound integral \eqref{eq:hepp-mellin} can be written as\footnote{%
	Instead of $\set{\SPlog_i=0}$, we can restrict to an arbitrary hyperplane $\set{\vec{\SPlog} \cdot \vec{\nu}=0}$ as long as $\nu_1+\ldots+\nu_N = 1$.
}
	\begin{equation}
		\Hepp{G,\vec{\ind}}
		%= \int_{\R^{N-1}} \left.e^{ -\sdc[\vec{\ind}]{\vec{\SPlog}}}\right|_{\SPlog_1=0} \td[N-1] \vec{\SPlog}
		= \int_{\R^{N}} e^{ -\sdc[\vec{\ind}]{\vec{\SPlog}}} \delta(\SPlog_i) \td[N] \vec{\SPlog}
		.
		\label{eq:hepp-exp}%
	\end{equation}
\end{corollary}
As the notation suggests, \autoref{def:sdc-vec} generalizes the superficial degree of convergence \eqref{eq:sdc}:
For the characteristic vector $\uv{\gamma} = \sum_{i\in\gamma} \uv{i}$ of a subgraph $\gamma \subseteq G$, we observe
\begin{equation*}
	-\uv{\gamma} \cdot \uv{T^c}
	= \abs{\gamma \cap T} - \abs{\gamma}
	= \nCG{\gamma}-\nCG{\gamma \cap T} - \loops{\gamma}
\end{equation*}
by \eqref{eq:euler} and $\loops{\gamma \cap T}=0$.
Here we consider both $\gamma$ and $\gamma\cap T$ as graphs on the same set of vertices, so the number of connected components $\nCG{\gamma \cap T}$ is at least $\nCG{\gamma}$.
Consequently,
\begin{equation}
	\sdc[\vec{\ind}]{\uv{\gamma}}
	= \uv{\gamma} \cdot \vec{\ind} - \frac{\Dim}{2} \loops{\gamma}
	+ \frac{\Dim}{2} \max_{T\in\STrees{G}} \Big( \nCG{\gamma}-\nCG{\gamma \cap T} \Big)
	= \sdc[\vec{\ind}]{\gamma}
	\label{eq:sdc-uv-gamma}%
\end{equation}
indeed coincides with \eqref{eq:sdc} and the maximum is attained precisely on all those trees $T$ for which $T\cap \gamma$ is a spanning forest of $\gamma$. Assuming $\Dim>0$, \eqref{eq:sdc-vec} is a maximum of linear forms, and hence $\sdc[\vec{\ind}]{\vec{\SPlog}}$ is a convex function of $\vec{\SPlog}$. It describes two convex polyhedra \cite{Handbook:BasicConvex}:
\begin{definition}\label{def:polytopes}
	We define the \emph{Newton polytope} $\Newton{G}(\vec{\ind})$ and its \emph{polar} $\Polar{G}(\vec{\ind})$ as
	\begin{align}
		\Newton{G}(\vec{\ind})
		&\defas \vec{\ind} -\tfrac{\Dim}{2} \conv \setexp{\uv{T^c}}{T\in \STrees{G}}
		\subset \R^N
		\quad\text{and}
	\label{eq:newton-conv}%
	\\
		\Polar{G}(\vec{\ind})
		&\defas \bigcap_{T\in\STrees{G}}
		\big\{\vec{\SPlog}\colon 
		\vec{\SPlog} \cdot \big(\vec{\ind}-\tfrac{\Dim}{2}\uv{T^c} \big) 
		\leq 1 
		\big\}
		\subset \R^N,
	\label{eq:polar-facets}%
	\end{align}
	where $\conv \set{\vec{v}_1,\ldots,\vec{v}_n} = \setexp{\sum_{i=1}^n \lambda_i \vec{v}_i}{\sum_{i=1}^n\lambda_i= 1\ \text{and all}\ \lambda_i \geq 0}$ is the convex hull. 
	We abbreviate the case of unit indices as $\Newton{G} \defas \Newton{G}(1,\ldots,1)$ and $\Polar{G} \defas \Polar{G}(1,\ldots,1)$.
\end{definition}
\begin{figure}
	\centering%
	\includegraphics[width=0.4\textwidth]{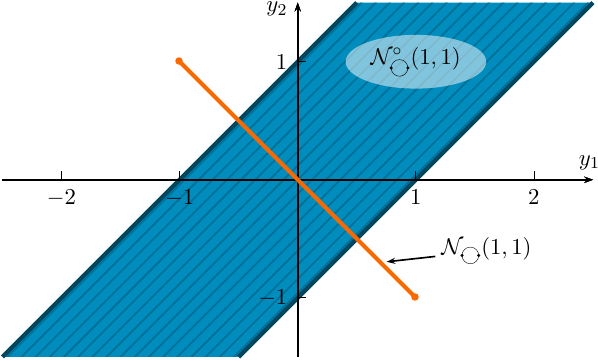}%
	\qquad%
	\includegraphics[width=0.4\textwidth]{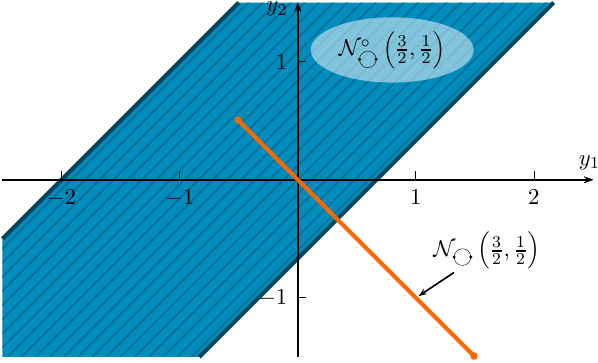}%
	\caption{The Newton polytope (line segment) and its polar (shaded unbounded region) for the bubble graph $\Cycle{2} = \Graph[0.3]{bubble}$. The figure to the left shows unit indices $\ind_1=\ind_2=1$, the other depicts $\vec{\ind}=(1.5,0.5)$.}%
	\label{fig:polar-bubble}%
\end{figure}
\begin{remark}
	Every vector $\vec{\ind} - \frac{\Dim}{2} \uv{T^c}$ is a vertex of $\Newton{G}(\vec{\ind})$ and lies in the hyperplane orthogonal to the diagonal vector $\vec{1} \defas (1,\ldots,1) = \uv{G}$, because $\vec{1}\cdot\uv{T^c} = \loops{G}$ as discussed above and $\vec{1}\cdot\vec{\ind} - \frac{\Dim}{2}\loops{G} = \sdc{G}=0$. Hence the dimension of $\Newton{G}(\vec{\ind})$ is at most $N-1$. It also implies that the function \eqref{eq:sdc-vec} is invariant under translations by $\vec{1}$, that is $\sdc[\vec{\ind}]{\vec{\SPlog}+\lambda \vec{1}} = \sdc[\vec{\ind}]{\vec{\SPlog}}$ for all real $\lambda$. Hence $\Polar{G}(\vec{\ind})$ contains the line $\R\cdot \vec{1}$. In summary,
\begin{align}
	\Newton{G}(\vec{\ind}) 
	&\subset \setexp{\vec{\SPlog}}{\vec{\SPlog} \cdot \vec{1} = 0}
	\subset \R^N
	\quad\text{and}
	\label{eq:Newton-hyperplane}%
	\\
	\Polar{G}(\vec{\ind})
	&= \Polar{G}(\vec{\ind}) + \R \cdot \vec{1}.
	\label{eq:Polar-lineality}%
\end{align}
\end{remark}
\begin{example}
	The cycle $\Cycle{N}$ with $N$ edges $\set{1,\ldots,N}$ has one spanning tree $T=\Cycle{N}\setminus i$ for each edge $i$, which contributes the monomial $\SP_i$ to $\PsiPol_{\Cycle{N}} = \sum_{i=1}^N \SP_i$. Hence
	\begin{equation*}
		\Newton{\Cycle{N}}(\vec{\ind})
		= \vec{\ind} - (\ind_1+\ldots+\ind_N) \conv \setexp{\uv{i}}{1 \leq i \leq N}
	\end{equation*}
	is an affine image of the standard simplex. For the bubble $\Cycle{2}=\Graph[0.3]{bubble}$ it is the line segment%with unit indices $\ind_1=\ind_2=1$, we have $\Dim=4$ and the associated polytope is the line segment
	\begin{equation*}
%		\Newton{\Graph[0.2]{bubble}} = \conv \set{
%			\begin{psmallmatrix} 1 \\ -1 \end{psmallmatrix},
%			\begin{psmallmatrix} -1 \\ 1 \end{psmallmatrix}
%		}
%		=\setexp{\begin{psmallmatrix} \lambda \\ -\lambda \end{psmallmatrix}}{-1 \leq \lambda \leq 1}
		\Newton{\Graph[0.2]{bubble}}(\ind_1,\ind_2) = \conv \set{
			\begin{psmallmatrix} \ind_1 \\ -\ind_1 \end{psmallmatrix},
			\begin{psmallmatrix} -\ind_2 \\ \ind_2 \end{psmallmatrix}
		}
		=\setexp{\begin{psmallmatrix} \lambda \\ -\lambda \end{psmallmatrix}}{-\ind_2 \leq \lambda \leq \ind_1}
		\subset \R^2
	\end{equation*}
	illustrated in \autoref{fig:polar-bubble}. Its polar is the unbounded region
	\begin{equation*}
		\Polar{\Graph[0.2]{bubble}}(\ind_1,\ind_2)
		=   
		\setexp{\vec{\SPlog}}{-\tfrac{1}{\ind_2} \leq \SPlog_1 - \SPlog_2 \leq \tfrac{1}{\ind_1}}
		\subset \R^2.
	\end{equation*}
\end{example}
The terminology in \autoref{def:polytopes} reflects that the convex hull $\conv \setexp{\uv{T^c}}{T\in \STrees{G}}$ of the exponents of the monomials in $\PsiPol_G$ is called the \emph{Newton polytope of $\PsiPol_G$}. Note that $\Newton{G}(\vec{\ind})$ is just the translate by $\vec{\ind}$ of this polytope, after scaling it by $-\frac{\Dim}{2}$. Similarly we can think of an affine transformation of the polytope $\conv\setexp{\uv{T}}{T\in \STrees{G}}$, which is called the \emph{spanning tree polytope} \cite{Chopra:STpolyhedron} of $G$ or more generally the \emph{matroid polytope} \cite{FeichtnerSturmfels:MatroidPolytopes} of $\GMat{G}$.
Indeed, the entire discussion in this section extends to arbitrary matroids by replacing the spanning trees $\STrees{G}=\Bases{\GMat{G}}$ with the bases $\Bases{M}$ throughout.
\begin{figure}
	$\prO{\SPlog_4=0} \Newton{\UM{4}{2}} = \Graph[1.0]{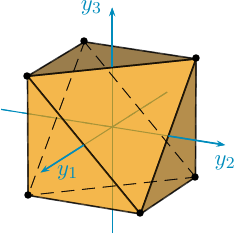}$
	\hfill
	$\Polar{\UM{4}{2}} \cap \set{\SPlog_4=0} = \Graph[1.2]{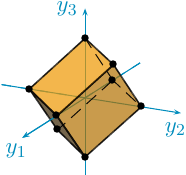}$
	\caption{The orthogonal projection of the Newton polytope of the matroid $\UM{4}{2}$ to the plane $\set{\SPlog_4=0}$, see \autoref{ex:Newton-U42}. Its polar is shown on the right.}%
	\label{fig:U42}%
\end{figure}
\begin{example}\label{ex:Newton-U42}
	The uniform matroid $\UM{4}{2}$ has $\binom{4}{2}=6$ bases of the form $\set{i,j}$ with $1\leq i<j\leq 4$, which form the vertices of an octahedron, see \autoref{fig:U42}. For unit indices $\ind_1=\cdots=\ind_4=1$, we find $\Dim=4$ and the associated polytope is explicitly
	\begin{equation*}
		\Newton{\UM{4}{2}} = \conv \set{
			\begin{psmallmatrix} 1 \\ 1 \\ -1 \\ -1 \end{psmallmatrix},
			\begin{psmallmatrix} 1 \\ -1 \\ 1 \\ -1 \end{psmallmatrix},
			\begin{psmallmatrix} 1 \\ -1 \\ -1 \\ 1 \end{psmallmatrix},
			\begin{psmallmatrix} -1 \\ 1 \\ 1 \\ -1 \end{psmallmatrix},
			\begin{psmallmatrix} -1 \\ 1 \\ -1 \\ 1 \end{psmallmatrix},
			\begin{psmallmatrix} -1 \\ -1 \\ 1 \\ 1 \end{psmallmatrix}
		}.
	\end{equation*}
\end{example}
\begin{remark}
	Matroid polytopes belong to the class of \emph{0-1-polytopes} \cite{Ziegler:Lectures01}, and $\Newton{M}$ is an affine image of such.
	In particular, when all indices $\ind_e=1$ are unity and the dimension equals $4$, like in \autoref{thm:Hepp=Vol}, then the vertices of $\Newton{M}$ are a subset of the cube $\set{-1,1}^N$.
\end{remark}
The piecewise linear function \eqref{eq:sdc-vec} characterizes the polytopes $\Newton{M}(\vec{\ind})$ and $\Polar{M}(\vec{\ind})$ as
\begin{equation}
	\sdc[\vec{\ind}]{\vec{\SPlog}}
	= \max_{\vec{z} \in \Newton{M}(\vec{\ind})} \left( \vec{\SPlog} \cdot \vec{z} \right)
	\quad\text{and}\quad
	\Polar{M}(\vec{\ind})
	= \big\{\vec{\SPlog}\colon
		\sdc[\vec{\ind}]{\vec{\SPlog}}
		\leq 1
	\big\}.
	\label{eq:newton-support}%
\end{equation}
In particular, $\sdc[\vec{\ind}]{\vec{\SPlog}}$ is also known as the \emph{support function} of $\Newton{M}(\vec{\ind})$ \cite{Handbook:BasicConvex}.
It is well known that the exponential integral \eqref{eq:hepp-exp} of the support function is the volume of the polar polytope, see for example \cite{Lasserre:LevelSetsPHF} and the references therein:
\begin{corollary}\label{cor:hepp-volume}
	Let $M$ denote a connected matroid on $\EG{M}=\set{1,\ldots,N}$ and pick any element $i\in \EG{M}$. For all indices $\vec{\ind}$ such that $\Polar{M}(\vec{\ind}) \cap \set{\SPlog_i=0}$ is bounded, we have
	\begin{equation}
		\Hepp{M,\vec{\ind}} 
		= (N-1)!\,\cdot
		\Vol \Big( 
			\,\Polar{M}(\vec{\ind}) \cap \set{\SPlog_i=0}
		\Big).
		\label{eq:hepp-volume}%
	\end{equation}
\end{corollary}
\begin{proof}
	The restriction $\sdc[\vec{\ind}]{\vec{\SPlog}}|_{\SPlog_i=0}$ is the support function of the projection of $\Newton{M}(\vec{\ind})$ onto the subspace $\set{\SPlog_i=0} \cong \R^{N-1}$ orthogonal to $\uv{i}$.
	The polar of this projection is $\Polar{M}(\vec{\ind}) \cap \set{\SPlog_i=0}$, and we can apply the formula discussed below equation~(2) in \cite{Lasserre:LevelSetsPHF}.
\end{proof}
\begin{example}
	For the bubble graph in \autoref{fig:polar-bubble}, the projection of $\Newton{\Graph[0.2]{bubble}}(\ind_1,\ind_2)$ onto $\SPlog_2=0$ is the line segment $[-\ind_2,\ind_1]$. Its polar is $\Polar{\Graph[0.2]{bubble}}(\ind_1,\ind_2) \cap \set{\SPlog_2=0} = [-\frac{1}{\ind_2},\frac{1}{\ind_1}]$ and has volume $\frac{1}{\ind_1}+\frac{1}{\ind_2} = \Hepp{\Graph[0.2]{bubble},\vec{\ind}}$, in agreement with \autoref{ex:hepp-bubble}.
\end{example}
We learn that the Hepp bound integrals \eqref{eq:hepp-mellin} and \eqref{eq:hepp-exp} converge precisely when the polyhedron $\Polar{M}(\vec{\ind}) \cap \set{\SPlog_i=0}$ is bounded. Equivalently, the orthogonal projection of $\Newton{M}(\vec{\ind})$ onto $\set{\SPlog_i=0}\cong \R^{N-1}$ must contain the origin in its interior. This requires that the dimension of $\Newton{M}\subset\{\vec{\SPlog}\cdot\vec{1}=0\}\subset \R^N$ is $N-1$ (as big as possible).
\begin{lemma}[{\cite[Corollary~3.10]{Fujishige:FacesBasePoly} and \cite[Proposition~2.4]{FeichtnerSturmfels:MatroidPolytopes}}]
	For a matroid $M$ on $\EG{M}=\set{1,\ldots,N}$, the dimension of the polytope $\Newton{M}(\vec{\ind}) \subset \R^N$ is equal to $N-\nCG{M}$.
\end{lemma}
For a disconnected matroid, the polar is thus never bounded. But if $M$ is connected, then the projection of $\Newton{M}$ onto $\set{\SPlog_i=0}$ has full dimension $N-1$. In this case, there exist translations $\vec{\ind}$ which put the origin inside, and then the polar is bounded. Explicitly, we see that $\vec{0} \in \Newton{M}(\vec{\ind})$ precisely when $\vec{\ind} \in \frac{\Dim}{2} \conv \setexp{\uv{T^c}}{T \in \Bases{M}}$, according to \eqref{eq:newton-conv}.
\begin{corollary}\label{cor:conv-conv}
	For a connected matroid $M$, the domain $\ConvDom \subset \R^N$ of absolute convergence of the integrals \eqref{eq:hepp-mellin} and \eqref{eq:hepp-exp} is not empty. It is the interior of the cone
% \autoref{cor:convergence-domain}
	\begin{equation*}
		\overline{\ConvDom} = \R_{\geq 0} \cdot \conv \setexp{\uv{T^c}}{T \in \Bases{M}}.
	\end{equation*}
\end{corollary}
It follows that the definitions of the Hepp bound by integrals \eqref{eq:hepp-mellin} and \eqref{eq:hepp-exp}, flags \eqref{eq:multi-hepp-from-sectors} and volumes \eqref{eq:hepp-volume} are completely equivalent:
\begin{enumerate}
	\item If $M$ is connected, then inside the region $\ConvDom$, the integral is finite and equal to the polar volume and the Hepp bound. This local information fixes the Hepp bound uniquely as a rational function by analytic continuation.
	\item If $M$ is not connected, the polar volume and the integral do not converge for any indices $\vec{\ind}$, and the combinatorial formula gives zero.
\end{enumerate}
%Our preference for the definition \eqref{eq:multi-hepp-from-sectors}, i.e.\ to assign a vanishing Hepp bound to disconnected matroids, is particularly convenient in light of the residue formula \eqref{eq:hepp-residue}.\todo{explain better}

\subsection{Singularities, facets and vertices}

The geometry explains the origin of the singularities of the Hepp bound. First, recall that every vertex $\vec{v}=\vec{\ind} - \frac{\Dim}{2} \uv{T^c}$ of the Newton polytope $\Newton{M}(\vec{\ind})$ lies in the half-spaces
\begin{equation*}
	h_{\gamma} \defas
	\setexp{\vec{\SPlog}}{\uv{\gamma} \cdot \vec{\SPlog} \leq \sdc[\vec{\ind}]{\gamma}}
	\subset \R^N
\end{equation*}
associated to every non-empty subset $\gamma \subsetneq M$, due to 
$\uv{\gamma} \cdot \vec{v} = \sdc[\vec{\ind}]{\gamma} - (\nCG{\gamma \cap T}-\nCG{\gamma}) \leq \sdc[\vec{\ind}]{\gamma}$
from \eqref{eq:sdc-uv-gamma}. Consequently, the Newton polytope is contained in the intersection of all these half-spaces. The boundary hyperplanes $\partial h_{\gamma}$ slice off the faces
\begin{equation*}
	F_{\gamma}
	\defas \Newton{M}(\vec{\ind}) \cap \setexp{\vec{\SPlog}}{\uv{\gamma} \cdot \vec{\SPlog} = \sdc[\vec{\ind}]{\gamma}}
%	\label{eq:Newton-faces}%
\end{equation*}
from $\Newton{M}(\vec{\ind})$.
When $\vec{\ind}$ approaches the boundary $\partial \ConvDom \subset \bigcup_{\gamma} \setexp{\vec{\ind}}{\sdc[\vec{\ind}]{\gamma} = 0}$ of the convergence cone, then the origin lands on some such face $F_{\gamma}$ so that the projected polar becomes unbounded.

The poles of the Hepp bound are therefore in bijection with the facets (faces of codimension one) of $\overline{\ConvDom} = \bigcap_{\gamma} \set{\sdc{\gamma} \geq 0}$ and $\Newton{M}(\vec{\ind})$. Hence \autoref{cor:hepp-poles} describes precisely those submatroids $\gamma \subset M$ for which the hyperplanes $\setexp{\vec{\SPlog}}{\sdc{\gamma}=0}$ and $\partial h_{\gamma}$ support facets of $\overline{\ConvDom}$ and $\Newton{M}(\vec{\ind})$, respectively.
This gives an alternative derivation of the well-known facet description of matroid polytopes:
\begin{lemma}[{\cite[Proposition~2.6]{FeichtnerSturmfels:MatroidPolytopes} and \cite[Corollary~3.14]{Fujishige:FacesBasePoly}}]
	Given a connected matroid $M$, the facets of the polytope $\Newton{M}$ are in bijection with the submatroids $\emptyset \neq \gamma\subsetneq M$ such that $\gamma$ and the quotient $M/\gamma$ are both connected.
\end{lemma}
We denote these submatroids in \eqref{eq:singularities} as $\Sing{M}$. The Newton polytope is therefore
\begin{equation}
	\Newton{M}(\vec{\ind})
	= 
	\bigcap_{\gamma \in \Sing{M}} \setexp{\vec{\SPlog}}{\uv{\gamma} \cdot \vec{\SPlog} \leq \sdc[\vec{\ind}]{\gamma}}
	\ \cap\ 
	\setexp{\vec{\SPlog}}{\vec{1}\cdot\vec{\SPlog}=0}
	\label{eq:Newton-H}%
\end{equation}
and none of these constraints is redundant. From these facets we read off the vertices of the polar, making the divergences of the Hepp bound on $\sdc[\vec{\ind}]{\gamma}\rightarrow 0$ for $\gamma \in \Sing{M}$ apparent:
\begin{corollary}\label{cor:polar-V}
	The polar of the Newton polytope of a connected matroid $M$ is
	\begin{equation}
		\Polar{M}(\vec{\ind})
		= \R \cdot \vec{1}
		+ \conv 
		\setexp{\frac{\uv{\gamma}}{\sdc[\vec{\ind}]{\gamma}}}{\gamma \in \Sing{M}},
		\label{eq:polar-V}%
	\end{equation}
	and no $\gamma\in\Sing{M}$ is redundant. Each such $\gamma$ labels a vertex of the intersection
	\begin{equation}
		\Polar{M}(\vec{\ind}) \cap \set{\SPlog_i=0}
		= \conv \left( 
		\setexp{\frac{\uv{\gamma}}{\sdc[\vec{\ind}]{\gamma}}}{i\notin\gamma \in \Sing{M}}
			\cup
			\setexp{\frac{-\uv{\gamma^c}}{\sdc[\vec{\ind}]{\gamma}}}{i\in\gamma \in \Sing{M}}
		\right).
		\label{eq:polar-V-yi=0}%
	\end{equation}
\end{corollary}
\begin{example}
	The uniform matroid $\UM{4}{2}$ has precisely $8$ singular submatroids $\gamma\in\Sing{\UM{4}{2}}$: the singletons $\set{e}$ and their complements (\autoref{ex:uniform-poles}).
	For unit indices, they all have $\sdc{\gamma}=1$ and we can read off the vertices of the polar in \autoref{fig:U42} directly from \eqref{eq:polar-V-yi=0}:
	\begin{equation*}
		\Polar{\UM{4}{2}} \cap \set{\SPlog_4=0}
		= \conv
		\Big\{
			\underset{\color{gray}\set{1}}{\Big(\begin{smallmatrix} 1 \\ 0 \\ 0 \end{smallmatrix}\Big)},
			\underset{\color{gray}\set{2}}{\Big(\begin{smallmatrix} 0 \\ 1 \\ 0 \end{smallmatrix}\Big)},
			\underset{\color{gray}\set{3}}{\Big(\begin{smallmatrix} 0 \\ 0 \\ 1 \end{smallmatrix}\Big)},
			\underset{\color{gray}\set{4}}{\Big(\begin{smallmatrix}-1 \\-1 \\-1 \end{smallmatrix}\Big)},
			\underset{\color{gray}\set{2,3,4}}{\Big(\begin{smallmatrix}-1 \\ 0 \\ 0 \end{smallmatrix}\Big)},
			\underset{\color{gray}\set{1,3,4}}{\Big(\begin{smallmatrix} 0 \\-1 \\ 0 \end{smallmatrix}\Big)},
			\underset{\color{gray}\set{1,2,4}}{\Big(\begin{smallmatrix} 0 \\ 0 \\-1 \end{smallmatrix}\Big)},
			\underset{\color{gray}\set{1,2,3}}{\Big(\begin{smallmatrix} 1 \\ 1 \\ 1 \end{smallmatrix}\Big)}
		\Big\}.
		\underset{\color{gray}(\leftarrow \gamma )}{\phantom{\Big(\begin{smallmatrix} 1 \\ 0 \\ 0 \end{smallmatrix}\Big)}}
	\end{equation*}
	This rhombohedron extends the cross-polytope $\conv\set{\pm \uv{1},\pm \uv{2},\pm \uv{3}}$ with volume $4/3$ by two tetrahedra $\pm \conv\set{\uv{1},\uv{2},\uv{3},\uv{1}+\uv{2}+\uv{3}}$ which each have volume $1/3$. So the total volume is $2$ and we find $\Hepp{\UM{4}{2}}= 3!\cdot 2 = 12$ from \eqref{eq:hepp-volume} in agreement with \eqref{eq:hepp-uniform}.
\end{example}
\begin{remark}
	The Newton polytope differs from the \emph{Feynman polytope} defined in \cite{Brown:FeynmanAmplitudesGalois}, which has many more facets. They are labelled by \emph{all} `motic' (bridgeless) subgraphs of the given graph $G$, whereas only very few of those belong to $\Sing{G}$.
\end{remark}

\subsection{Factorization}
The facet $F_{\gamma}$ defined above for a submatroid $\gamma \in \Sing{M}$ contains precisely those vertices $\vec{v} = \vec{\ind}-\frac{\Dim}{2} \uv{T^c}$ of $\Newton{M}(\vec{\ind})$ such that $\gamma \cap T$ is a basis of $\gamma$, see \eqref{eq:sdc-uv-gamma}. In this case we see that $T\setminus \gamma$ is a basis of $M/\gamma$, and hence the vertices of $F_{\gamma}$ are in bijection with the Cartesian product $\Bases{\gamma}\times\Bases{M/\gamma}$.
The corresponding factorization
\begin{equation}
	\partial \Newton{M}(\vec{\ind})
	\supset
	F_{\gamma}
%	\defas \Newton{M}(\vec{\ind}) \cap \set{\vec{\SPlog} \cdot \uv{\gamma} = \sdc{\gamma}}
	\cong \Newton{\gamma} \times \Newton{M/\gamma}
	\label{eq:Newton-facet-product}%
\end{equation}
is the geometric incarnation of the algebraic relation $\PsiPol_{M} \equiv \PsiPol_{\gamma} \PsiPol_{G/\gamma} \mod I$, where $I$ is the ideal generated by monomials of degree $\loops{\gamma}+1$ in the variables $\setexp{\ind_e}{e\in\gamma}$. We exploited this relation in \eqref{eq:Psi-flag-ineq}, and it is absolutely fundamental in the study of Feynman amplitudes and graph polynomials \cite{Schultka:ToricFeynman,Brown:FeynmanAmplitudesGalois,BlochEsnaultKreimer:MotivesGraphPolynomials,BrownKreimer:AnglesScales}.

Regarding the Hepp bound, we see from \eqref{eq:polar-V} that the divergence at $\sdc[\vec{\ind}]{\gamma} \rightarrow 0$ comes from a pyramid with the runaway apex $\uv{\gamma}/\sdc[\vec{\ind}]{\gamma}$. Its volume comes from the piece of the integral \eqref{eq:hepp-exp} where $\vec{\SPlog}$ ranges over the cone $\R_{>0} F_{\gamma}$. The factorization \eqref{eq:Newton-facet-product} thus leads to a product formula for the residue, namely \eqref{eq:hepp-residue}.

\subsection{Computations}
The calculation of the volume of a convex polytope is a difficult, but very well studied problem. In small dimensions, we can resort to exact algorithms like {\lrs} \cite{Avis:RevisedLRS}. We tested this program on all $\field^4$ periods with at most $7$ loops. Taking the vertex description \eqref{eq:polar-V-yi=0} as input, {\lrs} determines the facets and computes the volume. In all cases, these facets match precisely the spanning trees according to \eqref{eq:polar-facets}, and the volume reproduces the Hepp bound computed by \eqref{eq:hepp-1pi-flags} in line with \eqref{eq:hepp-volume}.
We found {\cddrp} \cite{Fukuda:cddr+} to be more efficient for the transformation between facets and vertices, but it does not provide a volume computation.

Such exact volume determinations are very time consuming and in our tests only practical up to around 15 dimensions (graphs with $16$ edges). Not surprisingly, a direct combinatorial recursion like \eqref{eq:hepp-flag-recursive} is much more efficient. Our implementation of this method succeeded for graphs with $30$ edges. The Hepp bound thus provides a dataset of polytopes in a large number of dimensions with exactly known volumes, which might be useful as a benchmark for general volume computation algorithms.

We are not aware of previous work on the polar volume of matroids, but the volume of the matroid polytope itself has been studied in the literature. For example, a combinatorial formula was given in \cite{ArdilaBenedettiDoker:MatroidPolytopes} using the theory of generalized permutohedra \cite{Postnikov:GenPermuto}. For fixed rank or corank, the volume of a matroid polytope can be computed in polynomial time \cite{DeLoeraHawsKoeppe:Ehrhart}. Our \autoref{cor:Hepp-polynomial-time} is the same statement about the volume of the polar.

Recently it has become feasible to \emph{approximate} volumes of polytopes in very high dimension, using rapidly mixing random walks. At least three implementations of such methods are readily available \cite{GeMa:FastPractical,CousinsVempala:PracticalVolume,EmirisFisikopoulos:PPVA}. These techniques open up the fascinating possibility to study properties of very large graphs and matroids, which might give insights into the asymptotic behaviour of the perturbation series of quantum field theory.

\begin{table}
	\centering
\begin{tabular}{rccccc}\toprule
	$G$ & $P_{7,4}{\setminus} r$ & $P_{7,4}{\setminus} q$ & $P_{7,4}{\setminus} t$ & $P_{7,4}{\setminus} p$ & $P_{7,4}{\setminus} s$ \\
	\midrule
	$\abs{\STrees{G}}$ & 1137 & 1166 & 1168 & 1197 & 1248 \\
	$\abs{\Sing{G}}$ & 66 & 64 & 69 & 69 & 81 \\
	\bottomrule
\end{tabular}
\qquad
$ P_{7,4} = \Graph[0.75]{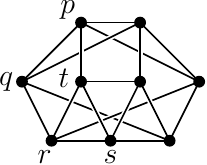}$ %
\caption{Facet and vertex counts of $\Polar{G}$ for different uncompletions of the graph $P_{7,4}$.}%
	\label{tab:polytope-uncompletions}%
\end{table}
\begin{remark}
	Completion invariance identifies the volumes of combinatorially distinct polytopes. The graph from \cite{Schnetz:Census} in \autoref{tab:polytope-uncompletions} has five non-isomorphic uncompletions $G=P_{7,4}{\setminus}v$ and each of them has a different number of spanning trees (facets of $\Polar{G}$).
	\autoref{tab:polytope-uncompletions} also shows the number of divergent subgraphs (vertices of $\Polar{G}$). By \autoref{thm:hepp-completion}, the volumes $13!\cdot \Vol (\Polar{G} \cap \set{\SPlog_i=0}) = \Hepp{G} = 123260$ of these five different polar polytopes are all the same.
	The volumes of the Newton polytopes $\Newton{G}$, however, are all distinct.

	In contrast, the duality \eqref{eq:hepp-duality} amounts to a simple reflection: $\Polar{G^{\dual}} = -\Polar{G}$.
\end{remark}

\subsection{Shape}
\label{sec:shape}

We saw in \autoref{ex:uniform-poles} that uniform matroids $\UM{n}{r}$ have only $2n$ vertices and huge number of bases, namely $\binom{n}{r}$. Their polars $\Polar{\UM{n}{r}}$ are similar to the cross-polytopes $\conv\set{\pm \uv{1},\ldots,\pm \uv{n}}$, which also have $2n$ vertices, and the maximal number $2^n$ of facets.

For {\PlogDiv} graphs $G$ in $4$ dimensions ($n=\abs{G}=2\ell$), like the $\field^4$ graphs considered above, we find experimentally a roughly similar behaviour. The number of spanning trees grows exponentially with $n$; this follows from \cite{McKay:SpanningRegular}. Therefore, $\Polar{G}$ has an exponentially large number of facets.
The number $\abs{\Sing{G}}$ of vertices of $\Polar{G}$, in contrast, appears to grow very slowly in comparison. For the zigzag graphs from \cite{BroadhurstKreimer:KnotsNumbers}, for example, one can show that $\abs{\Sing{G}}$ is quadratic in $n$.

The similarity between $\Polar{G}$ and the cross-polytope is particularly pertinent in regards of the volumes. We see from \eqref{eq:polar-facets} that
$	\pm \uv{k} \in \Polar{G} $
for all $1\leq k \leq n$. These correspond to single edge subgraphs and their complements. It follows that $\Polar{G} \cap \set{\SPlog_n=0}$ contains the cross-polytope of dimension $n-1=2\ell-1$, which implies the lower bound
\begin{equation}
	\Hepp{G} \geq
	(n-1)! \cdot \Vol \left( \conv \set{\pm \uv{1},\ldots,\pm \uv{n-1}} \right)
	=2^{n-1}.
	\label{eq:hepp-lower}%
\end{equation}
In fact, we get a slightly better lower bound $\Hepp{}(\UM{2\ell}{\ell})$ according to \autoref{rem:um-minimal}. However, the important point to note is that $\Hepp{G}$ grows at least exponentially with $n$.

Note that $\Polar{G}$ has an insphere of radius $1/\sqrt{n}$: The vectors $\vt{T} \defas \uv{T} - \uv{T^c} \in \set{-1,1}^n$ solve $\sdc{\vt{T}/n} = 1$ by \eqref{eq:polar-facets} and they have norm $\norm{\vt{T}}=\sqrt{n}$. The volume of this sphere grows like $(1/\sqrt{n})^n/\Gamma(1+n/2) \sim 1/n!$ for large $n$, up to exponential factors. With the prefactor $(n-1)!$ in \eqref{eq:hepp-volume}, we see that the volume of the cross-polytope and the insphere are of comparable magnitude.

The Hepp bound could however be much larger than \eqref{eq:hepp-lower}. Since the vertices $\pm\uv{k}$ have distance $1$ from the origin, the circumsphere of $\Polar{G}$ has radius $1$ and thus its volume exceeds that of the insphere by a factor of $(\sqrt{n})^n \sim \sqrt{n!}$.

It was proven in \cite{CalanRivasseau:Phi44} that the periods of $\field^4$ graphs grow only exponentially, and by \eqref{eq:hepp-period-bound} this implies an exponential upper bound on $\Hepp{G}$ as well. We conclude that the volume of $\Polar{G}$ is concentrated near the insphere.

\begin{remark}
	The polytopes $\Newton{G}$ are typically not simple and far from simplicial. A facet $F_{\gamma}=\Newton{G} \cap \setexp{\vec{\SPlog}}{\uv{\gamma}\cdot \vec{\SPlog} = \sdc{\gamma}}$ corresponding to a singular subgraph $\gamma\in\Sing{G}$ is bounded by $\abs{\STrees{\gamma}}\cdot|\STrees{G/\gamma}|$ vertices, and a vertex $\vt{T}$ belongs to all facets $F_{\gamma}$ such that $\gamma\cap T$ spans $\gamma$.
\end{remark}
\begin{example}
	The complete graph $K_4$ has 16 singular subgraphs: $6$ singletons $\set{e}$, their complements $K_4{\setminus}e\cong \gKite$, and four triangles $\cong\gTri$. The vertex $\vt{T}$ of $\Newton{K_4}$ corresponding to a star $T\cong\Graph[0.15]{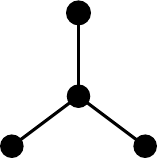}$ lies on $9$ facets: the three edges $e\in T$, the complements $K_4{\setminus}f$ of the three edges $f\notin T$, and $3$ of the triangles. In the case of a path $T\cong\Graph[0.15]{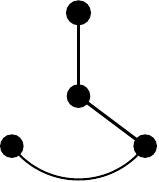}$, the vertex $\vt{T}$ lies on $8$ facets, because it only supports $2$ triangles.
	Conversely, the $4$-dimensional facets $F_{\gamma}$ have $9$ vertices if $\gamma\cong\gTri$; in this case, $F_{\gamma} \cong \Newton{\Graph[0.2]{tri}} \times \Newton{\Graph[0.2]{triple}}$ is a product of simplices.
	The other facets $F_{e} \cong \Newton{\Graph[0.2]{tri221}}$ and $F_{K_4{\setminus}e}\cong\Newton{\Graph[0.2]{hkite}}$ have $8$ vertices each.
\end{example}

\subsection{Tree decomposition}
The facets of the polar $\Polar{G}$ are in bijection with the spanning trees $T$ of the graph $G$, according to \eqref{eq:polar-facets}. For unit indices in $4$ dimensions, these facets are concretely
\begin{equation*}
	F_T \defas \Polar{G} \cap \setexp{\vec{\SPlog}}{\vec{\SPlog} \cdot (\uv{T}-\uv{T^c}) = 1}.
\end{equation*}
They induce a decomposition $\Polar{G} = \bigcup_{T\in\STrees{G}} \conv(\{\vec{0}\}\cup F_T)$ into pyramids with disjoint interiors. The volume of such a pyramid corresponds in \eqref{eq:hepp-exp} to the integral over the region where $\sdc{\vec{\SPlog}} = \vec{\SPlog} \cdot (\uv{T}-\uv{T^c})$. In Schwinger parameters \eqref{eq:hepp-mellin}, this is the sector
\begin{equation}
	\HeppSec{T}
	\defas \Big\{
		\vec{\SP}
	\colon
		\prod_{e\notin T} \SP_e > \prod_{e\notin T'} \SP_e 
		\quad\text{for all}\quad T' \in \STrees{G}
	\Big\}
%	= \Big\{
%		\vec{\SP}
%	\colon
%		\PsiTrop_G(\vec{\SP}) = \prod_{e\notin T} \SP_e
%	\Big\}
	= \bigcup_{\sigma\colon T_{\sigma}=T} \HeppSec{\sigma}
	\label{eq:tree-sector}%
\end{equation}
where the maximum monomial $\PsiTrop_G(\vec{\SP}) = \prod_{e\notin T} \SP_e$ is given by $T$. Each $\HeppSec{T}$ subsumes many Hepp sectors $\HeppSec{\sigma}$ as determined by \autoref{lem:Kruskal}. The pyramid decomposition corresponds thus to the partition $\bigcup_{T\in\STrees{G}} \HeppSec{T}$ of the integration domain in \eqref{eq:hepp-mellin}, such that
\begin{equation}
	\Hepp{G}
	= \sum_{T\in\STrees{G}} \Hepp{T\subset G}
	\quad\text{where}\quad
	\Hepp{T\subset G}
	\defas 
	\int_{\HeppSec{T}} \frac{\Omega}{\prod_{e\notin T} \SP_e^2}.
	\label{eq:hepp-from-trees}%
\end{equation}
\begin{lemma}\label{lem:hepp-tree-max}
	Given a spanning tree $T\in\STrees{G}$ and an edge $e\notin T$, let $\TreePath{T}{e} \subseteq T$ denote the unique path in $T$ that connects the endpoints of $e$. Then $T$'s contribution to $\Hepp{G}$ is
	%For $f\in T$, write $\TreeCut{T}{f} \subsetneq T^c$ for the set of edges $e$ such that $f\in \TreePath{T}{e}$.
	\begin{equation}
		\Hepp{T\subset G}
		=
		\int_{\Projective^T}
%		\bigg(\prod_{i\in T} \int_0^{\infty} \td \SP_i \bigg)
%		\delta(1-\SP_f)
%		\prod_{e\notin T} \frac{1}{\max\setexp{\SP_k}{k \in \TreePath{T}{e}}}.
		\frac{\Omega}{\prod_{e\notin T}\max\setexp{\SP_k}{k \in \TreePath{T}{e}}}.
		\label{eq:hepp-tree-max}%
	\end{equation}
\end{lemma}
\begin{proof}%[Proof of \autoref{lem:hepp-tree-max}]
	From Kruskal's greedy algorithm (\autoref{lem:Kruskal}), we see that
	\begin{equation*}
		\HeppSec{T} = \Big\{
			\vec{\SP}
		\colon
			\SP_e > \max\setexp{\SP_k}{k\in \TreePath{T}{e}}
			\quad\text{for all}\quad
			e\notin T
		\Big\}.
	\end{equation*}
	The integrals of $\td \SP_e/\SP_e^2$ over all $e\notin T$ therefore yield the claim.
\end{proof}
\begin{example}\label{ex:K4-tree-hepp}
	Consider the star $T=\set{2,4,6}= \Graph[0.15]{w3star} \subset \Graph[0.15]{w3small}$ in $G=K_4$ with the labels of \autoref{fig:completion-K5}. Each edge $e\notin T$ is connected in $T$ by two edges: $\TreePath{T}{1} = \set{2,6}$, $\TreePath{T}{3}=\set{2,4}$ and $\TreePath{T}{5}=\set{4,6}$. The corresponding projective integral
	\begin{equation*}
		\Hepp{ \Graph[0.2]{w3star}\subset \Graph[0.2]{w3small} }
		= \int_{\Projective^T} \frac{\Omega}{
			\max\set{\SP_2,\SP_4}
			\max\set{\SP_2,\SP_6}
			\max\set{\SP_4,\SP_6}
		}
		= 6
	\end{equation*}
	is easily evaluated by setting $1=\max\set{\SP_2,\SP_4,\SP_6}$ to obtain $6$ times the affine integral $\int_{0<\SP_2<\SP_4<\SP_6=1} 1/\SP_4 = 1$.
	For the path $T= \set{2,4,5}=\Graph[0.15]{w3path}$ we find $\TreePath{T}{6} = \set{4,5}$ and note that $\TreePath{T}{1}=T$ has three edges. Therefore, the projective integral is smaller and we find
	\begin{equation*}
		\Hepp{ \Graph[0.2]{w3path}\subset \Graph[0.2]{w3small} }
		= \int_{\Projective^T} \frac{\Omega}{
			\max\set{\SP_2,\SP_4}
			\max\set{\SP_2,\SP_4,\SP_5}
			\max\set{\SP_4,\SP_5}
		}
		= 5.
	\end{equation*}
	As $K_4$ has $4$ stars and $12$ paths as spanning trees, its Hepp bound is $4\cdot 6 + 12 \cdot 5 = 84$.
\end{example}
\begin{remark}\label{rem:hepp-tree-upper-bound}
	The contribution $\Hepp{T\subset G}$ depends not only on $T$, but also on $G$ and on the embedding of $T$ into $G$. It is possible, however, to give an intrinsic upper bound on $\Hepp{T\subset G}$ that only depends on $T$. Such a bound is constructed in \cite[Appendix~B]{CalanRivasseau:Phi44} and shown to grow only exponentially with the size of the tree.
	% These estimates are crucial to prove local existence of the Borel transform in $\field^4$ theory \cite{CalanRivasseau:Phi44}.
\end{remark}
\begin{remark}
	Tree contributions \eqref{eq:hepp-tree-max} to the Hepp bound are not to be confused with \emph{constructive tree weights} $w(G,T)$ considered in \cite{RivasseauWang:HowToResum}. The latter sum to $1$ and count the fraction of Hepp sectors $\HeppSec{\sigma}$ where $T=T_{\sigma}$ gives the maximal monomial $\prod_{e\notin T} \SP_e$ of $\PsiPol_G$.
	These weights also have an integral representation \cite[Theorem~2.1]{RivasseauWang:HowToResum}
	\begin{equation*}
		w(G,T) = \left( \prod_{i\in T} \int_0^1 \td \SP_i \right)
		\prod_{e\notin T} \min \setexp{\SP_k}{k\in \TreePath{T}{e}},
	\end{equation*}
	but in \autoref{ex:K4-tree-hepp} this results in $w(K_4,\Graph[0.15]{w3star}) = 1/15$ and $w(K_4,\Graph[0.15]{w3path}) = 11/180$.
\end{remark}

\subsection{Period correlation}
\label{sec:period-correlation}

For {\PlogDiv} graphs with unit indices in $\Dim=4$ dimensions, the support function \eqref{eq:sdc-vec} is
\begin{equation}
	\sdc{\vec{\SPlog}} = \max_{T\in\STrees{G}} \vec{\SPlog} \cdot \vt{T}
	\qquad\text{where}\qquad
	\vt{T} \defas \uv{T}-\uv{T^c}
	=\sum_{k\in T} \uv{k} - \sum_{k\notin T} \uv{k}.
	\label{eq:vt}%
\end{equation}
Integrating over the norm $\lambda \defas \norm{\vec{\SPlog}}$ in \eqref{eq:hepp-exp} gives $\int_0^{\infty} \lambda^{N-2} e^{-\lambda \omega} \td \lambda = (N-2)!/\omega^{N-1}$, so
\begin{equation}
	\Hepp{G} = (N-2)! \int_{S^{N-1} \cap \set{\SPlog_i=0}} \frac{\Omega}{\sdc{\vec{\SPlog}}^{N-1}}
	\label{eq:hepp-projective}%
\end{equation}
is an integral over the $N-2$ dimensional sphere $S^{N-1} \cap \set{\SPlog_i=0}\cong S^{N-2}$. The period has a very similar representation that explains the correlation.
%We can replace $\delta(\SPlog_i)$ in \eqref{eq:hepp-exp} with $\delta(\frac{1}{N}\sum_e \SPlog_e)$, corresponding to $\delta(1-(\prod_e\SP_e)^{1/N})$ in \eqref{eq:hepp-mellin}, and then
%\begin{equation}
%	\Hepp{G} = (N-2)! \sqrt{N} \int_{S^{N-1} \cap \set{\SPlog_1+\cdots+\SPlog_N=0}} \frac{\td[N-2]\vec{\SPlog}}{\sdc{\vec{\SPlog}}^{N-1}}
%	\label{eq:hepp-projective}%
%\end{equation}
\begin{lemma}\label{lem:Period-SPlog}
	The period of a {\PlogDiv} graph $G$ with $N$ edges in $4$ dimensions is equal to
	\begin{equation}
		\Period{G}
		= (N-2)! \int_{S^{N-1} \cap \set{\SPlog_i=0}}
		\frac{\Omega}{\sdc{\vec{\SPlog}}^{N-1}} \RadInt{\vec{\SPlog}},
		\label{eq:period-projective}
	\end{equation}
	where $
		\RadInt{\vec{\SPlog}}=\RadInt{\lambda \vec{\SPlog}} 
%		\in \big[1/\abs{\STrees{G}}^{2},1\big]
	$ 
	for all $\lambda>0$ takes values between $1/\abs{\STrees{G}}^2$ and $1$, and is defined by
	\begin{equation}
		\RadInt{\vec{\SPlog}}
		\defas 
		\int_0^{\infty}
		\frac{\lambda^{N-2}}{(N-2)!}
		\bigg\{
			\sum_{T \in \STrees{G}}
			\exp\left( \frac{\lambda}{2} \frac{\vec{\SPlog} \cdot \vt{T}}{\sdc{\vec{\SPlog}}} \right)
		\bigg\}^{-2}
		\td \lambda
%		\in \left[ \abs{\STrees{G}}^{-\Dim/2} ,1 \right]
		.
		\label{eq:radint}%
	\end{equation}
\end{lemma}
\begin{proof}
	Substitute $\SP_k= \exp(-\lambda\SPlog_k/\sdc{\vec{\SPlog}})$ into \eqref{eq:period-mellin}, this gives \eqref{eq:period-projective}. For the lower bound on $\RadInt{\vec{\SPlog}}$, note that the exponents are at most $\lambda/2$, because $\vec{\SPlog} \cdot \vt{T} \leq \sdc{\vec{\SPlog}}$ by \eqref{eq:vt}, so
	\begin{equation*}
		\RadInt{\vec{\SPlog}} \geq
		\int_0^{\infty}
		\frac{\lambda^{N-2}}{(N-2)!}
		\left\{
			\abs{\STrees{G}}
			\exp\left( \lambda/2\right)
		\right\}^{-2}
		=
		\abs{\STrees{G}}^{-2}
		.
%		\int_0^{\infty}
%		\frac{\lambda^{N-2}}{(N-2)!}
%		e^{-\lambda} \td \lambda
	\end{equation*}
	Furthermore, the spanning tree $T$ that dominates for a given value of $\vec{\SPlog}$ contributes a summand with $\vec{\SPlog} \cdot \vt{T} = \sdc{\vec{\SPlog}}$, which gives the upper bound $\RadInt{\vec{\SPlog}} \leq 1$.
\end{proof}
By comparison with \eqref{eq:hepp-projective}, the bounds $1/\abs{\STrees{G}}^2 \leq \RadInt{\vec{\SPlog}} \leq 1$ imply the relations \eqref{eq:hepp-period-bound}. The corresponding lower bound
%	\begin{equation}
$		\Period{G} \geq 2^{N-1}/\abs{\STrees{G}}^2 $
%		\label{eq:period-lower}%
%	\end{equation}
from \eqref{eq:hepp-lower} can be improved easily:%
\begin{lemma}
	The period of a {\PlogDiv} graph in $\Dim=4$ is at least
%	\begin{equation}
$		\Period{G} \geq 4^{N-1}/\abs{\STrees{G}}^2$.
%		\label{eq:period-lower-bound}%
%	\end{equation}
	\label{lemma:ST-lower-bound}%
\end{lemma}
\begin{proof}
	Contraction and deletion of an edge show that $\PsiPol_G = \SP_1 \PsiPol_{G\setminus 1 } + \PsiPol_{G/1}$. Therefore,
	\begin{equation*}
		\int_0^{\infty} \frac{\td \SP_1}{\PsiPol_G^2}
		= 
		\frac{1}{\PsiPol_{G\setminus 1} \PsiPol_{G/1}}
		\geq \frac{4}{(\PsiPol_{G\setminus 1}+\PsiPol_{G/1})^2}
		= \frac{4}{\PsiPol_G^2|_{\SP_1=1}}
	\end{equation*}
	by the arithmetic-geometric mean inequality. Repeating this estimate for the remaining edges $2,\ldots,N-1$ and setting $\SP_N=1$ proves the claim, since
	$\PsiPol_G|_{\SP_1=\cdots=\SP_N=1} = \abs{\STrees{G}}$.
\end{proof}
As discussed in \autoref{sec:shape}, the exponential upper bound on periods (and hence on the Hepp bounds) from \cite{CalanRivasseau:Phi44} implies that the integral \eqref{eq:hepp-projective} is dominated by the contributions from regions where $\sdc{\vec{\SPlog}}$ is large, that is, where $\vec{\SPlog}$ is close to one of the spanning tree directions $\vt{T}$.

Within a tree sector $\HeppSec{T}$, the function $\sdc{\vec{\SPlog}}=\vec{\SPlog}\cdot \vt{T}$ is linear. Together with the well-known log-concavity of matroid polynomials \cite{AnariGharanVinzant:LogConcaveBases,NagaokaYazawa:StrictLog}, it can be shown that the function $\RadInt{\vec{\SPlog}}$ is also log-concave within each tree sector. Consequently, it is minimized at the vertices \eqref{eq:polar-V} of the polar $\Polar{G}$:
\begin{equation*}
	\frac{\Period{G}}{\Hepp{G}}
	\geq
	\min_{\gamma \in \Sing{G}} \RadInt{ \uv{\gamma}/\sdc{\gamma} }. 
\end{equation*}

More importantly, the log-concavity implies that $\RadInt{\vec{\SPlog}}$ has a unique maximum in each tree sector; far away from the vertices and thus closer to the centre $\vt{T}$ of the corresponding facet of $\Polar{G}$. As explained above, this region dominates the integral \eqref{eq:hepp-projective}.

The arguments above give an intuition and qualitative explanation for the correlation between $\Hepp{G}$ and $\Period{G}$. Developing these ideas further, it should be possible to give a rigorous quantitative formulation and proof of the correlations observed in \autoref{fig:correlation}.

\section{Outlook}
\label{sec:outlook}

We explored the first properties of the Hepp bound and illustrated its rich structure, connecting algebraic, combinatorial and geometric aspects. Many interesting questions remain, including:
\begin{enumerate}
	\item What is the reason behind the unexplained relations like \eqref{eq:hepp-8loop-pairs}? Are the corresponding periods identical as predicted by \autoref{con:faithful}?
\item Which ${\PlogDiv}$ matroids maximize the Hepp bound? In line with \cite[Conjecture~1]{BrownSchnetz:ZigZag}, we conjecture that the zig-zag graphs give the largest Hepp bounds among primitive $\field^4$ graphs. What are the maximizers in other classes of graphs and matroids?
\item What is the computational complexity of the Hepp bound? Many techniques have been proposed to compute matroid polytope volumes, see for example \cite{UmerAshraf:AnotherVolumeMatroid}; can these be applied to compute also the polar volumes more efficiently?
\item How can one prove the strong correlation between the period and the Hepp bound? Can it be improved to provide even better approximations for periods?
\item
How can one extract other matroid invariants from $\Hepp{M,\vec{\ind}}$ as in \autoref{sec:other-invariants}?
\end{enumerate}
The Hepp bound $\Hepp{M}$ with unit indices is not defined when $M$ has a divergent submatroid $\gamma\subsetneq M$ such that $\abs{\gamma}/\loops{\gamma} = \abs{M}/\loops{M}$. This can be remedied as follows:
Lift the restriction to $\sdc{M}=0$, and consider instead the dimension $\Dim$ as a free parameter. Then extend the summands in \eqref{eq:multi-hepp-from-sectors} by the missing $N$th denominator, and
\begin{equation*}
	f_d(M)
	=
	\sum_{\sigma \in \perms{M}} \frac{1}{\sdc{M^{\sigma}_1}\cdots \sdc{M^{\sigma}_{N}}}
	\in \Q(\Dim)
\end{equation*}
is a well-defined rational function of $\Dim$, with poles of the form $\Dim=2\abs{\gamma}/\loops{\gamma}$. The residue at $\sdc{M}=0$ would be the Hepp bound, were it not for the higher order of that pole due to the presence of a divergence. Note that $f_d(M) = \Char(\Phi(M))$ in terms of the map $\Phi(M)$ defined in \autoref{sec:other-invariants}. In fact, $\Phi(M)$ is a multiplicative map of Hopf algebras, and with \eqref{eq:Char-shuffle}, also $f_d(M)$ is multiplicative.
See \cite{Schmitt:Incidence,CrapoSchmitt:FreeSubMatroids} for the Hopf algebra of matroids.

The standard Hopf-algebraic renormalization techniques \cite{CK:RH1,BrownKreimer:AnglesScales} can therefore be applied, and one obtains a renormalized character $f^+_d(M)$ without a pole in the dimension of choice.
One also obtains a counterterm $f^-_d(M)$, which is subject to renormalization group exponentiation \cite{CK:RH2}. The Hepp bound may then be defined as the corresponding contribution to the beta function. These analogies lead to a variant of perturbative quantum field theory where all Feynman integrals can be computed in rational terms. This theory will be explored in detail elsewhere.

The Hepp--period correlation may provide a new route to numerical approximations of more general Feynman integrals with kinematic dependence.

\bibliography{qft}

\end{document}